\documentclass[journal]{IEEEtran}  

\IEEEoverridecommandlockouts                        

\usepackage{amsmath}
\usepackage{amsthm}
\usepackage{amsfonts}
\usepackage{graphicx}
\usepackage{multirow}
\usepackage{algorithm}
\usepackage{caption}
\usepackage{subcaption}
\usepackage[noend]{algpseudocode}
\usepackage{hyperref}
\usepackage{cite}
\usepackage{color}

\newcommand{\norm}[1]{\left\lVert#1\right\rVert}

\newtheorem{proposition}{Proposition}
\newtheorem{definition}{Definition}

\newtheorem{lemma}{Lemma}
\newtheorem{assumption}{Assumption}
\newtheorem{remark}{Remark}

\DeclareMathOperator*{\argmax}{arg\,max}
\DeclareMathOperator*{\argmin}{arg\,min}

\title{\LARGE \bf
Minimax Multi-Agent Persistent Monitoring of a Network System
}

\author{Samuel C. Pinto$^1$, Shirantha Welikala$^{2}$, Sean B. Andersson$^{1,2}$, Julien M. Hendrickx$^3$, and Christos G. Cassandras$^{2,4}$
\\
$^1$Dept. of Mechanical Engineering, $^1$Division of Systems Engineering,\\ $^4$Dept. of Electrical and Computer Engineering \\
Boston University, Boston, MA 02215, USA \\
$^3$ICTEAM Institute, UCLouvain, Louvain-la-Neuve 1348, Belgium \\
\{samcerq,shiran27,sanderss,cgc\}@bu.edu, julien.hendrickx@uclouvain.be

\thanks{This work was supported in part by NSF under grants ECCS-1931600, DMS-1664644, CNS-1645681, and CMMI-1562031, by ARPA-E's NEXTCAR program under grant DE-AR0000796, by AFOSR under grant FA9550-19-1-0158,  and by the MathWorks. The work of J. Hendrickx was supported by the “RevealFlight” Concerted Research Action (ARC) of the Federation Wallonie-Bruxelles, by the Incentive Grant for Scientific Research (MIS) “Learning from Pairwise Data” of the F.R.S.-FNRS.
}
}

\begin{document}

\maketitle
\thispagestyle{empty}
\pagestyle{empty}

\begin{abstract}
We investigate the problem of optimally observing a finite set of {targets} using a fleet of mobile {agents} over an infinite time horizon. First, the single-agent case is considered where the agent is tasked to move in a {network}-constrained structure to gather information so as to minimize the worst-case uncertainty about the internal states of the targets. Hence the agent has to decide its sequence of target-visits and the corresponding sequence of dwell-times to be spent collecting information on each visited target. For a given visiting sequence, we prove that in an optimal dwelling time allocation, the peak uncertainty is the same among all the targets. This allows us to formulate the negotiation of dwelling times as a resource allocation problem and solved using a novel efficient algorithm. Next, determining the optimal visiting sequence is formulated as a vehicle routing problem and solved using a greedy algorithm. Then, we extend these ideas to the multi-agent case by {clustering} the network of targets using {spectral clustering} and assigning the agents to these target clusters. In this clustering process, a novel {metric} that can efficiently evaluate a given visiting sequence is extensively exploited. Based on the same metric, a distributed {cluster refinement} algorithm is also proposed. Finally, numerical results are included to illustrate the contributions.

\end{abstract}

\section{Introduction}
\label{sec:introduction}
Continuously surveying a set of points ({targets}) in a {mission space} using mobile {agents} is a problem that has applications in many domains such as in ecological monitoring \cite{lin2018kalman}, infrastructure safety verification \cite{ostertag2019robust}, ocean temperature surveillance \cite{lan2016rapidly}, deep-sea exploration \cite{Alam:2018ie} and multiple particle tracking \cite{pinto2020multipleparticletracking}. In all these scenarios, the targets of interest have uncertain  {states} that evolve over time and the goal of the agent team is to move around the mission space collecting data (through a stochastic observation model) so as to minimize a metric of the overall {estimation uncertainty} associated with the target states over very long periods of time. The agents need to move around the environment in order to visit different targets periodically to keep the target state uncertainties low as time goes to infinity. This problem is commonly known as Persistent Monitoring (PM). 


Many authors have studied PM from different points of view. Earlier works \cite{grocholsky2003information,julian2012distributed} have define the problem focusing on the interplay of estimation and trajectory planning. These works assumed very general observation models and agent dynamics, but their approach to solving the problem can only handle finite time horizons. Other works \cite{lan2014variational,hussein2008kalman} have framed the trajectory planning part of the problem as an optimal control problem. However their solution requires solving computationally intractable two-point boundary value problems. In \cite{lan2016rapidly}, a variation of the well-known rapidly-exploring random tree algorithm \cite{Lavalle1998} was used to optimally design infinite-horizon agent trajectories. However, our earlier work \cite{pinto2020sdppm} has shown how the computational complexity of this approach hinders its applicability. More recent work \cite{chen2020multi} has focused on jointly designing trajectories and learning the target uncertainty evolution using reinforcement learning. However, this solution needs to be trained specifically for the mission space of interest, and thus, is computationally expensive and lacks generality. 


We have also previously studied PM under different settings. In \cite{cassandras2013optimal}, each target state uncertainty is assumed to be a non-negative metric that decreases (or increases) linearly with time when the target is observed (or not observed). It then shows that the optimal trajectory can be parametrized and uses an efficient gradient-based parametric control technique \cite{cassandras2010perturbation} to design the agent trajectories. This idea was extended in \cite{pinto2020periodicfull} considering more natural target and agent models where each target state was assumed to follow linear stochastic dynamics and each agent is assumed to have the capability to observe a target state according to a linear stochastic observation model. Then, to approach the infinite horizon problem, agents are constrained to follow periodic trajectories. In this case, the corresponding gradient-based parametric control solution relied on the solution of $N\times M\times P$ matrix differential equations obtained at each gradient step, where $N$ is the number of targets, $M$ the number of agents and $P$ the number of parameters used to describe the trajectory of an agent. Therefore, solving these matrix differential equations was a major computational burden for settings with many agents and targets. By contrast, the approach in the current paper does not require the solution of any matrix differential equation. 


In this work, our main contribution is to overcome such computational limitations by using simple, lightweight algorithms that can scale to networks with large numbers of targets and agents.
We then make the fundamental assumption that each target is only monitored by one agent (also known as the ``no target sharing'' assumption in \cite{Welikala2020P7,Welikala2019P3}) - motivated by the empirical results reported in \cite{pinto2020periodicfull,Yu2016} where optimal agent trajectories do not share targets. As a result of this assumption, we can decompose a multi-agent problem into a set of single-agent ones. Further, unlike previous work, we assume that the PM goal is to minimize the worst-case (instead of the average) uncertainty among all the target states. This choice of the PM goal not only leads to a considerable reduction in the computational burden, but is also an appropriate choice in many PM applications. For example, when monitoring safety-critical systems that cannot operate over a given threshold (for instance, a maximum temperature), one wants to optimize the ``worst-case'' performance (as opposed to optimizing an ``average'' chance of violating it). Some examples of applications where a critical threshold on the state uncertainty should not be exceeded include monitoring wildfire or faults in civil infrastructure systems using unmanned aerial vehicles \cite{lin2018kalman, shakhatreh2019unmanned}.


Here
we initially study the single-agent problem. This simplification allows us to prove that, for a fixed sequence of target visits (\emph{visiting sequence}), the worst-case uncertainty is {\it the same for all targets} when the agent uses the corresponding optimal sequence of dwell-times (\emph{dwelling sequence}). This property alone is sufficient to determine the optimal dwelling sequence - if each target in the considered visiting sequence is visited only once during a single \emph{cycle}. In particular, this problem of determining the optimal dwelling sequence can be seen as a resource allocation problem where each target competes for the agent's dwell-time at that target. We next prove that a simple feedback law can be used to determine this optimal dwelling sequence efficiently. This same notion is then extended to the case where each target in the considered visiting sequence is allowed to be visited multiple times during a cycle. Afterward, we explore the problem of determining the optimal visiting sequence. We show that if each target is {visited at most once} during a cycle, a high-performing sub-optimal visiting sequence can be found by solving a Traveling Salesman Problem (TSP) and then executing a sequence of greedy Cycle Modification Operations (CMOs) on the obtained TSP solution. 


Our final contribution is to extend these notions to address the multi-agent scenario. The key idea is to \emph{cluster} the network of targets and assign the agents to these target clusters. To obtain these target clusters, the standard spectral clustering algorithm \cite{Luxburg2007} is used together with a novel \emph{similarity metric} that reflects the cost of the persistent monitoring task when two targets belong to the same cluster. Moreover, we propose a distributed target cluster refinement algorithm that systematically trades targets between clusters as long as it improves the overall cost of the problem. 


This paper is organized as follows. Section \ref{sec:formulation} gives the problem formulation and \ref{Sec:PropertiesOfAnOptimalPolicy} discusses general properties of an optimal policy. Section \ref{sec:optimal_dwelling_single_visit} gives an algorithm to compute dwelling times, assuming some constraints in the target trajectory. Then, Section \ref{sec:multiple_visits} extends the process of computing dwelling times by relaxing these trajectory constraints. Section \ref{Sec:GreedyExplorations} discusses the process of finding a high-performing sub-optimal visiting sequence. The proposed single-agent solution so far is extended in Section \ref{sec:multi_agent_extension} to address multi-agent scenarios. Simulation results are presented at the end of each of the Sections \ref{sec:multiple_visits}, \ref{Sec:GreedyExplorations} and \ref{sec:multi_agent_extension}. Finally, Section \ref{sec:conclusion} concludes the paper. 



\section{Problem Formulation}
\label{sec:formulation}
We consider an undirected graph $\mathcal{G}=(\mathcal{V},\mathcal{E})$ where the set $\mathcal{V}=\{1,2,.\ldots,M\}$ represents $M$ nodes (targets) and the set $\mathcal{E}= \{(i,j):\,i,j\in \mathcal{V}\}$ represents all the edges (available for agents to travel between targets). Each edge $(i,j)\in\mathcal{E}$ has an associated value $d_{i,j}$ that represents the travel-time an agent has to spend in order to travel from target $i$ to target $j$. As stated earlier, Sections \ref{sec:formulation}-\ref{Sec:GreedyExplorations} of this paper assume the single-agent case (i.e., only one agent is deployed). 

Each target $i\in\mathcal{V}$ has an internal state $\phi_i\in \mathbb{R}^{L_i}$ that follows the dynamics
\begin{equation}
    \label{eq:dynamics_phi}
    {{\dot{\phi}}_i(t) = A_i{\phi}_i(t) + B_iu_i(t)+ {w}_i(t)},
\end{equation}
where $\{w_i\}_{i\in\mathcal{V}}$ are mutually independent, zero mean, white, Gaussian distributed processes with $E[{w}_i(t){w}_i'(t)]=Q_i$ and $Q_i$ is a positive definite matrix. We assume that each target $i$ selects (or is aware of) the current value of the control $u_i(t)$ at all times. These control signals can be designed either in a open or closed loop fashion and our analysis is independent such target controls, since they naturally do not affect the covariance dynamics.

When the agent visits a target $i\in\mathcal{V}$, it observes the target's internal state $\phi_i$ according to a linear observation model 
\begin{equation}
    \label{eq:observation_model_ij}
    {z}_{i}(t)=H_i{\phi}_i(t)+{v}_{i}(t), \qquad {v}_{i}(t)\sim \mathcal{N}(0,R_i),
\end{equation}
where $z\in\mathbb{R}^{L_i}$ and the noise $v_i$ is also assumed to be white and statistically independent of $v_j$ if $i\neq j$ and $w_k$, $\forall k$.

We can fully describe the behavior (trajectory) of the agent by its sequence of target visits (visiting sequence): $\Xi=[y_1,y_2,...,y_N]$ where each $y_k\in\mathcal{V}$ and the corresponding sequence of dwell-times (dwelling sequence)  spent at each visited target: $\mathcal{T} = [t_1,...,t_N]$ where each $t_k\geq0$. 

Considering models \eqref{eq:dynamics_phi} and \eqref{eq:observation_model_ij}, the maximum likelihood estimator $\hat{\phi}_i(t)$ of the internal target state $\phi_i(t)$ is a Kalman-Bucy Filter (evaluated at target $i$) given by
$$\dot{\hat{{\phi}}}_i(t) = A_i\hat{{\phi}}_i(t)+{B_iu_i(t)}+\eta_i(t)\Omega_i(t){H}_i'{R}_i^{-1}({{z}}_i(t)-{H}_i\hat{{\phi}}_i(t)),$$
where $\Omega_i(t)$ is the covariance matrix of the estimation error ($\Omega_i(t) = E(e_i(t)e_i'(t))$ with $e_i(t) = \phi_i(t)-\hat{\phi}_i(t)$) given by 
\begin{equation}
    \dot{\Omega}_i(t) = A_i\Omega_i(t)+\Omega_i(t)A_i'+Q_i-\eta_i(t)\Omega_i(t){G}_i\Omega_i(t), \label{eq:dynamics_omega} 
\end{equation}
with $G_i=H_i'R_i^{-1}H_i$ and $\eta_i(t)=1$ if target $i$ is observed at time $t$ and $\eta_i(t)=0$ otherwise.

The goal is to design an agent trajectory (i.e., to design $\Xi$ and $\mathcal{T}$) that minimizes the persistent monitoring objective 
\begin{equation}
    \label{eq:cost}
    J(\Xi,\mathcal{T}) = \max_{i\in \mathcal{V}} \limsup_{t\rightarrow \infty} g_i(\norm{\Omega_i(t)}), 
\end{equation}
which represents the maximum over time and over all the targets of a weighted norm of the long term covariance matrix. In \eqref{eq:cost}, the target-specific possibly non-linear weighting function $g_i(\cdot)$ is strictly increasing with $g_i(0)=0$ and $\lim_{x \rightarrow \infty}g_i(x)=\infty$ and $\norm{\cdot}$ is a norm on the space of positive semi-definite matrices. A usual choice is to have $g_i(x)=\alpha_i x$, where $\alpha_i$ is a constant target-specific weight, and $\norm{X}=\text{tr}(X)$. If an optimal agent trajectory $(\Xi^*,\mathcal{T}^*)$ exists, then we denote its associated cost 
$J^*(\Xi^*,\mathcal{T}^*) = \min_{\Xi,\mathcal{T}} J(\Xi,\mathcal{T})$.

\subsection{Periodic Policies}
\label{subsec:periodic_policies}
In this work, we restrict ourselves to periodic agent trajectories. Note that the number of parameters necessary to describe such a periodic trajectory does not increase with the time horizon and earlier work \cite{zhao2014optimal} in a similar setting showed that periodic trajectories can perform arbitrarily well compared to a true optimal trajectory (when it exists). Therefore, periodic agent trajectories are suitable for infinite horizon analysis in the context of offline trajectory planning. Moreover, if a target is visited in the period, it will be visited infinitely often, with inter-visit time upper bounded by the period. This notion fits very well with the paradigm of persistent monitoring since one of the high-level goals is to ensure such a persistence of visits.

In order to analyze the infinite horizon behavior of the covariance matrix, we first make the following natural assumption that ensures that the covariance of the internal state estimate $\hat{\phi}_i$ can be finite over long time horizons. 

\begin{assumption}
    The pair $(A_i,H_i)$ is observable, $\forall \in \mathcal{V}$.
\end{assumption}

We now restate the following Proposition given in \cite{pinto2019periodicfull}:
\begin{proposition}
    \label{prop:unique_attractive_sol_riccati_eq}
    If $\eta_i(t)$ is $T$-periodic and $\eta_i(t) > 0$ for some $[a,b]\in[0,T]$, then, under Assumption 1, there exists a unique non-negative stabilizing $T$-periodic solution $\bar{\Omega}_i(t)$ to \eqref{eq:dynamics_omega}.
\end{proposition}

A consequence of this proposition is that, in a periodic schedule, visiting a given target for any finite amount of time is enough to guarantee that the covariance will converge to a limit cycle that does not depend on the initial conditions. Moreover, this proposition also implies that the $\limsup$ in \eqref{eq:cost} can be replaced by the maximum over one period, i.e.
\begin{equation}
\limsup_{t\rightarrow \infty} g_i(\norm{\Omega_i(t)}) = \max_{t\in[0,T]}g_i(\norm{\bar{\Omega}_i(t)}).
\end{equation}

\section{Properties of an Optimal Policy}
\label{Sec:PropertiesOfAnOptimalPolicy}
\subsection{Target's Perspective of a Periodic Policy}
\label{subsec:simple_notation}

We begin by defining some notations used in the remainder of this section. Recall that the goal is to optimize the visiting sequence $\Xi$ and the corresponding dwelling sequence  $\mathcal{T}$. Also, recall that the length $N_\Xi$ of the vectors $\Xi$ and $\mathcal{T}$ correspond to one full cycle of the periodic trajectory. 

The vectors $\Xi$ and $\mathcal{T}$ captures the agent's perspective of its behavior. Although our goal is to plan the agent movement policy, much of the notation and discussion in this paper is simplified if we also consider the target's perspective. In order to fully determine each target's steady-state covariance behavior it is not necessary to fully know $\Xi$ and $\mathcal{T}$. From the target's perspective, its covariance matrix only depends on when this particular target was observed. This is illustrated in Fig. \ref{fig:multiple_observations_same_target}, where $t_{\text{on},i}^k$ is the amount of time that target $i$ was observed when it was visited for the $k$\textsuperscript{th} time within a cycle and $t_{\text{off},i}^k$ is the time spent between its $k$\textsuperscript{th} and $(k+1)$\textsuperscript{th} visits. The covariance values in the transition from not observed to observed are denoted by $\overline{P}_i^k$. Here, $k\in \{1,2,\ldots,N_i\}$, where $N_i$ is the total number of visits by the agent to target $i$ within a cycle.  All these quantities can be easily determined from the original agent-centric decision variables $\Xi$ and $\mathcal{T}$; more details on these computations are given in Eqs. \eqref{eq:targets_perspective_variables}-\eqref{eq:targets_perspective_variables_2} in Appendix \ref{app:proof_props_ton_toff}. We note that, since the steady state covariance is periodic, the initial time is arbitrary, and we consider (as illustrated in Fig. \ref{fig:multiple_observations_same_target}) that at $t=0$ the target had just switched from not being observed to being observed.


\begin{figure}[h!]
    \centering
    \includegraphics[width=0.9\columnwidth]{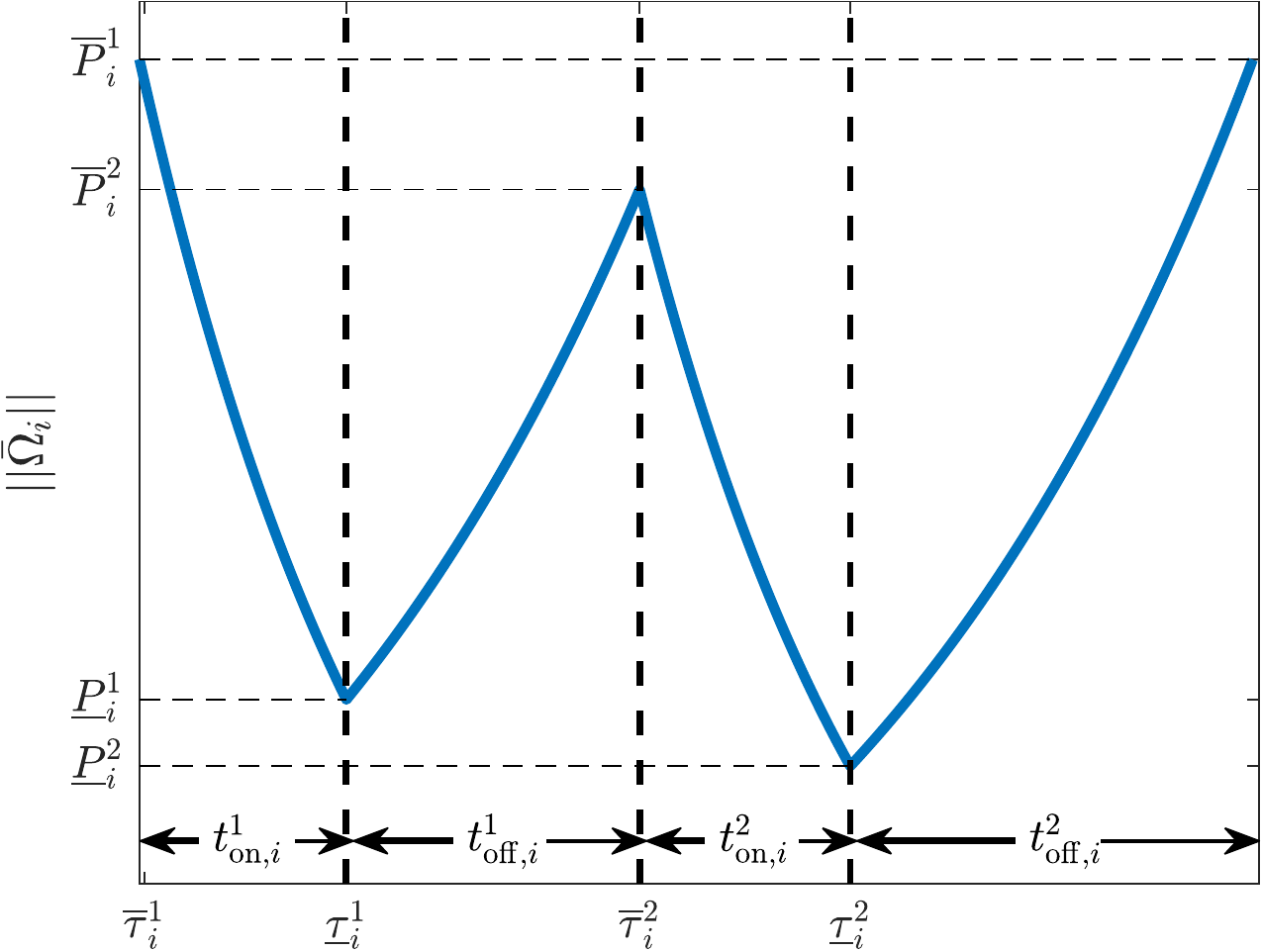}
    \caption{Temporal evolution of the steady state covariance matrix and waiting/observation times.}
    \label{fig:multiple_observations_same_target}
    \vspace{-3mm}
\end{figure}


\subsection{Necessary Condition for Optimality}

In this section, our main goal is to show that, for any visiting sequence $\Xi$ that contains every target in $\mathcal{V}$, the corresponding optimal dwelling sequence $\mathcal{T}$ must be such that $\limsup_{t\rightarrow\infty}g_i(\norm{\Omega_i(t)})$ is the same for all $i\in\mathcal{V}$.

Towards this main goal, we first introduce an auxiliary result that states that the steady-state covariance increases when the target is not observed ($\eta_i=0$) and decreases otherwise ($\eta_i=1$). However, we stress that this only holds at steady state. Observe, for example, that if the initial uncertainty over target $i$ is very small, then the transient uncertainty could temporarily grow even if $i$ is being observed. Similarly, the uncertainty along certain directions may initially decrease even when $i$ is not observed if it was very large at the initial time.

\begin{lemma}
    \label{lemma:pd_nd_time_derivative}
    If $t_{\text{off},i}^k>0$ and $t_{\text{on},i}^p>0$ for some $p,k$ such that $1\leq p,k \leq N_i$, then  $\dot{\bar{\Omega}}_i(t) \prec 0$ when the target is observed ($\eta_i(t)=1$) and $\dot{\bar{\Omega}}_i(t) \succ 0$ otherwise (i.e. when $\eta_i(t)=0$).
\end{lemma}

The intuitive explanation of the importance of this Lemma is that it gives an insight on how to optimize the dwelling times: when local maximum uncertainty peaks are different, it means that it is possible to observe for less time the target with lowest peaks (increasing its uncertainty) and for more time the targets with largest peak (thus decreasing it). This way, maximum uncertainty over all the targets would go down. Since its proof is long and not necessary for the comprehension of the main results in this paper, it will be given in Appendix \ref{app:proof_pd_nd_time_derivative}.


We now show that the peak uncertainties can only be achieved at very specific instants of time: those where the target switches from not being observed to being observed.  This is an important result, since it guarantees that we can compute the cost function by only looking at the a finite number of time instants. As a reminder, we define the covariance at the beginning of the $k$\textsuperscript{th} observation of that target as $\overline{P}_i^k$.

\begin{lemma}
    \label{lemma:max_on_peak}
    If target $i$ is visited for a a strictly positive amount of time, then 
    $\limsup_{t\rightarrow \infty} g_i(\norm{{\Omega}_i(t)})=\max_{1\leq k \leq N_i}g_i(\big{\Vert} \overline{P}_i^k \big{\Vert})$.
\end{lemma}

\begin{proof}
    First, note that since $\Omega_i(t)$ converges to the bounded periodic function $\bar{\Omega}_i(t)$, $\limsup_{t\rightarrow \infty}g_i\left(\norm{{\Omega}_i(t)}\right)=\max_{0\leq t\leq T}g_i\left(\norm{\bar{\Omega}_i(t)}\right)$.
    For any time $t$ which $\exists\  \epsilon>0$ such that $\eta_i(t+\epsilon)=0$, Lemma \ref{lemma:pd_nd_time_derivative} implies that $\bar{\Omega}_i(t+\epsilon)\succ \bar{\Omega}_i(t)$. 
    Conversely, if $\exists\  \epsilon>0$ $\eta_i(t-\epsilon)=1$, $\bar{\Omega}_i(t-\epsilon)\succ \bar{\Omega}_i(t)$.
    
    Therefore, the maximization of $g_i\left(\norm{\bar{\Omega}_i}\right)$ can only happen in one of the instants when the target switches from not being observed ($\eta_i=0$) to being observed ($\eta_i=1$). The covariance at these instants is given by $\overline{P}_i^k$ (see also Fig. \ref{fig:multiple_observations_same_target}).
\end{proof}

Using Lemmas \ref{lemma:pd_nd_time_derivative}, \ref{lemma:max_on_peak}, we next show how the peak values $\overline{P}_i^k$ vary with $t_{\text{on},i}^m$ and $t_{\text{off},i}^m$, for any $k,m$ where $1\leq k,m \leq N_i$. 
For this, we use the fact that the parameters $t_{\text{on},i}^m$ and $t_{\text{off},i}^m, \forall m$ fully define the evolution of steady state covariance matrix $\bar{\Omega}_i(t)$. The following proposition has an intuitive interpretation: when a target is observed for a longer time, its peak uncertainty will be lower. Conversely, if the time between observations increases, then the peak uncertainty will be higher.

\begin{proposition}
    \label{prop:derivative_ton_toff}
    $\frac{\partial \overline{P}_i^k}{\partial t_{\text{on},i}^m} \prec 0$ and $\frac{\partial \overline{P}_i^k}{\partial t_{\text{off},i}^m} \succ 0$.
\end{proposition}

The proof of Proposition \ref{prop:derivative_ton_toff} is given in Appendix \ref{app:proof_props_ton_toff}.

The following proposition, which is the main result of this section, can be interpreted analogously to resource allocation problems where different targets are competing for the same resource $t_{\text{on},i}^k$. Therefore, an equilibrium in the minimax sense is reached when all the ``active" targets have the same utility. Note that an active target is one that is visited for a non-null amount of time at least once during a cycle, in other words: 
\begin{definition}
A target $i$ is said to be {\bf active} if $\sum_{k=1}^{N_i} t_{\text{on},i}^k>0$.
\end{definition} 

Unlike typical resource allocation problems, here the total resource $\sum_{i=1}^M \sum_{k=1}^{N_i}t_{\text{on},i}^k$ is not fixed. The reason why the total resource does not go to infinity (i.e. the period is finite) is that increasing $t_{\text{on},i}^k$ to one target has an adverse effect to all other active targets. 

Let us define the set of all active targets as $\mathcal{A} \subseteq \mathcal{V}$. Note that a target with unstable internal state dynamics (i.e., $A_i$ is unstable) has to be active; otherwise, the cost \eqref{eq:cost} will be unbounded. However, if a target has stable $A_i$, its infinite horizon uncertainty without ever being observed can be lower than some other target's uncertainty, and thus the optimal policy would be to make such targets inactive.

\begin{proposition}
    \label{prop:same_peak}
    For a fixed visiting sequence $\Xi$, the corresponding dwelling time sequence $\mathcal{T}$ that optimizes the cost \eqref{eq:cost} satisfies:
    $$\limsup_{t\rightarrow\infty}g_i(\norm{\Omega_i(t)})=\limsup_{t\rightarrow\infty}g_j(\norm{\Omega_j(t)}),$$
    for all $i, j\in\mathcal{A}$. Additionally, if $i\in\mathcal{A}$ and $p\not\in\mathcal{A}$, then 
    $$\limsup_{t\rightarrow\infty}g_i(\norm{\Omega_i(t)})\geq\limsup_{t\rightarrow\infty}g_p(\norm{\Omega_p(t)}).$$
\end{proposition}
\begin{proof}
    First we focus on active targets and prove the first part of the proposition by contradiction, showing that if the property given in the proposition does not hold, then there is a way to re-balance the observation times that is guaranteed to improve the performance. Suppose that for some target $i$
    \begin{equation}
        \label{eq:inequal_peak_hypothesis}
        g_i\left(\norm{\overline{P}_i^{\max} (t_{\text{on},i}^{1:N_i},t_{\text{off},i}^{1:N_i})}\right)<g_j\left(\norm{\overline{P}_j^{\max} (t_{\text{on},j}^{1:N_j},t_{\text{off},j}^{1:N_j})}\right),
    \end{equation}
    where the upper index $\max$ indicates that among all the peaks $P_i^k$, we pick the value of $k$ that yields the highest $g_i(\norm{P_i^k})$. We now propose to decrease the amount of time $i$ is observed by $\epsilon$. According to \eqref{eq:def_toff}, this implies that the waiting time between observations for all the other targets will decrease. This updated policy (dwelling sequence) generates a new set of observation times for target $i$, (denoted by $\tilde{t}_{\text{on},i}^{1:N_i}$), and updated waiting times between visits for all the other active targets, (denoted by $\tilde{t}_{\text{off},j}^{1:N_j}$), while maintaining ${t}_{\text{off},i}^{1:N_i}$ and ${t}_{\text{on},j}^{1:N_j}$ constant. Note that $\exists \epsilon>0$ such that $\tilde{t}_{\text{on},i}^ k={t}_{\text{on},i}^k-\epsilon$ for some $k\in\{1,..,N_i\}$ and $\tilde{t}_{\text{off},j}^m<{t}_{\text{off},j}^m$ for some $m\in \{1,...,N_j\}$. Using Proposition \ref{prop:derivative_ton_toff}, we get:
    \begin{align}
        \overline{P}_i^{\max} (\tilde{t}_{\text{on},i}^{1:N_i},t_{\text{off},i}^{1:N_i}) &\succ  \overline{P}_i^{\max} ({t}_{\text{on},i}^{1:N_i},t_{\text{off},i}^{1:N_i}),\\
        \overline{P}_j^{\max} ({t}_{\text{on},j}^{1:N_j},\tilde{t}_{\text{off},j}^{1:N_j}) &\prec  \overline{P}_j^{\max} ({t}_{\text{on},j}^{1:N_j},t_{\text{off},j}^{1:N_j}).
    \end{align}
     Using the fact that both the norm and the derivative are continuous and strictly increasing, we can always pick an $\epsilon$ small enough such that the new peak of target $i$ is lower or equal to the new peak of $j$, i.e.,
    \begin{align*}
        g_j\left(\norm{\overline{P}_j^{\max} ({t}_{\text{on},j}^{1:N_j},\tilde{t}_{\text{off},j}^{1:N_j})}\right) &<  g_j\left(\norm{\overline{P}_j^{\max} ({t}_{\text{on},j}^{1:N_j},t_{\text{off},j}^{1:N_j})}\right),\\
        g_i\left(\norm{\overline{P}_i^{\max} (\tilde{t}_{\text{on},i}^{1:N_i},t_{\text{off},i}^{1:N_i})}\right)&\leq g_j\left(\norm{\overline{P}_j^{\max} ({t}_{\text{on},j}^{1:N_j},\tilde{t}_{\text{off},j}^{1:N_j})}\right).
    \end{align*}
    
    Since under the updated policy, all the peaks $\overline{P}_j^m$, $1\leq m \leq N_j$, $1\leq j\leq M$, are lower for all the targets except $i$, we recall Lemma \ref{lemma:max_on_peak} and conclude that this updated policy has a lower cost than the previous one. Hence, the previous policy cannot be optimal, which proves the first part of the proposition.
    
    For the second part of the proposition, we note that if a target is inactive, its peak uncertainty could be reduced by increasing its dwell-time (and thus increasing the peak uncertainty of the active targets), using a very similar argument as in the first part of this proof. Thus, in an optimal dwelling sequence, a target will be inactive only if its peak uncertainty is lower (or equal) than that of the active targets.
    \end{proof}

This proposition gives a necessary condition for the optimality of the dwelling sequence. Moreover, its constructive proof also gives insight on how to locally optimize the dwelling sequence for a given visiting sequence. However, in general, this property is not sufficient for determining the optimal dwelling sequence. In the next section, we will restrict ourselves to a specific set of visiting sequences $\Xi$ where the optimality condition in Proposition \ref{prop:same_peak} can indeed be exploited to optimize the dwelling times at each target.


\section{Optimal Dwelling Sequence on a Constrained Visiting Sequence}
\label{sec:optimal_dwelling_single_visit}

While in the previous section we discussed a necessary condition for optimal dwelling time allocation, this condition alone is not sufficient to fully determine a globally optimal trajectory. Thus, in this section we restrict ourselves to a specific type of visiting sequence that allows us to further exploit this necessary condition to design a practical algorithm for dwelling sequence allocation. We then extend that algorithm to more general scenarios in the next section. More specifically, here we constrain ourselves to consider only visiting sequences $\Xi$ where {\it each target is visited at most once during a cycle} (hence also called ``constrained'' visiting sequences). Under this assumption, we develop a practical algorithm that optimizes the dwelling sequence corresponding to a given visiting sequence $\Xi$. In some specific setups, we show that this approach gives the globally optimal dwelling sequence. Later on in this section, we discuss the process to obtain an optimal visiting sequence $\Xi^*$. As a side note, in this section we will omit the upper index of $\overline{P}_i, t_{\text{on},i}$ and $t_{\text{off},i}$, since $N_i=1, \forall i\in\Xi$. Note also that $i\in \Xi$ if and only if $i\in\mathcal{A}$.

The main idea behind our approach is to exploit the property that all peak uncertainties ($\overline{P}_i$) of active targets must coincide in an optimal solution. In particular, we develop an iterative scheme that balances the dwell-times ($t_{\text{on},i}$) such that the peak uncertainties coincide upon convergence of the algorithm. The update law used to update the dwell-times $t_{\text{on},i}$ in this iterative scheme is:
    \begin{equation}
        \label{eq:consensus_update_discrete}
        t_{on,i}[k+1] = t_{on,i}[k] +k_p\log\left(\frac{g_i\left(\norm{\overline{P}_i}\right)}{g_{\text{avg}}}\right),
    \end{equation}
where $g_{\text{avg}}=\left(\prod_{j\in\mathcal{A}}g_j\left(\norm{\overline{P}_j}\right)\right)^{\frac{1}{M}}$ and $k_p$ is a small positive constant. It can be interpreted as an algorithm aiming to achieve ``consensus" regarding the peak uncertainties $\overline{P}_i$ among the active targets, and thus its structure is very similar to geometric mean consensus algorithms \cite{FB-LNS}. The expression \eqref{eq:consensus_update_discrete} does not require the computation of gradients, which makes it computationally much less demanding than gradient-based approaches, such as \cite{pinto2020periodicfull}. We highlight that the computation of gradients involves the numerical solution of matrix differential equations at each step of the optimization, which imposes a major computational burden.

\begin{remark}
    At each iteration of \eqref{eq:consensus_update_discrete}, it is necessary to compute $\overline{P}_i$ and $g_{\text{avg}}$. For computing $\overline{P}_i$, we use the algorithm described in \cite{chu2004structure}, which converges quadratically and is numerically stable. On the other hand, we note that $g_{\text{avg}}$ can be computed distributively through a consensus protocol if $\overline{P}_i$ is computed locally by each target.
\end{remark}
    
In order to simplify the convergence analysis of this update law, we abstract it with the following differential equation:
    \begin{equation}
        \label{eq:update_law_single_target}
        \frac{d}{dr}{t}_{\text{on},i} = \begin{cases}
        0,&\text{if }i \not\in \mathcal{A} \text{ and } \frac{g_i(\norm{\bar{P}_i})}{g_{\text{avg}}}\leq 1,\\
        k_p\log\left(\frac{g_i\left(\norm{\overline{P}_i}\right)}{g_{\text{avg}}}\right),& \text{otherwise}.
        \end{cases}
    \end{equation}
    Note that this version of the update law considers continuous parameter variation, i.e., the auxiliary variable $r$ should be understood as the continuous time equivalent of ``iteration index" and does not carry any ``time'' interpretation \textcolor{blue}{and at each instant $P_i$ is a function of $t_{\text{on}},i$ and $t_{\text{off},i}$}. 
    
    \begin{proposition}
        \label{prop:assymptotic_stability}
        Under the update law \eqref{eq:update_law_single_target}, the function $\max_{i,j\in\mathcal{A}}|g_i(\norm{\overline{P}_i})-g_j(\norm{\overline{P}_j})|$ is asymptotically stable.
    \end{proposition}
    \begin{proof}
    Note that 
    $
            \frac{d \overline{P}_i}{d r} = \frac{\partial \overline{P}_i}{\partial t_{\text{on},i}}\frac{dt_{\text{on},i}}{dr}+\frac{\partial \overline{P}_i}{\partial t_{\text{off},i}}\frac{dt_{\text{off},i}}{dr}
    $   
     and 
        \begin{align}    
        \sum_{i=1}^M \frac{d{t}_{\text{on},i}}{dr} &=k_p\log \big(\frac{\prod_{j\in \mathcal{A}}g_j\left(\norm{\overline{P}_j}\right)}{g_{\text{avg}}}\big) \nonumber = k_pM\log \frac{g_{\text{avg}}}{g_{\text{avg}}}=0.
        \end{align}
         Therefore, under this control law, the period is constant. Thus, $\frac{dt_{\text{on},i}}{dr}=-\frac{dt_{\text{off},i}}{dr}.$ Hence,
         $
         \frac{d \overline{P}_i}{d q} = (\frac{\partial \overline{P}_i}{\partial t_{\text{on},i}}-\frac{\partial \overline{P}_i}{\partial t_{\text{off},i}})k_p\log \frac{g_i(\norm{\overline{P}_i})}{g_{\text{avg}}}
         $.
         Consequently, if $g_i(\norm{\overline{P}_i})>g_{\text{avg}}$, $\frac{d \overline{P}_i}{d r} \prec 0.$ 
        Conversely, if $g_i(\norm{\overline{P}_i})<g_{\text{avg}}$, then $\frac{d \overline{P}_i}{d r} \succ 0$. 
        Note also that
        \begin{equation*}
            \max_{i,j\in \mathcal{A}}|g_i(\norm{\overline{P}_i})-g_j(\norm{\overline{P}_j})| = \max_{i\in \mathcal{A}}g_i(\norm{\overline{P}_i}) - \min_{j\in \mathcal{A}}g_j(\norm{\overline{P}_j}).
        \end{equation*}
        Since $g_{\text{avg}}$ is the geometric mean (i.e., its value is lower than the maximum and higher than the minimum $g_j(\norm{\overline{P}_j})$), we get that if $\max_{i\in\mathcal{A}}g_i(\norm{\overline{P}_i})\neq g_{\text{avg}}$, then
        $\frac{d}{dr}\max_{i\in \mathcal{A}}g_i(\norm{\overline{P}_i})<0$
        and 
        $\frac{d}{dr}\min_{j\in\mathcal{A}}g_j(\norm{\overline{P}_j})>0$. 
        Therefore, 
        $
        \frac{d}{dr}\max_{i,j}|g_i(\norm{\overline{P}_i})-g_j(\norm{\overline{P}_j})|<0
        $. This completes the proof.
    \end{proof}

    \begin{remark}
        \label{remark:monoticity}
        In the proof of Proposition \ref{prop:assymptotic_stability}, we see that $\frac{d}{dr}\max_ig_i(\norm{\overline{P}_i})<0$. Therefore, the update law \eqref{eq:update_law_single_target} always reduces the cost defined in \eqref{eq:cost}. This also implies that if not all the targets have the same peak value $g_i\left(\norm{\overline{P}_i}\right)$, then the cost can be reduced by that update law. 
    \end{remark}
    \begin{remark}
        The $\log$ in \eqref{eq:update_law_single_target} is only one among many options of the update law. The essential feature is that \eqref{eq:update_law_single_target} preserves the total dwelling time among different targets, and Proposition \ref{prop:assymptotic_stability} would still hold if we used other update laws that preserved this property. One of the possible update laws is $\frac{d}{dr}{t}_{\text{on},i}=g_i(\norm{\overline{P}_i}/\sum_{j}g_i(\norm{\overline{P}_j})$ when $i\in{A}$, however the reason why we use the $\log$ structure instead of this is that in our simulations we observed that the $\log$ structure usually led to faster convergence.
    \end{remark}

    Now we show that the value achieved by update law \eqref{eq:update_law_single_target} is unique, i.e. it does not depend on the observation time distribution at $r=0$. 
    \begin{lemma}
        \label{lemma:unique_consensus}
        For a given cycle period T and a fixed visiting sequence $\Xi$, there is a unique observation times distribution such that  $g_i(\norm{\overline{P}_i})=g_j(\norm{\overline{P}_j})$, $\forall \ i,j \in \Xi.$
    \end{lemma}
    \begin{proof}
        Note that $g_i(\norm{\overline{P}_i})$ is a function only of $t_{\text{on},i}$, since the period $T$ is fixed (i.e. $t_{\text{off},i}=T-t_{\text{on},i}$). Additionally, Prop. \ref{prop:derivative_ton_toff} says that $g_i(\norm{\overline{P}_i})$ is strictly decreasing with $t_{\text{on},i}$.
        
        Suppose there are two different sets of dwelling times ($t_{\text{on},i}$ and $t_{\text{on},i}'$), and consequently different costs $g_{con}$ and $g_{con}'$ such that all targets have the same peak value. Without loss of generality, we assume $g_{con}<g_{con}'$, which implies that $t_{\text{on},i}>t_{\text{on},i}'$. However, since the period is the same, we must have $\sum_i t_{\text{on},i}=\sum_it_{\text{on},i}'$, which yields a contradiction.
    \end{proof}
    
    Finally, we give a specialization of Proposition \ref{prop:same_peak} to the particular case discussed in this section. 
    \begin{proposition}
        For a fixed visiting sequence $\Xi$ and a given cycle period $T$, the dwelling sequence under the update law \eqref{eq:update_law_single_target} converges to the optimal dwelling sequence (i.e., to $\mathcal{T}^*$) that minimizes the cost function $J(\Xi,\mathcal{T})$ in \eqref{eq:cost}.
    \end{proposition}

    \begin{proof}
        For any dwelling sequence such that the period is $T$, the update law \eqref{eq:update_law_single_target} always reduces the cost while maintaining $T$ constant if $g_i(\norm{\overline{P}_i})$ is not the same for all the targets. Since there is a unique way such that every target has the same peak (and the update law \eqref{eq:update_law_single_target} ensures convergence to it), the dwelling sequence after convergence of \eqref{eq:update_law_single_target} has to be optimal, otherwise \eqref{eq:update_law_single_target} would be able to improve the cost.
    \end{proof}
    
    
    For a fixed visiting sequence $\Xi$, we so far have shown a simple way to optimize the dwelling sequence given a fixed cycle period $T$. However, we have not addressed the problem optimizing the value of the cycle period $T$. For this task, we use golden ratio search \cite{kiefer1953sequential} that finds the global optimum in a unimodal function and local optima in a generic single variable function. Moreover, in Appendix \ref{app:unimodality_proof} we show that when the target internal state $\phi_i$ is a scalar, this function is indeed unimodal. The complete golden ratio search procedure is given in Alg. \ref{alg:compute_optimal_visitation_times}. The function $g_{con}(T)$ corresponds to running the update law in \eqref{eq:consensus_update_discrete} until convergence and the value $g_{con}(T)$ is the value of $g_{\text{avg}}$ at the final iteration and $r=(1+\sqrt{5})/2$ (i.e., the search intervals are divided according to the golden ratio).
    \begin{algorithm}
\caption{Search for the Optimal Cycle Period}
\label{alg:compute_optimal_visitation_times}
    \begin{algorithmic}[1]
    \State{\bf Input: $T_{\min}$, $T_{\max}$}.
    \State{$T_1 \leftarrow T_{\max} - (T_{\max}-T_{\min})/r$}
    \State{$T_2 \leftarrow T_{\min} + (T_{\max}-T_{\min})/r$}
    \While{$|g_{con}(T_2 )-g_{con}(T_1 )|<\epsilon$}
    \If{$f(T_2)>f(T_1)$}
    \State{$T_{\max}\leftarrow T_2$}
    \Else{$\ T_{\min}\leftarrow T_1$}
    \EndIf
    \State{$T_1 \leftarrow T_{\max} - (T_{\max}-T_{\min})/r$}
    \State{$T_2 \leftarrow T_{\min} + (T_{\max}-T_{\min})/r$}
    \EndWhile
    \State{\bf Return: $(T_1+T_2)/2$}
\end{algorithmic}
\end{algorithm}

    \paragraph*{\textbf{Optimal Visiting Sequence}}
    
    We now focus on determining the optimal visiting sequence (i.e., $\Xi^*$). In particular, we will show that the optimal visiting sequence is the minimum length tour that visits every target of interest (active) in the network $\mathcal{G}=(\mathcal{V},\mathcal{E})$.

    \begin{proposition}
        \label{prop:optimal_cycle_tsp}
        Among all the constrained visiting sequences where all targets are visited, for any dwelling sequence $\mathcal{T}$, the visiting sequence given by the TSP solution ($\Xi=\Xi_{TSP}$) has the lowest cost $J(\Xi,\mathcal{T})$\eqref{eq:cost}.
    \end{proposition}
    
    \begin{proof}
        Recall that $t_{\text{off},i}$ is the sum of all the entries in the dwelling sequence $\mathcal{T}$ (omitting $t_{\text{on},i}$) and the sequence of travel-times corresponding to the visiting sequence $\Xi$. Since $\Xi$ is a constrained visiting sequence, every target is visited only once. Hence, the visiting sequence given by the TSP solution $\Xi_{TSP}$ is the one with the least amount of total travel-time and, according to Prop. \ref{prop:derivative_ton_toff}, $g_i(\norm{\overline{P}_i})$ strictly increases with $t_{\text{off},i}$. Since for any dwelling sequence $\mathcal{T}$, the values of $t_{\text{off},i},\ \forall i$, are minimized if $\Xi = \Xi_{TSP}$ we conclude that the visiting sequence $\Xi = \Xi_{TSP}$ yields the lowest cost.
    \end{proof}

     This proposition gives a good insight into the optimal visiting sequence. However, we highlight that our procedure for optimizing the dwelling sequence does not require the agent to physically visit each target only once during a cycle; it only requires the agent to spend a non-zero dwell-time at most in one of the visits to each target. This implies that we can use the algorithm developed in this section for visiting sequences where each target may be visited multiple times during a cycle, but the agent will dwell there on only one of those visits. This is an important distinction because the TSP cycle may be longer than the shortest path that visits all the targets while also allowing any target to be visited multiple times. Additionally, some graphs do not have a closed path that visits every target only once, the simplest example of this being a graph with three nodes and two edges, one that connects nodes 1 and 2 and the other connecting nodes 2 and 3. Therefore, in every closed cycle, node 2 has to be visited at least twice. In order to overcome this issue, a common approach is to preprocess the graph by replacing the edge between two nodes by the shortest path between them.
    
    \begin{remark}
        The problem of computing the optimal TSP cycle is NP-hard. However, efficient sub-optimal solutions are available (see e.g. \cite{TANG2000267}). The approach we discussed for optimizing the dwelling sequence does not rely on having the optimal TSP cycle and indeed can handle any cycle as long as every target is sensed (with a non-zero dwell-time) exactly once. Therefore, if finding the TSP cycle is computationally infeasible, we still can use a sub-optimal cycle (visiting sequence).
    \end{remark}
    
However, note that in an optimal visiting sequence, inactive targets do not necessarily have to be a part of the agent's tour (and, even if they are, the agent will not have to dwell on them). Thus, when searching for the optimal visiting sequence, we can exclude the agent to move through inactive targets, as long as the cost is benefited from this exclusion.

We now describe a procedure to obtain this optimal visiting sequence. Intuitively, this procedure starts by finding the optimal cycle that visits every target in the network and computing the corresponding optimal peak uncertainty and dwell-time sequences. Then, the algorithm checks if every target is active, and if that is the case, the algorithm stops.  However, if there is some inactive target, the algorithm checks if it is beneficial to permanently exclude one of such targets from the search (note the distiction from {\it inactive} and {\it excluded}: inactive targets may become active at a subsequent step in the search, but excluded targets no longer can become active, and are completely removed from the search process). First, the inactive target with the lowest peak uncertainty is excluded. After this exclusion, two stopping criteria are checked: (i) are all the remaining targets active?, or (ii) is the new consensus peak uncertainty $g_{con}$ lower than the steady state uncertainty of the excluded target? If one of these criteria is met, no further change in the visiting sequence will lower the cost, as the limiting factor is the already excluded target. If none of the excluding criteria are met, the algorithm proceeds to remove the next inactive target.

Algorithm \ref{alg:single_visit_visiting_sequence} describes this procedure more formally.  In this algorithm, MinimumLengthCycle($\mathcal{W}$) gives the minimum length cycle that visits all the nodes contained in $\mathcal{W}$. This function can be implemented as the solution of the TSP. 


    \begin{algorithm}
    \caption{Optimal Visiting Sequence Computation}
    \label{alg:single_visit_visiting_sequence}
    \begin{algorithmic}[1]
    \State $\mathcal{W} \leftarrow \mathcal{V}$; 
    \State $\Xi \leftarrow \text{MinimumLengthCycle}(\mathcal{W})$;
    \State $[\mathcal{T},C,g_{con}]\leftarrow\text{OptimizeDwelling}(\Xi)$;
    \While{True}
        \State $[c_{\min},i_{\min}] \leftarrow \min(C)$;
        \Comment{$C=$ peak uncertainties.}
        \If{$c_{\min}=g_{con}$} 
            \State \textbf{Break}; 
            \Comment{Every visited target is active.}
        \Else{
            \State $\mathcal{W} \leftarrow \mathcal{W} \backslash \{i_{\min}\}$;
            \State $\Xi \leftarrow \text{MinimumLengthCycle}(\mathcal{W})$;
            \State $[\mathcal{T},C,g_{con}]\leftarrow\text{OptimizeDwelling}(\Xi);$
            \If{ $g_{con}<c_{\min}$} 
                \State \textbf{Break}; \Comment{Inactive target is the bottleneck.}
            \EndIf
            }
        \EndIf
        
    \EndWhile
    \State{\bf Return: $\Xi,\mathcal{T},C,g_{con}$} 
    \end{algorithmic}
    \end{algorithm}

    \paragraph*{\textbf{Some Simulation Results}}
    \begin{figure*}[htp!]
    \centering
    \begin{subfigure}[t]{0.31\textwidth}
        \centering\includegraphics[width=\textwidth]{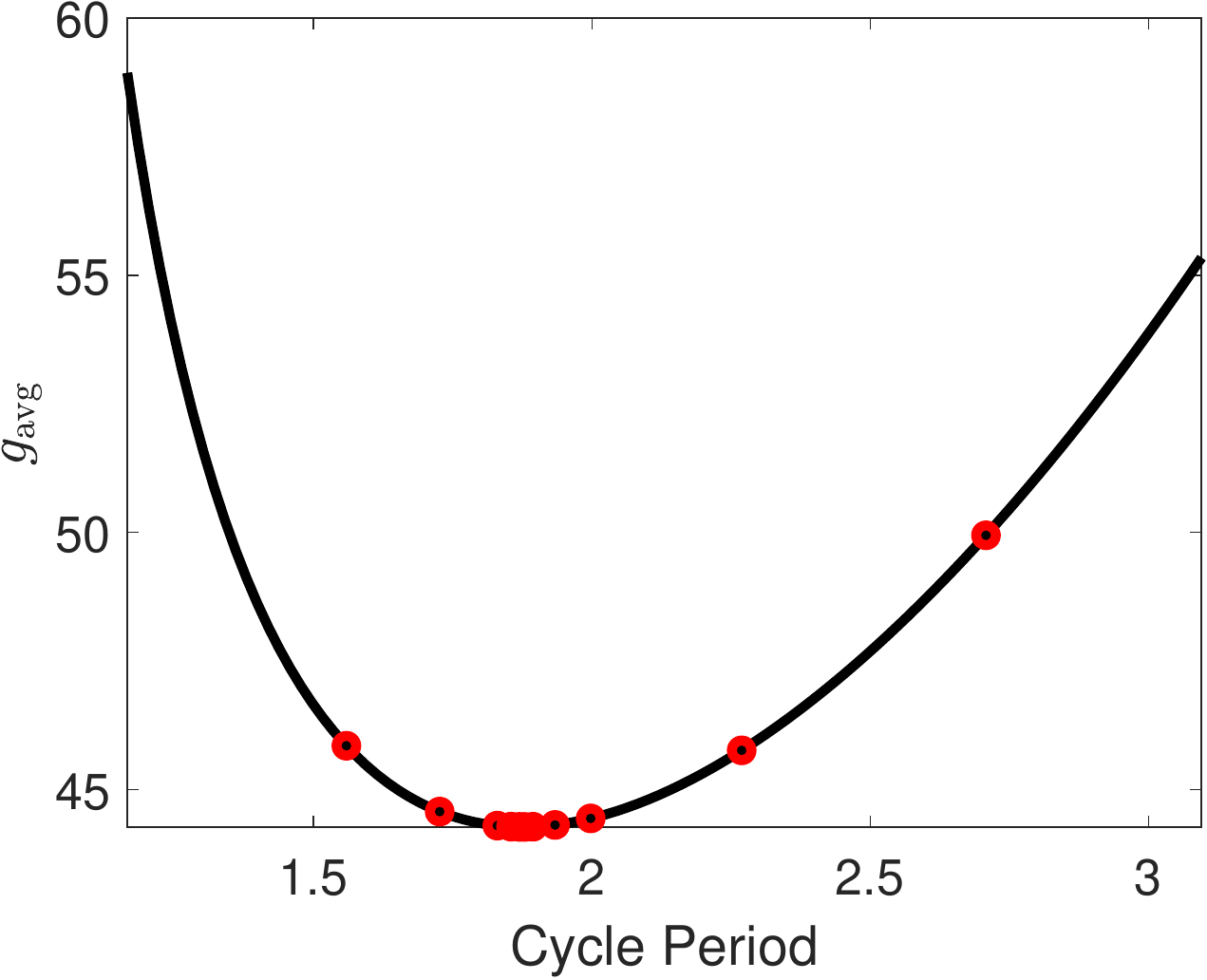}
        \caption{Peak value after balance among targets, as a function of cycle period.}
        \label{fig:peak_period}
    \end{subfigure}\hfill
    \begin{subfigure}[t]{0.32\textwidth}
        \centering\includegraphics[width=\textwidth]{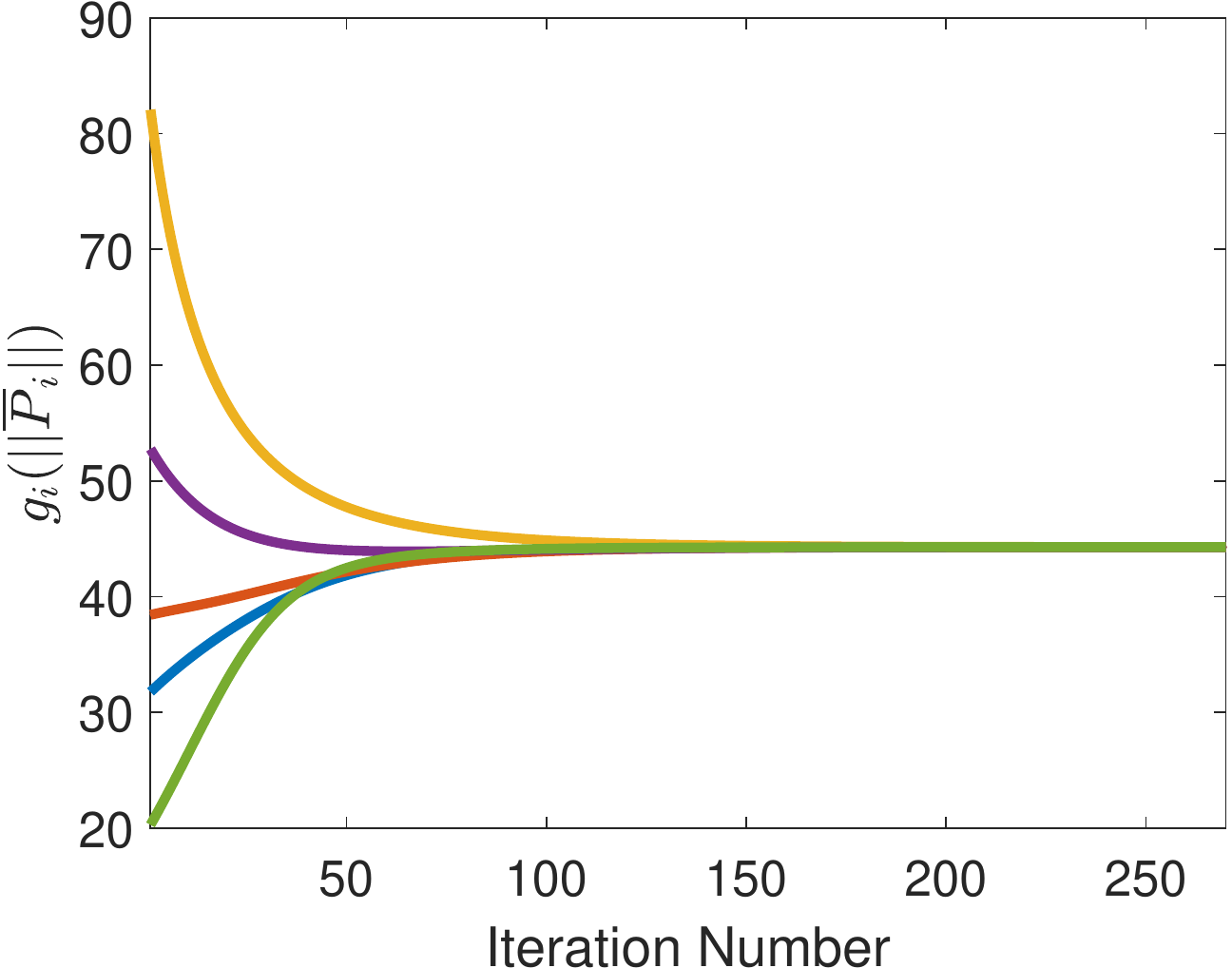}
        \caption{Peak uncertainty at the optimal period.}
        \label{fig:cost_iteration}
    \end{subfigure}\hfill
    \begin{subfigure}[t]{0.32\textwidth}
        \centering\includegraphics[width=\textwidth]{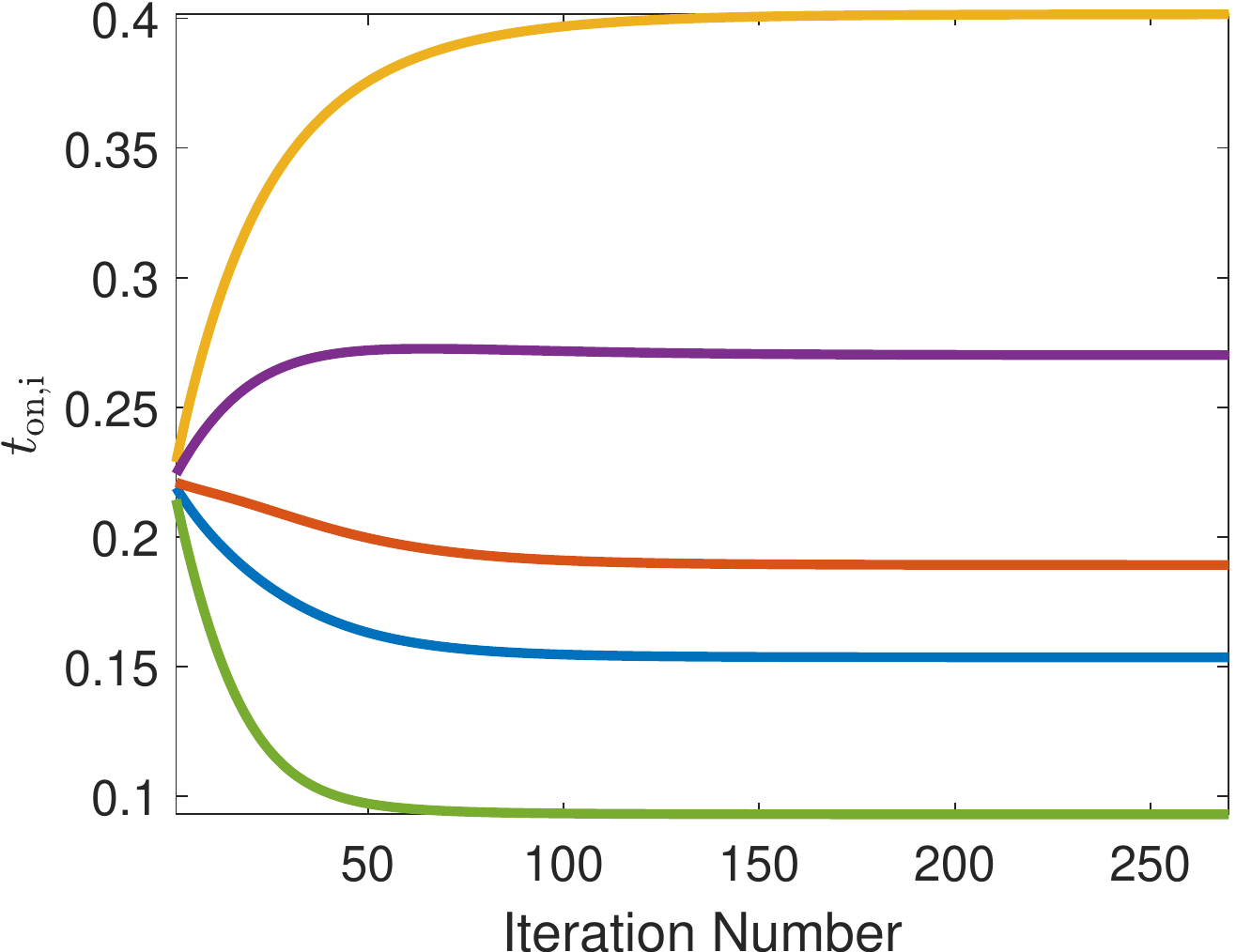}
        \caption{Dwell-time at the optimal period.}
        \label{fig:time_iteration}
    \end{subfigure}
    \caption{Results of simulating Algorithm \ref{alg:compute_optimal_visitation_times}. In (a), the balanced peak uncertainty, as a function of the cycle period. The red dots mark the values of $T$ that were explored by the golden ratio search. In (b)-(c), we show the evolution of the peak uncertainty and the dwell-time for each target.}
    \label{fig:results_1_one_target}
    \vspace{-3mm}
    \end{figure*}
    
    We have implemented the model described in \eqref{eq:dynamics_phi} and \eqref{eq:observation_model_ij}, with parameters indicated in Table \ref{tab:sim_param}. Note that targets are also assigned colors that will be used to identify each target in figures containing the simulation results. For simplicity, the internal states of the targets were assumed to be scalars. The targets location was drawn from a uniform distribution in $[0,0.5]\times[0,0.5]$. The target locations are displayed in Fig. \ref{fig:traj_simulation_single_agent} and the graph was assumed to be fully connected, with edge costs being the distance between two targets. Moreover, for the definition of the optimization goal as in \eqref{eq:cost}, we used $g_i(\xi)=\xi$, $\forall i$, and $\norm{\Gamma}=|\Gamma|$.
    
\begin{table}[h!]
\centering
\caption{Parameters used in the simulation.}
\label{tab:sim_param}
\begin{tabular}{ |c||c|c|c|c|c| } 

 \hline
 Target & 1 & 2 & 3 & 4 & 5  \\ 
 \hline
 Color& blue & red & yellow & purple & green  \\ 
 \hline
 $A_i$ & 0.3487 & 0.1915 & 0.4612 & 0.2951 & 0.1110  \\
 \hline
 $Q_i$ & 1.1924 & 1.2597 & 0.8808 & 1.7925 & 0.4363  \\
 \hline
 $R_i$ & 2.3140 & 7.1456 & 4.2031 & 5.2866 & 7.5314  \\
 \hline
\end{tabular}
\vspace{-3mm}
\end{table}

    For the visiting sequence, we considered the TSP cycle and Alg. \ref{alg:compute_optimal_visitation_times} was then used to find the corresponding optimal dwelling sequence. The parameter $k_p$ in Eq. \eqref{eq:consensus_update_discrete} was chosen to be $10^{-2}$ and we set $[T_{\min},T_{\max}]=[0.1t_{\text{travel}},3t_{\text{travel}}]$, where $t_{\text{travel}}$ is the total travel-time required to complete the cycle. 
    
\begin{figure}[htp!]
    \centering
        \begin{subfigure}[t]{0.55\columnwidth}
    \vskip 0pt
    \includegraphics[width=\columnwidth]{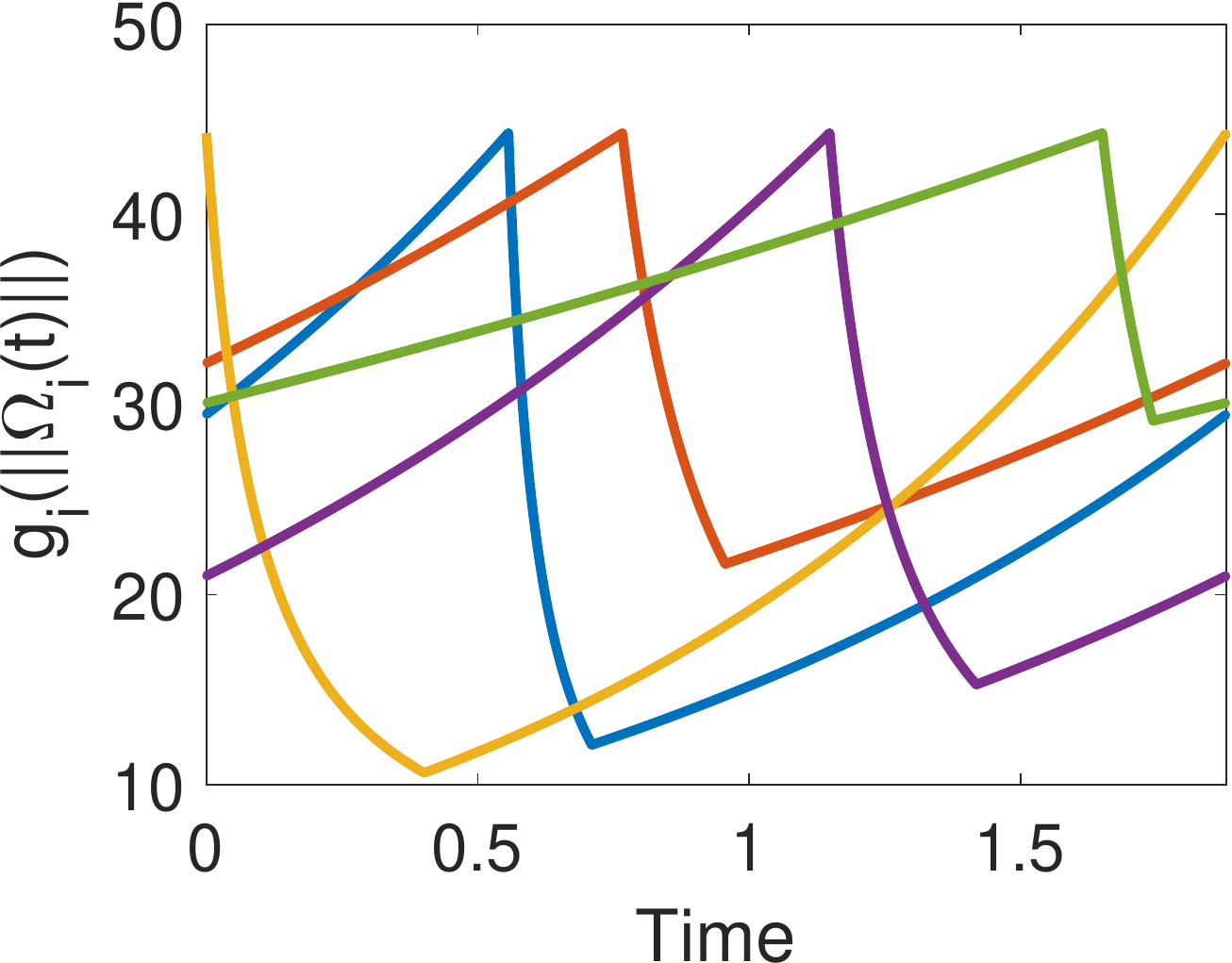}%
    \caption{Covariance over one period.}
    \end{subfigure}%
    \begin{subfigure}[t]{0.45\columnwidth}
    \vskip 5pt
    \includegraphics[width=\columnwidth]{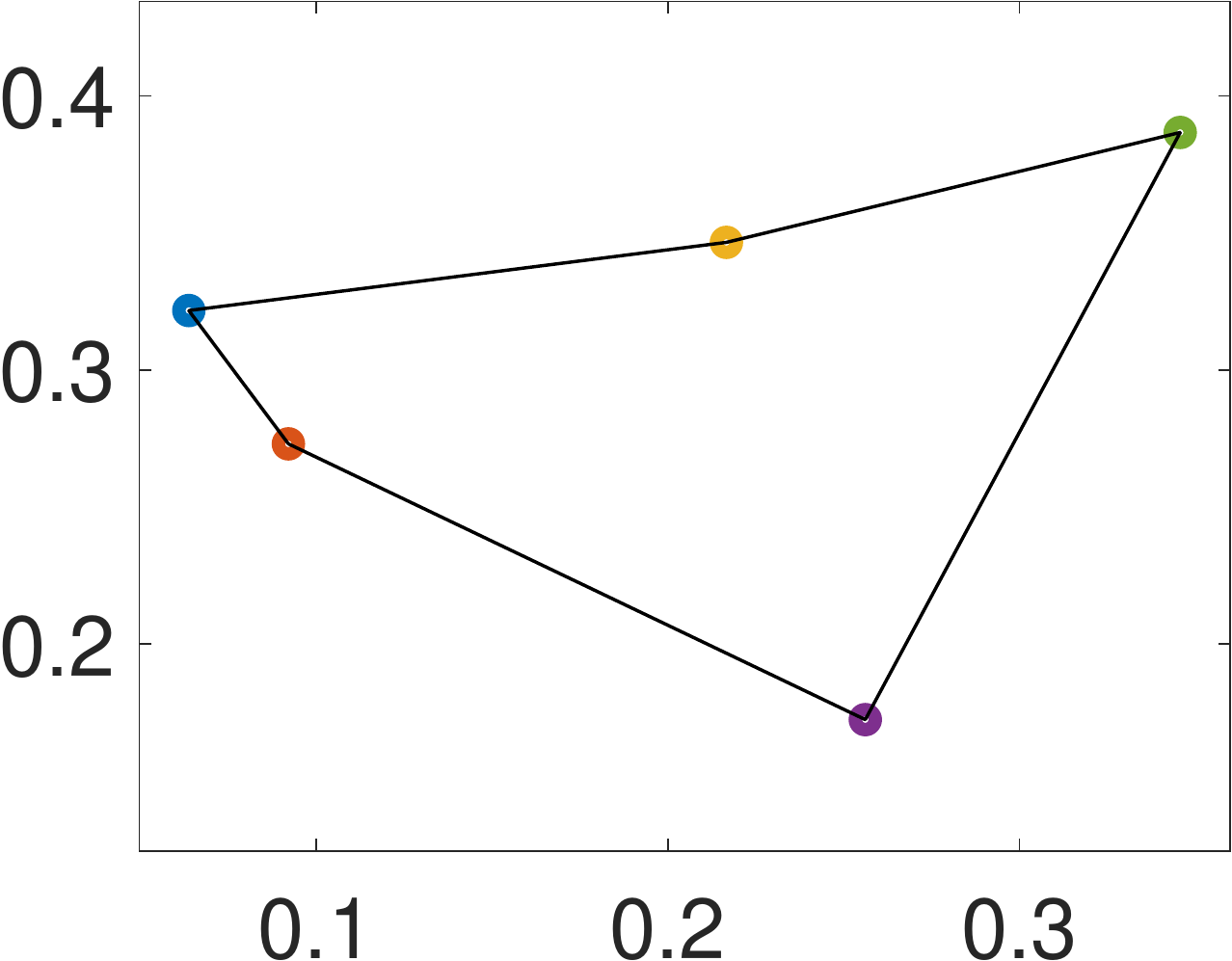}%
    \caption{Agent trajectory (black) and target locations (colored).}
    \label{fig:traj_simulation_single_agent}%
    \end{subfigure}
    \caption{Results obtained after optimizing the visiting and dwelling sequences.}
    \label{fig:results_2_one_target}
    \vspace{-3mm}
    \end{figure}
    
    The results are shown in Figs. 
    \ref{fig:results_1_one_target} and \ref{fig:results_2_one_target}. In particular, Fig. \ref{fig:results_1_one_target} shows the evolution of the steady-state covariance matrix norm over one complete period of the agent trajectory, under the optimal dwelling sequence. In Fig. \ref{fig:results_1_one_target}, details of the optimization process are highlighted. In Fig. \ref{fig:peak_period}, we can see how the peak uncertainty behaved as a function of the period (after balancing the dwelling times among targets). Moreover, this figure also highlights that the golden ratio search scheme was efficiently converged to what appears to be a global minimum. Figs. \ref{fig:cost_iteration} and \ref{fig:time_iteration} show how the dwelling sequence and the peak covariance varied while using the update law \eqref{eq:update_law_single_target}. Initially, all targets are visited for the same amount of time. However, as the iterations go on, the dwell-time spent on some targets becomes larger than that of the others. As expected, in the final iteration, all the peak covariances have converged to the same value (consensus).




\section{Optimal Dwelling Sequence on an Unconstrained Visiting Sequence}
\label{sec:multiple_visits}

In Section \ref{sec:optimal_dwelling_single_visit}, we only considered situations where each target is (effectively) visited at most once in every cycle (i.e., we assumed that visiting sequences are constrained). Then, we designed a procedure that computes the optimal dwelling sequence for a given visiting sequence and made some remarks about selecting the optimal visiting sequence. This section extends these ideas to scenarios where targets can be visited multiple times in a cycle (we call such visiting sequences ``unconstrained''). In particular, we design an algorithm that aims to obtain a dwelling sequence such that the peak uncertainty (i.e., $g_i(\norm{\bar{P}_i^k})$) is the same at each visit $k$ for every active target $i$. In other words, the goal is not only to have every target with the same maximum peak uncertainty, but also that the peak uncertainties within visits to the same target have the same peak value. Note that in this section, we assume the visiting sequence as a given (the process of determining the optimal visiting sequence is discussed in Section \ref{Sec:GreedyExplorations}).

As we now consider situations where targets can be visited multiple times, to maintain the same notation as before, let 
\begin{equation}
    t_{\text{on},i}=\sum_{r=1}^{N_{i}}t_{\text{on},i}^r,
\end{equation}
where the index $r$ refers to the $r$\textsuperscript{th} visit to the target of interest, in the given visiting sequence $\Xi$. Thus, we maintain the same update law as in the case with constrained visiting sequences, given in \eqref{eq:consensus_update_discrete}.

However, only defining the total dwell-time spent at each target is not sufficient to determine the peak uncertainties as one has to determine how long the agent dwells at each of its visits at each target. 
To this end, we consider an additional optimization step, that takes place for each target for each total dwell-time update step $k$ (of \eqref{eq:consensus_update_discrete}). In particular, starting from an arbitrary valid distribution of dwell-times $t_{\text{on},i}^p[1]$ such that $\sum_{p=1}^{N_i}t_{\text{on},i}^p[1]=t_{on,i}[1]$, we propose the update law:
\begin{equation}
    t_{\text{on},i}^p[m+1] = t_{\text{on},i}^p[m]+k_p \log \left(\frac{g_i\left(\norm{\overline{P}_i^p}\right)}{g_{\text{avg,i}}}\right),
\end{equation}
where $g_{\text{avg,i}}=\left(\Pi_{p=1}^{N_i} \bar{P}_{i}^p\right)^{\frac{1}{N_i}}$. The intuition behind this update law is that, similar to how a total dwell-time can be split among different targets to reach the same peak value, it can be split within the same target in order to yield the same peak value (at the beginning of each visit to that target). To better illustrate the relationship between both update stages, the complete optimization process is described in Alg. \ref{alg:compute_visitation_times_multiple_visits}. We use the procedure in Alg. \ref{alg:compute_visitation_times_multiple_visitsMV} (referenced as ``MV'') that is responsible for balancing the dwelling time of multiple visits to the same target. Additionally, we point out that the procedure ``computeToff'' is responsible for computing the time between subsequent visits to the same target using \eqref{eq:def_toff} and ``SSPeaks'' computes the steady state peak uncertainty values.

\begin{algorithm}[!h]
\caption{Computing the optimal dwelling sequence}
\label{alg:compute_visitation_times_multiple_visits}
    \begin{algorithmic}[1]
    \State{{\bf Input:} Visiting sequence $\Xi$, Cycle period $T$}.
    \State{$k\leftarrow 1$};
    \State{$t_{\text{on},i}[k] \leftarrow 1/N$};
    \For{$i\in\{1,...,N\}$}
    \For{$p\in\{1,...,N_i\}$}
        \State{$t_{\text{on},i}^p[k] \leftarrow t_{\text{on},i}[k]/N_i$};
    \EndFor
    \EndFor
    \State{$g_{prev} \leftarrow \infty$};
    \While{True}
    \For{$i\in\{1,...,N\}$}
    \State{$[\bar{P}_{i}[k],t_{\text{on},i}^p[k]]\leftarrow \text{MV}(i,t_{\text{on},i}[k],\tilde{t}_{\text{on},j}^p[k],T,\Xi)$};
    \EndFor
    \State{$g_{\text{avg}}\leftarrow\left(\prod_{j\in\mathcal{A}}g_j\left(\norm{\overline{P}_j}\right)\right)^{\frac{1}{M}}$};
    \If{$|g_{avg}-g_{prev}|<tol$}
        \State{\textbf{Break}};
    \Else
        \State{$t_{on,i}[k+1] = t_{on,i}[k] +k_p\log\left(\frac{g_i\left(\norm{\overline{P}_i}\right)}{g_{\text{avg}}}\right)$};
        \For{$i\in\{1,...,N\}$}
        \For{$p\in\{1,...,N_i\}$}
            \State{$\tilde{t}_{\text{on},i}^p[k+1] \leftarrow \frac{t_{\text{on},i}[k+1]}{t_{\text{on},i}[k]}\tilde{t}_{\text{on},i}^p[k]$};
        \EndFor
        \EndFor
        \State{$k\leftarrow k+1$};
    \EndIf
    
    \EndWhile    
    
    \State{\bf Return: $\bar{P}_i^p[k],t_{\text{on},i}[k]$}
\end{algorithmic}
\end{algorithm}

\begin{algorithm}[!h]
\caption{Computing dwell-times at a target: Process MV}
\label{alg:compute_visitation_times_multiple_visitsMV}
    \begin{algorithmic}[1]
    \State{\bf Input: $t$, $t_{\text{on},i}$, $t_{\text{on},j\neq i}^p$, $T$, $\Xi$.}
    \State{$t_{\text{off},i}^p\leftarrow \text{computeToff}(t_{\text{on},j\neq i}^p, T, \Xi)$};
    \State{$g_{prev,i} \leftarrow \infty$};
    \State{$m\leftarrow 1$};
    \While{True}
    \State{$[\bar{P}_i^1,...,\bar{P}_i^{N_i}]=\text{SSPeaks}(t_{\text{on},i}^p,t_{\text{off},i}^p)$};
    \State{$g_{\text{avg},i} \leftarrow \left(\prod_{j\in\mathcal{A}}g_j\left(\norm{\overline{P}_j}\right)\right)^{\frac{1}{M}}$};
    \If{$|g_{\text{avg},i}-g_{prev,i}|<tol$}
    \State{\textbf{Break}};
    \Else
    \For{$p\in\{1,...,N_i\}$}
            \State{$t_{\text{on},i}^p[m+1] = t_{\text{on},i}^p[m]+k_p \log \left(\frac{g_i\left(\norm{\overline{P}_i^p}\right)}{g_{\text{avg,i}}}\right)$};
    \EndFor
    \State{$g_{prev,i}\leftarrow g_{\text{avg},i}$};
    \State{$m\leftarrow m+1$};
    \EndIf
    \EndWhile
    \State{\bf Return: $g_{\text{avg},i},t_{\text{on},i}^p$};
\end{algorithmic}
\end{algorithm}

When the visiting sequence is unconstrained, one special difficulty in optimizing the dwelling sequence is: to compute the steady state covariance, a target must know the information about the agent dwell-times at other targets. Recall that when each target was only visited once during a cycle, the number of visits $N_i=1=p$ and thus $t_{\text{off},i}^p$ was computed based only on the dwell-time at target $i$, i.e. $t_{\text{off},i}^p=T-t_{\text{on},i}^p$. 

On the other hand, when multiple visits are allowed to each target, to compute the steady state covariance, one has to know the values of $t_{\text{on},i}^p$, $t_{\text{off},i}^p$ as well as $t_{\text{on},j}^p$ for all $j\neq i$ (see \eqref{eq:def_toff}). This interdependence does not allow independently optimizing the dwell-times and computing the steady state covariance at each target. Thus, as can be observed in line 18 of Alg.  \ref{alg:compute_visitation_times_multiple_visits}, to compute the steady state covariance (and dwelling sequence) at each target, what we propose is to assume each dwell-time holds the same proportional share of the total time as it did in the previous iteration. This challenge, however, imposes additional difficulties in formally proving the convergence and optimality of the proposed scheme, and thus this is a current topic of research. Nevertheless, our simulation results lead us to believe that this algorithm (Alg. \ref{alg:compute_visitation_times_multiple_visits}) converges for any positive cycle period $T$.

\section{A Greedy Solution for Optimal Visiting Sequence} \label{Sec:GreedyExplorations}

In previous sections, we explored how the optimal dwelling sequence should be determined for a given fixed visiting sequence. Naturally, to optimize the cost \eqref{eq:cost}, the next step is to determine the optimal visiting sequence on the given target network $\mathcal{G}$. However, for this problem, notice that the visiting sequence given by the corresponding TSP is only a sub-optimal solution as we now allow the visiting sequence to be unconstrained (see Proposition \ref{prop:optimal_cycle_tsp}). In this paper, we introduce a heuristic approach specifically designed for the problem we study to determine a high-performing (still sub-optimal) visiting sequence. In contrast to using computationally complex mixed-integer programming techniques \cite{Hari2019}, we use a genetic algorithm based TSP solver \cite{Kirk2020} together with a specifically designed greedy algorithm.  

In particular, the proposed greedy algorithm uses the TSP solution as its input, and, in each iteration, it \emph{explores} several possible modifications to the current solution and \emph{executes} the most profitable (greedy) modification. Intuitively, this greedy algorithm's computational efficiency depends on two factors: (i) the execution time required for the exploration of a single modified solution and (ii) the number of different modified solutions explored in each iteration. In order to strategically limit the prior, motivated by \cite{Welikala2019P3}, we propose a novel \emph{metric} that can be utilized to efficiently evaluate a solution (i.e., to explore a visiting sequence). Subsequently, to limit the latter, we exploit several structural properties of this PM setup. 


\subsection{The metric used to evaluate a visiting sequence}

The optimal cost for a given a visiting sequence (say $\bar{\Xi}$) is $J(\bar{\Xi},\mathcal{T}^*)$, where $J$ is as in \eqref{eq:cost} and $\mathcal{T}^*$ is the optimal dwelling sequence. In the optimization process corresponding to $\bar{\Xi}$ given by the optimization process proposed in Section \ref{sec:multiple_visits}  \ref{alg:compute_visitation_times_multiple_visits}). 
However, a greedy scheme requires repeated evaluations of the cost in each iteration, and thus a faster method for cost estimation is desired. Therefore, in this section, we establish a \emph{lower bound} to $J(\bar{\Xi},\mathcal{T}^*)$ as $\hat{J}(\bar{\Xi})\leq J(\bar{\Xi},\mathcal{T}^*)$, so that it (i.e., $\hat{J}(\bar{\Xi})$) can be used to efficiently evaluate any visiting sequence $\bar{\Xi}$.

Let us denote a generic visiting sequence (also called cycle) by $\bar{\Xi} = \{p_j\}_{j=1,2,\ldots,N}$ where each $p_j\in \mathcal{V}$. Note that the visiting sequence $\bar{\Xi}$ fully defines a corresponding sequence of edges $\bar{\xi} \subseteq \mathcal{E}$ as $\bar{\xi} = \{(p_{j-1},p_j)\}_{j=1,2,\ldots,N}$ with $p_{0}=p_N$. Since targets are allowed to be visited more than once during a cycle, some elements in $\bar{\Xi}$ may have repeated entries (i.e., there may be $p_j,p_l\in\bar{\Xi}$ such that $p_j=p_l$ even though $j\neq l$). For notational convenience, we define an equivalent cycle to $\bar{\Xi}$ that, however, has unique entries as $\Xi = \{p_j^k\}_{j=1,2,\ldots,N}$ where $p_j^k$ represents the $k$\textsuperscript{th} instance (visit) of the target $p_j\in\bar{\Xi}$. 
Recall that we previously used $N_i$ to represent the number of times target $i$ is visited in a cycle. Therefore, for each $p_j^k \in \Xi,\ 1 \leq k \leq N_{p_j}$. The corresponding sequence of edges of $\Xi$ is denoted by $\xi =\{(p_{j-1}^l,p_j^k)\}_{j=1,2,\ldots,N}$. For example, if $\mathcal{V}=\{1,2,3\}$ in $\mathcal{G}$, $\bar{\Xi} = \{2,3,1,3,1,3\}$ is an example cycle where its equivalent version would be $\Xi=\{2^1,3^1,1^1,3^2,1^2,3^3\}$. In essence, any given visiting sequence can be represented by the corresponding cycle $\bar{\Xi}$ or by any of its equivalent representations $\bar{\xi}$, $\Xi$ or ${\xi}$.   

Next, let us define the \emph{auxiliary target-pool} of a target $p_j\in\bar{\Xi}$ as $\tau_{p_j}=\{p_{j}^1,p_{j}^2,\ldots,p_j^{N_{p_j}}\}$. The \emph{sub-cycle} of a target $p_j^k\in\Xi$ is denoted as $\Xi_{p_j}^k$ and is defined as the ordered portion of elements in $\Xi$ starting after $p_j^{k-1}$ until $p_j^k$. For instance, in the previous example, the sub-cycles corresponding to targets $3^2$ and $1^1$ are $\Xi_3^2 = \{1^1,3^2\}$ and $\Xi_1^1 = \{3^3,2^1,3^1,1^1\}$ respectively. Similarly, $\xi_{p_j}^k$ is used to denote the sequence of edges of the corresponding sub-cycle $\Xi_{p_j}^k$. We further define $w_{p_j}^k$ as the total sub-cycle travel time required to traverse the edges in $\xi_{p_j}^k$. Notice that $w_{p_j}^k$ can be viewed as the $k$\textsuperscript{th} revisit time of the target $p_j\in\bar{\Xi}$ if no dwell-time was spent at any target in the sub-cycle $\xi_{p_j}^k$.


We next prove that, for any dwelling sequence $\mathcal{T}$,
\begin{equation}
    \label{eq:lower_bound_definition}
    J(\bar{\Xi},\mathcal{T})  \geq \max_{i \in \bar{\Xi}} L_i(\bar{t}_i)
\end{equation}
where $\bar{t}_i$ is the maximum revisit time of target $i\in\bar{\Xi}$, i.e.,  
\begin{equation}\label{eq:maximum_revisit_time}
    \bar{t}_i = \max_{k:i^k \in \tau_{i}}\ w_i^k, 
\end{equation}
and 
\begin{multline}
    \label{eq:lower_bound_expression}
    L_{i}(\bar{t}_i) = g_i\Big(||\exp{(A_i \bar{t}_i)}\Omega_{\text{ss},i}\exp(A_i^T \bar{t}_i)\\
    +\int_0^{\bar{t}_i} \exp(A_i(\bar{t}_i-\tau))Q_i\exp(A_i^T(\bar{t}_i-\tau))d\tau||\Big),
\end{multline}
with $\Omega_{\text{ss},i}$ being the positive definite solution of the algebraic Riccati Equation
$$A_i\Omega_{\text{ss},i}(t)+\Omega_{\text{ss},i}(t)A_i'+Q_i-\Omega_{\text{ss},i}(t){G}_i\Omega_{\text{ss},i}(t)=0.
$$

\begin{remark}
    The definition of $\bar{t}_i$ in \eqref{eq:maximum_revisit_time} is equivalent to $\max_{1\leq k \leq N_i}t_{\text{off},i}^k$ when $t_{\text{on},m}^n = 0$ for all targets $m$ and visiting instances $n$. In other words, $\bar{t}_i$ is the maximum time between two consecutive visits to the target $i$ given no dwell-time is spent on any of the targets in the visiting sequence.
\end{remark}

The intuition behind this lower bound is that it computes the peak target covariance (which directly affects the cost $J(\bar{\Xi},\mathcal{T})$ in \eqref{eq:cost}) as if when the target is visited, its covariance instantaneously decreases to the steady state value. This is a key result as it provides a computationally efficient metric:
\begin{equation}\label{Eq:CycleMetricGlobal}
    \hat{J}(\bar{\Xi}) = \max_{i\in \bar{\Xi}} L_i(\bar{t}_i)
\end{equation}
to estimate the cost $J(\bar{\Xi},\mathcal{T})$ of a known visiting sequence $\bar{\Xi}$. As stated earlier, we use this $\hat{J}(\cdot)$ metric in our greedy scheme to efficiently evaluate and thus compare the cost of different visiting sequences so as to find the optimal  visiting sequence.

\begin{proposition}
    $J(\bar{\Xi},\mathcal{T})  \geq \max_{i\in \bar{\Xi}} L_i(\bar{t}_i)$ holds for any dwelling sequence $\mathcal{T}$ and visiting sequence $\bar{\Xi}$, with $L_i$ defined as in \eqref{Eq:CycleMetricGlobal}.
\end{proposition}
\begin{proof}
    From the general solution of a linear matrix Riccati equation we get that
    \begin{multline}
        \overline{P}_i^k = \exp(A_i t_{\text{off},i}^{k-1})\underline{P}_i^{k-1}\exp(A_i^T t_{\text{off},i}^{k-1})\\+\int_0^{t_{\text{off},i}^{k-1}}\exp(A_i (t_{\text{off},i}^{k-1}-\tau))Q_i\exp(A_i^T (t_{\text{off},i}^{k-1}-\tau))d\tau
    \end{multline}
    Moreover, we note that $\bar{\Omega}_i(t) \succ \Omega_{\text{ss},i}$, where $\Omega_{\text{ss},i}$ is steady state covariance matrix of target $i$ when it is permanently observed. Therefore, 
    \begin{multline}
        \overline{P}_i^k \prec \exp(A_i t_{\text{off},i}^{k-1})\Omega_{\text{ss},i}\exp(A_i^T t_{\text{off},i}^{k-1})\\+\int_0^{t_{\text{off},i}}\exp(A_i (t_{\text{off},i}^{k-1}-\tau))Q_i\exp(A_i^T (t_{\text{off},i}^{k-1}-\tau))d\tau
    \end{multline}
    
    Note that if we replace $t_{\text{off},i}$ by $\bar{t}_i$, we get the definition of $L_i(\bar{t}_i)$ in \eqref{eq:lower_bound_definition}. Using Lemma \ref{lemma:pd_nd_time_derivative}, we know that if $\bar{t}_i\leq \max_{1\leq k\leq N_i}t_{\text{off},i}^k$, then 
    $
    J(\bar{\Xi},\mathcal{T})  \geq \max_{i\in \mathcal{V}} L_i(\bar{\Xi})
    $
\end{proof}

Note that this metric does not require the computation of dwelling times, thus can be immediately computed for a given trajectory. Additionally, we highlight that this metric for fast approximate cost evaluation can be used in conjunction with any heuristic method for exploration of visiting sequences, and is not exclusively tied to the specific heuristic we consider.

\subsection{Possible types of modifications for a visiting sequence}

As stated earlier, in each greedy iteration, we explore several modified versions of the current visiting sequence. In particular, we use three types of cycle modification operations (CMOs) to obtain modified cycles. Before discussing each of them, we first introduce some notations and a lemma.

Let us denote $\bar{\Xi}$ as the current cycle in a greedy iteration and $\Xi$ as its equivalent representation with unique entries (the respective sequences of edges $\bar{\xi}$ and $\xi$). We denote the \emph{critical target} $i^k\in\Xi$ as $i=\argmax{i\in\Xi}\hat{J}(\bar{\Xi})$ and $k=\argmax_{k\in\tau_i}{w_k}$ (i.e., the optimal $i\in\bar{\Xi}$ in \eqref{eq:lower_bound_definition} and the optimal $k\in\{1,2,\ldots,N_i\}$ in \eqref{eq:maximum_revisit_time}) as $i^{k*}$. The corresponding sub-cycle, the sequence of edges and the total sub-cycle travel time are denoted as $\Xi_{i}^{k*}$, $\xi_{i}^{k*}$ and $w_{i}^{k*}$, respectively.

\begin{lemma}\label{Lm:CriticalTargetVisit}
    The metric $\hat{J}(\bar{\Xi})$ can only be reduced (improved) by modifying the sub-cycle $\Xi_{i}^{k*}$ so that $w_{i}^{k*}$ is decreased. 
\end{lemma}
\begin{proof}
From \eqref{eq:lower_bound_definition}, \eqref{eq:maximum_revisit_time} and \eqref{eq:lower_bound_expression}, it is clear that $\hat{J}(\bar{\Xi}) = L_i(w_i^{k*})$. According to \eqref{eq:lower_bound_expression}, $L_i(\cdot)$ is a monotonically increasing function. Therefore, to reduce the metric $\hat{J}(\bar{\Xi})$, the maximum revisit time $w_i^{k*}$ should be decreased. This can only be achieved if the corresponding sub-cycle $\Xi_i^{k*}$ is modified. 
\end{proof}

The above lemma implies that we only need to modify a portion of the complete cycle $\bar{\Xi}$, i.e., the sub-cycle $\Xi_{i}^{k*}$ to improve the metric $\hat{J}(\bar{\Xi})$. Hence, this result significantly reduces the number of modified cycles that need to be explored in a greedy iteration (thus reduces the computational complexity). We are now ready to introduce the three types of cycle modification operations shown in Fig. \ref{Fig:CMOs}. 

\begin{figure}[!h]
    \centering
    \includegraphics[width=\columnwidth]{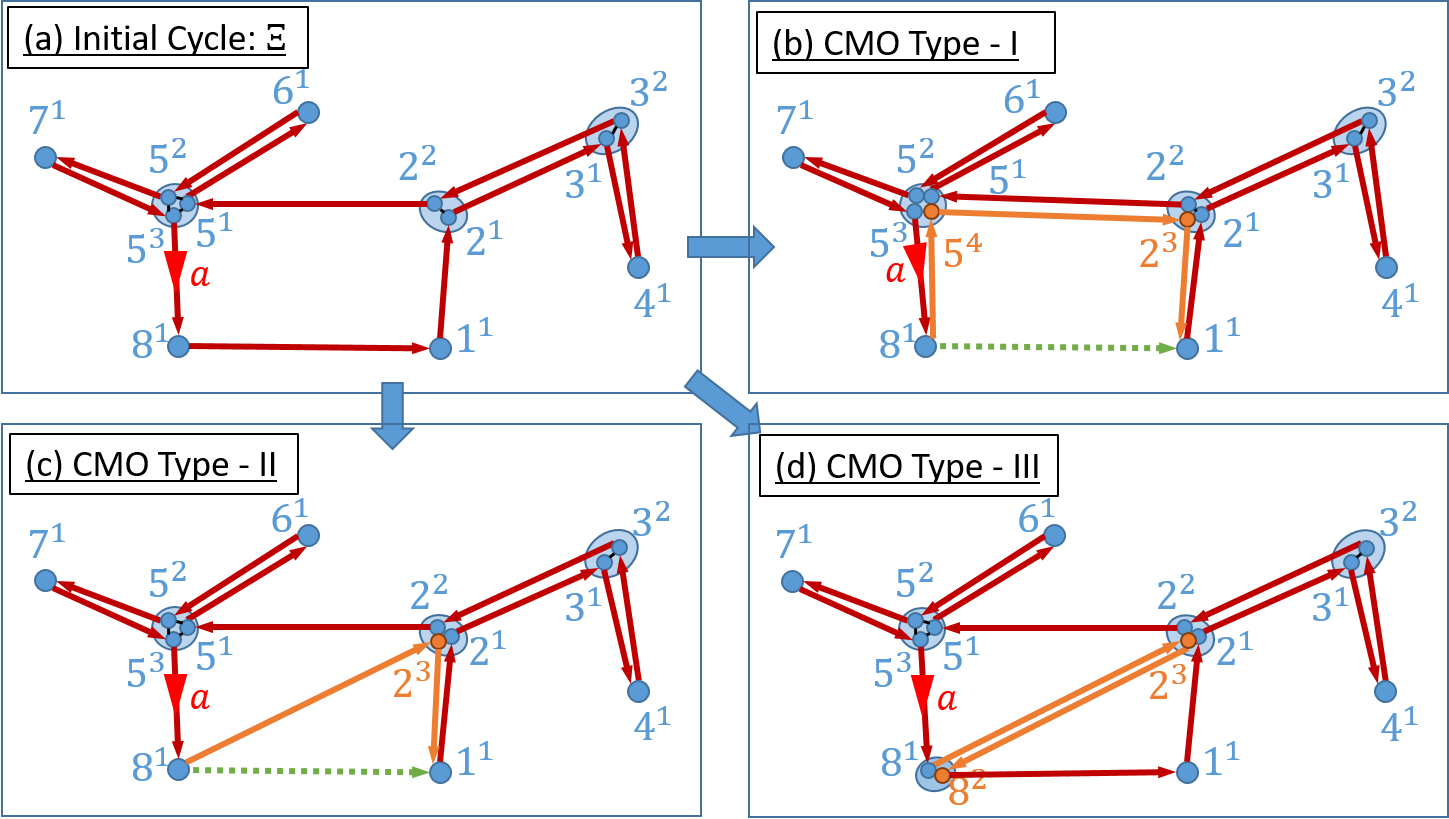}
    \caption{Three types of cycle modification operations (CMOs).}
    \label{Fig:CMOs}
\end{figure}

\paragraph*{\textbf{CMO Type - I}} Remove an edge $(l^m,j^n) \in \xi_{i}^{k*}$ and replace it with the fastest path between targets $l, j \in \bar{\Xi}$. 

Clearly, this modification is only effective if the fastest path between targets $l, j \in \bar{\Xi}$ is not the direct path $(l^m,j^n) \in \xi_{i}^{k*}$ that we remove. In practice, this CMO is useful in early greedy iterations - as we propose to start the greedy process with the TSP solution, where the agent is constrained to visit each target only once. Hence, such a TSP cycle may contain edges with high travel-time values that can be omitted if the agent is allowed to make multiple visits to some targets. 

One example where CMO Type - I has helped to reduce the $\hat{J}$ value is illustrated in Fig. \ref{fig:ExampleCMOTypeI}. Note that the gray colored edges in the \emph{graph diagram} in Fig. \ref{fig:ExampleCMOTypeI}(a) represent the edges with very high (infinite) travel-time values.  The \emph{cycle diagrams} in Fig. \ref{fig:ExampleCMOTypeI}(b) and (c) convey the $\{L_i(w_i^k), i^k\in\Xi\}$ values (as vertical gray colored bars) and travel-time values between targets (as circular red colored segments) of the cycle. Note that each such cycle diagram also indicates: (i) the cycle version with the unique entries $\Xi$, (ii) the lower bound metric value $\hat{J}(\bar{\Xi})$ and (iii) the critical target $i^{k*}$.

\begin{figure}[!h]
\centering
\begin{subfigure}[t]{0.48\columnwidth}
\includegraphics[width=\columnwidth]{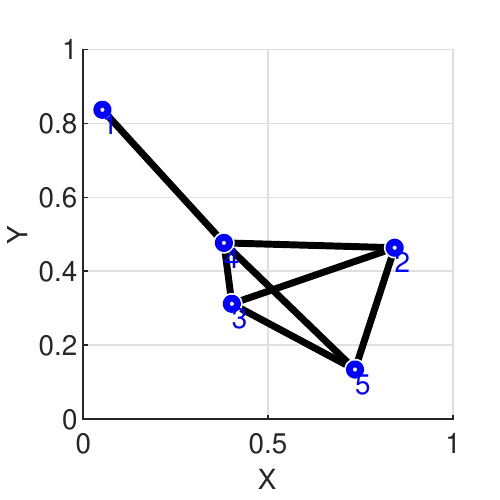}%
\caption{Graph $\mathcal{G}=(\mathcal{V},\mathcal{E})$}
\end{subfigure}%
\begin{subfigure}[t]{0.48\columnwidth}
\includegraphics[width=\columnwidth]{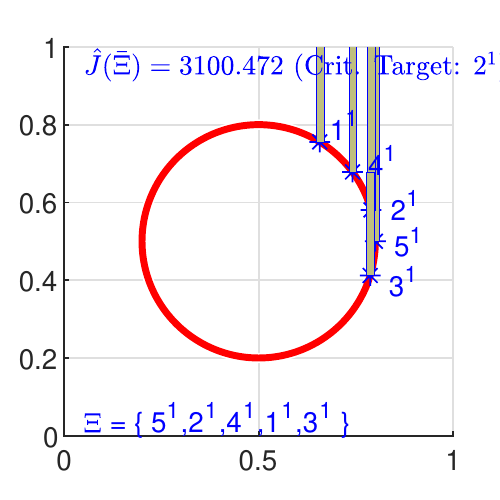}%
\caption{The Initial Cycle}
\end{subfigure}
\begin{subfigure}[t]{0.96\columnwidth}
\centering
\includegraphics[width=0.5\columnwidth]{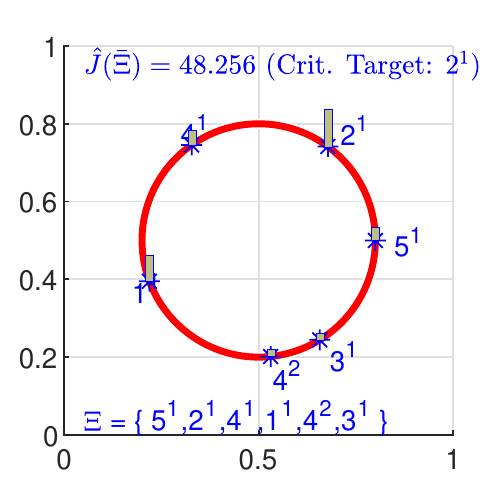}%
\caption{After a CMO Type - I: Replacement of the edge \\ \centering $(1^1,3^1)$ with the shortest path $(1^1, 4^2),(4^2, 3^1)$.}
\end{subfigure}
\caption{An example for CMO Type - I.}
\label{fig:ExampleCMOTypeI}
\end{figure}

\paragraph*{\textbf{CMO Type - II}}
Remove an edge $(l^m,j^n)\in\xi_{i}^{k*}$ such that $l^m \neq i^{k*}$ and $j^n \neq i^{k*}$. Then, replace it with two edges: $(l^m,i^{k*}),(i^{k*},j^n)$.

The following lemma can be established regarding the effectiveness of this CMO. 
\begin{lemma}\label{Lm:CMO_Type_II}
    If the travel-times between targets $l,m,i\in\bar{\Xi}$ are such that $|w_{l,m}- w_{i,m}| < w_{l,m}$, then the CMO Type - II described above improves the $\hat{J}(\bar{\Xi})$ value. 
\end{lemma}
\begin{proof}
The CMO Type - II adds an extra visit to the target $i^{k*}$ by dividing its sub-cycle $\Xi_{i}^{k*}$ into two: say $\Xi_i^{k_1}$ and $\Xi_i^{k_2}$. Using the triangle inequalities, it can be shown that the corresponding new revisit times: say $w_i^{k_1}$ and $w_i^{k_2}$, will be lower than the previous (critical) revisit time $w_i^{k*}$. The proof is complete by observing that $\hat{J}(\bar{\Xi}) = L_i(w_i^{k*})$ (from \eqref{eq:lower_bound_definition}, \eqref{eq:maximum_revisit_time}, \eqref{eq:lower_bound_expression}) and $L_i(\cdot)$ is a monotonically increasing function.  
\end{proof}

\paragraph*{\textbf{CMO Type - III}}
Select a target $j^n\in\Xi_i^{k*}$ such that $j^n\neq i^{k*}$ and insert two new target visits: $\{i^{k*},j^{n+1}\}$ soon after it.

Similar to the CMO Type - II, this CMO also adds an extra visit to the critical target $i^{k*}$. Using Lemma \ref{Lm:CMO_Type_II}, it is easy to see that this CMO can always improve the metric $\hat{J}(\bar{\Xi})$. An illustrative example of this CMO, along with CMO Type - II, is shown in Fig. \ref{fig:ExampleCMOTypeIIAndIII}.

\begin{figure}[h!]
\centering
\begin{subfigure}[t]{0.48\columnwidth}
\includegraphics[width=\columnwidth]{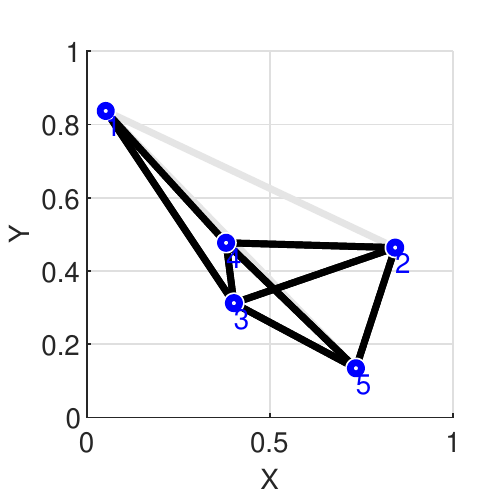}
\caption{Graph $\mathcal{G}=(\mathcal{V},\mathcal{E})$}
\end{subfigure}%
\begin{subfigure}[t]{0.48\columnwidth}
\includegraphics[width=\columnwidth]{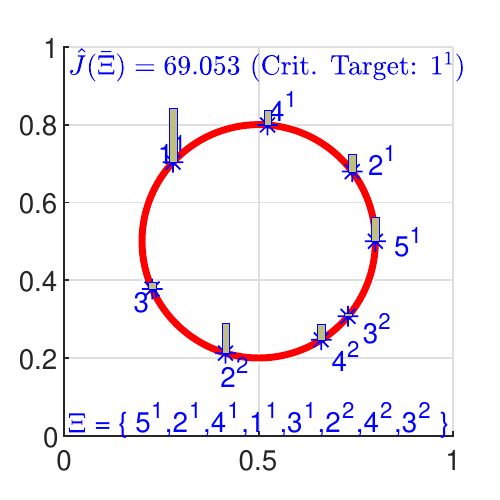}
\caption{The Initial Cycle}
\end{subfigure}
\begin{subfigure}[t]{0.48\columnwidth}
\includegraphics[width=\columnwidth]{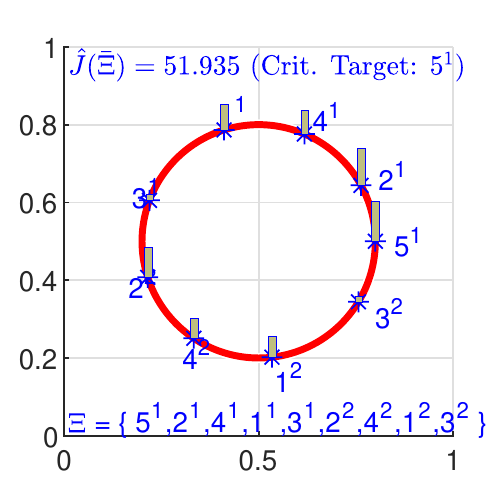}
\caption{After a CMO Type - II \\ \centering (Insertion of target $1^2$ to the initial cycle in (b))}
\end{subfigure}
\begin{subfigure}[t]{0.48\columnwidth}
\includegraphics[width=\columnwidth]{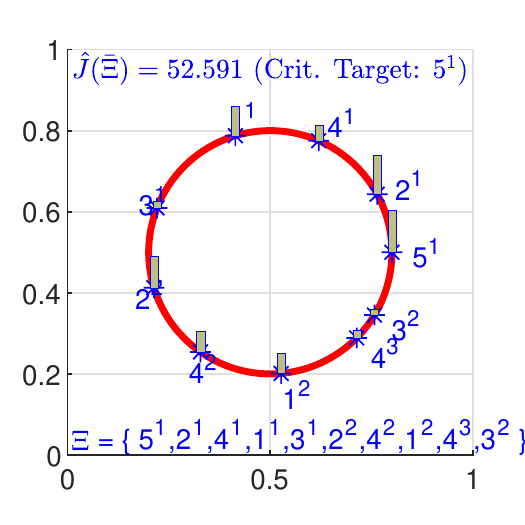}
\caption{After a CMO Type - III \\ \centering (Insertion of targets $1^2,\,4^3$ to the initial cycle in (b))}
\end{subfigure}
\caption{Examples for CMO Type II and III.}
\label{fig:ExampleCMOTypeIIAndIII}
\end{figure}

Under these three types of CMOs, the total number of modified cycles to be explored in a greedy iteration is strictly less than $3\vert \Xi_{i}^{k*} \vert$ ($\vert\cdot\vert$ stands for the cardinality operator when the argument is a set). Most importantly, as greedy iterations go on, according to the CMO Type - II and III discussed above, we always break the critical sub-cycle $\Xi_{i}^{k*}$ into two. Therefore  $\vert \Xi_{i}^{k*} \vert$ value will effectively decrease with the iterations.

\subsection{Greedy Algorithm}
To construct a high-performing sub-optimal visiting sequence (say $\bar{\Xi}_G$), we propose the greedy scheme given in Alg. \ref{Alg:Greedy_Exploration_Algorithm}. It starts with the TSP solution and involves two sequential greedy expansion loops. In the first loop, only the CMO Type - I is explored, and in the second loop, both CMO Type II and III are explored in parallel. In a greedy iteration, the gain of executing a CMO defined as  
\begin{equation}\label{Eq:CMOGain}
    \Delta_G = \hat{J}(\bar{\Xi})-\hat{J}(\bar{\Xi}')
\end{equation} 
is explored where $\bar{\Xi}$ is the current visiting sequence and $\bar{\Xi}'$ is the modified visiting sequence.

\begin{algorithm}
\caption{Greedy Construction of a Visiting Sequence}
\label{Alg:Greedy_Exploration_Algorithm}
    \begin{algorithmic}[1]
    \State{\bf Input}: Network topology: $\mathcal{G}=(\mathcal{V},\mathcal{E})$
    \State $\bar{\Xi}\leftarrow \bar{\Xi}_{TSP} = \{$TSP solution for $\mathcal{G}=(\mathcal{V},\mathcal{E})$ using \cite{Kirk2020}$\}$;
    \State{$\Delta_G \leftarrow \epsilon\,$}; \Comment{$\epsilon$ is an infinitely small positive number.}
    \While{$\Delta_G \geq \epsilon$} \Comment{First greedy loop.}
        \State{$\bar{\Xi}' \leftarrow \underset{\bar{\Xi}'}{\argmax}\ \Delta \hat{J}(\bar{\Xi}'\vert \bar{\Xi}, \mbox{I})$};
        \State{$\Delta_G \leftarrow \hat{J}(\bar{\Xi})-\hat{J}(\bar{\Xi}')$};
        \State{$\bar{\Xi} \leftarrow \bar{\Xi}'$};
    \EndWhile
    \State{$\Delta_G \leftarrow \epsilon$};
    \While{$\Delta_G \geq \epsilon$} \Comment{Second greedy loop.}
        \State{$\bar{\Xi}'\leftarrow \underset{\bar{\Xi}',\,Y\in\{\mbox{II,III}\}}{\argmax} \Delta \hat{J}(\bar{\Xi}'\vert\bar{\Xi}, Y)$}
        \State{$\Delta_G \leftarrow \hat{J}(\bar{\Xi})-\hat{J}(\bar{\Xi}')$};
        \State{$\bar{\Xi} \leftarrow \bar{\Xi}'$};
    \EndWhile
    \State{\bf Return: $\bar{\Xi}_G \leftarrow \bar{\Xi}$} \Comment{Greedily constructed cycle.}
\end{algorithmic}
\end{algorithm}

\begin{lemma}\label{Lm:Convergence}
The greedy algorithm given in Alg. \ref{Alg:Greedy_Exploration_Algorithm} converges after a finite number of iterations.
\end{lemma}
\begin{proof}
The first greedy loop executes the CMO Type - I iteratively. Since it attempts to replace edges of the TSP solution with alternative fastest paths, it will only run for at most $\vert \Xi_{TSP} \vert$ iterations. 
To prove the convergence of the second greedy loop, we use the fact that $\hat{J}(\bar{\Xi})$ is lower bounded: $\hat{J}(\bar{\Xi}) \geq 0$. Note also that each greedy iteration maintains the condition $\Delta_G\geq\epsilon$, which implies that $\hat{J}(\bar{\Xi})$ decreases monotonically over the greedy iterations. Hence, it is clear that $\Delta_G \rightarrow 0$ as greedy iterations go on, and thus the second greedy loop will converge.
\end{proof}

\paragraph*{\textbf{Simulation Results}} 
We create random persistent monitoring problems by randomly generating network topologies together with target parameters. Such a PM setup is shown in Fig. \ref{Fig:GreedyAlgorithmExample1}(a). Figures \ref{Fig:GreedyAlgorithmExample1}(b)-(e) show the evolution of the greedy visiting sequence $\bar{\Xi}$ and its cost (in terms of the metric $\hat{J}(\bar{\Xi})$) over three consecutive greedy iterations. Another two randomly generated persistent monitoring problems together with their respective high-performing greedily constructed visiting sequences are shown in Figs. \ref{Fig:GreedyAlgorithmExample2} and \ref{Fig:GreedyAlgorithmExample3}.  

\begin{figure*}[h!]
\centering
\begin{subfigure}[t]{0.40\columnwidth}
\includegraphics[width=\columnwidth]{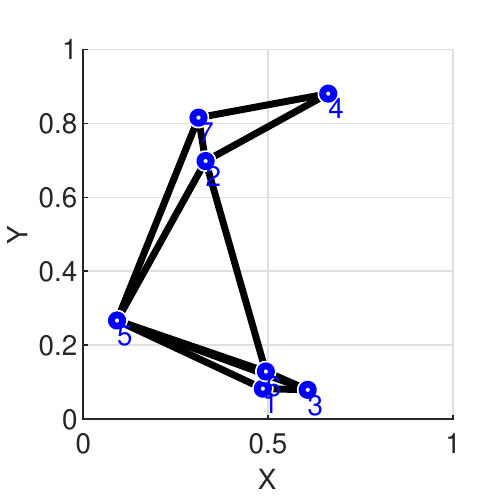}
\caption{Graph $\mathcal{G}=(\mathcal{V},\mathcal{E})$}
\end{subfigure}%
\begin{subfigure}[t]{0.40\columnwidth}
\includegraphics[width=\columnwidth]{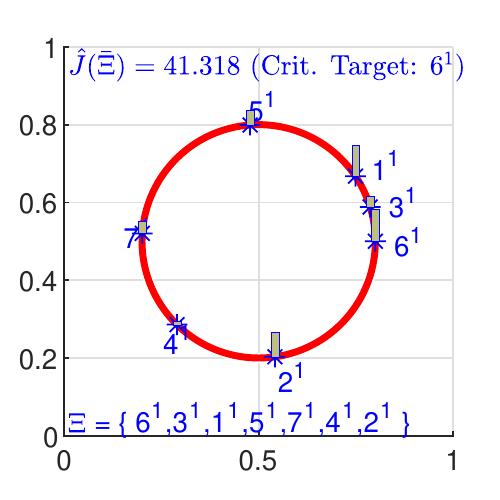}
\caption{Initial Cycle: $\bar{\Xi}_{TSP}$}
\end{subfigure}
\begin{subfigure}[t]{0.40\columnwidth}
\includegraphics[width=\columnwidth]{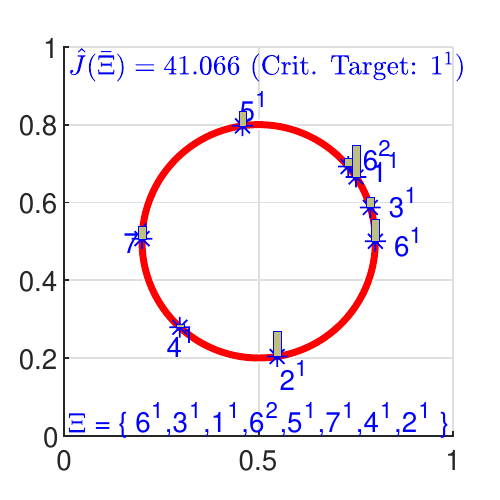}
\caption{Intermediate Cycle 1}
\end{subfigure}
\begin{subfigure}[t]{0.40\columnwidth}
\includegraphics[width=\columnwidth]{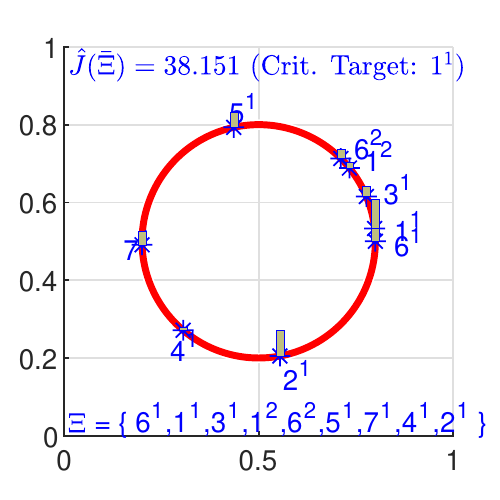}
\caption{Intermediate Cycle 2}
\end{subfigure}
\begin{subfigure}[t]{0.40\columnwidth}
\includegraphics[width=\columnwidth]{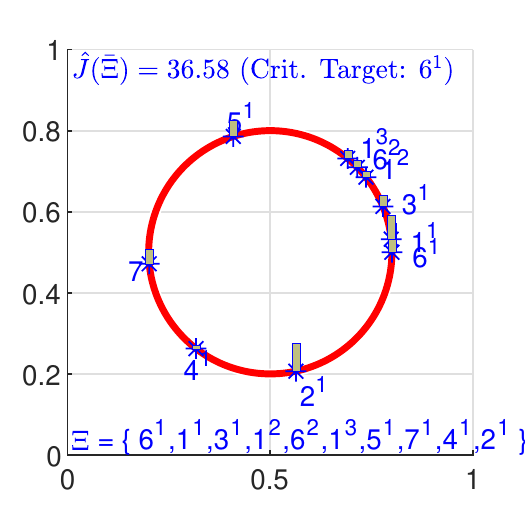}
\caption{Greedy Cycle: $\bar{\Xi}_G$}
\end{subfigure}
\caption{Example 1: The Greedy Cycle Construction Process.}
\label{Fig:GreedyAlgorithmExample1}
\end{figure*}

\begin{figure}[h!]
\centering
\begin{subfigure}[t]{0.48\columnwidth}
\includegraphics[width=\columnwidth]{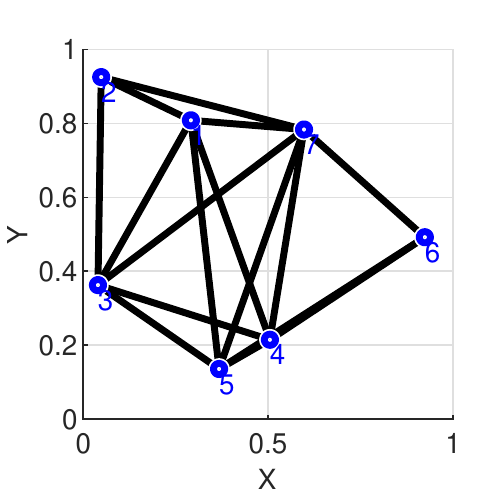}
\caption{Graph $\mathcal{G}=(\mathcal{V},\mathcal{E})$}
\end{subfigure}
\begin{subfigure}[t]{0.48\columnwidth}
\includegraphics[width=\columnwidth]{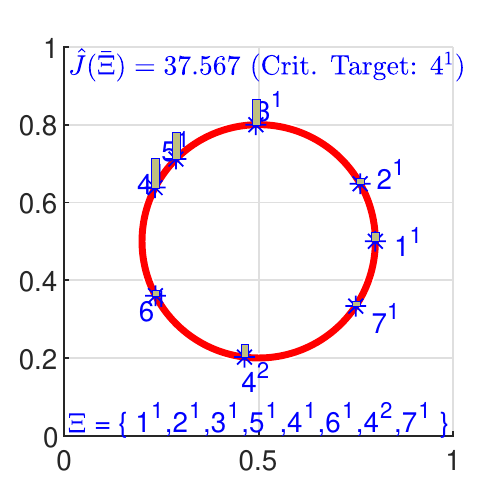}%
\caption{Greedy Cycle: $\bar{\Xi}_G$}
\end{subfigure}
\caption{Example 2: A Constructed Greedy Cycle.}
\label{Fig:GreedyAlgorithmExample2}
\end{figure}

\begin{figure}[h!]
\centering
\begin{subfigure}[t]{0.48\columnwidth}
\includegraphics[width=\columnwidth]{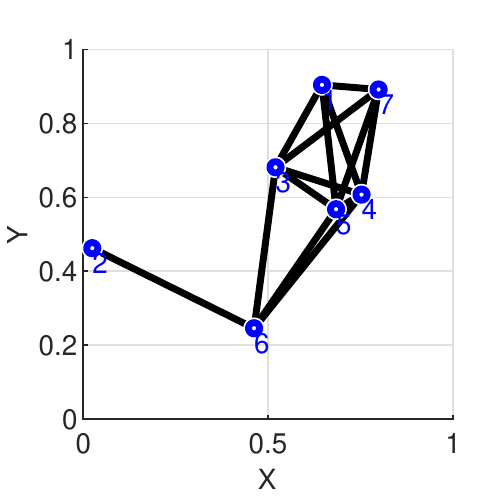}
\caption{Graph $\mathcal{G}=(\mathcal{V},\mathcal{E})$}
\label{fig:traj_simulation}%
\end{subfigure}
\begin{subfigure}[t]{0.48\columnwidth}
\includegraphics[width=\columnwidth]{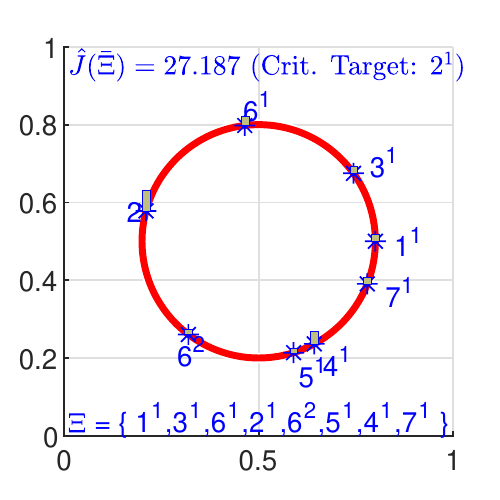}
\caption{Greedy Cycle: $\bar{\Xi}_G$}
\label{fig:traj_simulation2}%
\end{subfigure}
\caption{Example 3: A Constructed Greedy Cycle.}
\label{Fig:GreedyAlgorithmExample3}
\end{figure}

\section{Extension to Multi-Agent Problems}
\label{sec:multi_agent_extension}

In this section, we extend the developed single-agent persistent monitoring solution to handle multi-agent systems. Let us denote the set of agents as $\mathcal{A} = \{1,2,\ldots,N\}$. The key idea here is to partition the target network $\mathcal{G}$ into $N$ sub-graphs and then assign individual agents to each of those sub-graphs. This ``divide and conquer'' approach was motivated by two main reasons: 
(i) to uphold the ``no target sharing'' assumption made early on (due to the fact that sharing is ineffective  \cite{Welikala2020P7,Welikala2019P3}) and 
(ii) to maintain the applicability of the developed single-agent persistent monitoring solution in previous sections. The main steps of the proposed multi-agent persistent monitoring solution are outlined in Alg. \ref{Alg:OverallSolution}. 
In it, to execute Steps 3 and 5, we respectively use the techniques developed in Sections \ref{Sec:GreedyExplorations} and \ref{sec:multiple_visits} (i.e., Alg. \ref{Alg:Greedy_Exploration_Algorithm} and \ref{alg:compute_visitation_times_multiple_visits}). The details of (the remaining) Steps 2 and 4 are provided in the following two subsections. 

\begin{algorithm}[!h]
\caption{The multi-agent persistent monitoring solution.}
\begin{algorithmic}[1]
\State \textbf{Inputs}: The network of targets $\mathcal{G}=(\mathcal{V},\mathcal{E})$ and agents $\mathcal{A}$.
\State Partition the given graph $\mathcal{G}$ into $N$ sub-graphs $\{\mathcal{G}^a\}_{a\in \mathcal{A}}$.
\State Find a high-performing cycle $\bar{\Xi}^{a}$ on each sub-graph $\mathcal{G}^a$.
\State Refine the sub-graphs and respective cycles: $\{\mathcal{G}^a,\bar{\Xi}^a\}_{a\in\mathcal{A}}$. 
\State Compute the optimal dwelling sequence $\mathcal{T}^a$ corresponding to each visiting sequence $\bar{\Xi}^{a}$.
\end{algorithmic}\label{Alg:OverallSolution}
\end{algorithm}

\subsection{Spectral Clustering Based Graph Partitioning}

In order to partition the graph $\mathcal{G}=(\mathcal{V},\mathcal{E})$, we use the well-known \emph{spectral clustering} technique \cite{Luxburg2007} mainly due to its advantages of: (i) simple implementation, (ii) efficient solution and (iii) better results compared to conventional partitioning techniques like the $k$-means algorithm \cite{Luxburg2007}. In particular, the spectral clustering method derives the graph partitions based on a set of inter-target \emph{similarity} values $\{s_{ij}\in \mathbb{R}_{\geq0}:i,j\in\mathcal{V}\}$ so two targets that have high similarity will end up belonging to the same partition and two targets with low similarity to different ones. 


\paragraph*{\textbf{Similarity Values}}
In persistent monitoring, a similarity value $s_{ij}$ should represent the effectiveness of covering both targets $i$ and $j$ in $\mathcal{V}$ using a single agent \cite{Welikala2019Ax2}. Specifically, we define $s_{ij}$ using the \emph{Gaussian Similarity Function} 
\begin{equation}
\label{Eq:DisparityToSimilarity}
s_{ij} = \exp{\left(-\frac{\vert d(i,j) \vert^2 }{2\sigma^2}\right)}.
\end{equation} 
where $d:\mathcal{V}\times\mathcal{V}\rightarrow \mathbb{R}_{\geq 0}$ is a \emph{disparity metric} between two targets and $\sigma$ is a parameter that controls how rapidly the similarity $s_{ij}$ falls off with the disparity $d(i,j)$. According to \eqref{Eq:DisparityToSimilarity}, note that the similarity and disparity metrics are inversely related. We next focus on defining an appropriate disparity metric for the considered persistent monitoring paradigm.

\begin{remark}\label{Rm:InaccurateDisparityMetrics}
As a candidate for the disparity metric $d(i,j)$, neither using the physical distance nor the shortest path distance (between the targets $i$ and $j$) provides a reasonable characterization to the underlying persistent monitoring aspects of the problem - as such metrics disregard target parameters as well as agent behaviors when monitoring targets.
\end{remark}

Considering the above remark, we propose a novel disparity metric named \emph{covering cycle cost} (CCC): 
\begin{equation}\label{Eq:DisparityMetric}
d(i,j) = \min_{\bar{\Xi}:\,i,j \in \bar{\Xi}} \hat{J}(\bar{\Xi}),   
\end{equation} 
with $\hat{J}(\bar{\Xi})$ defined as in \eqref{Eq:CycleMetricGlobal}. The intuition behind this metric is that, given two targets, the similarity is given by the lowest value of the lower bound (that serves as a proxy for the actual post) among all the closed path that contains these two targets.
We also name the $\arg\min$ of \eqref{Eq:DisparityMetric} as the \emph{optimal covering cycle} (OCC) and denote it as $\bar{\Xi}^*_{ij}$. Simply, the OCC $\bar{\Xi}_{ij}^*$ is the best way to cover targets $i,j\in\mathcal{V}$ in a single cycle so that the corresponding cycle metric $\hat{J}(\cdot)$ is minimized. Therefore, if the CCC value is higher for a certain target pair, it implies that it is not effective to cover both those targets in a cycle (in other words, by a single agent). Thus, it is clear that the disparity metric $d(i,j)$ defined in \eqref{Eq:DisparityMetric} provides a good characterization of the underlying persistent monitoring aspects of the problem.

Due to the combinatorial nature of the computation of the metric in \eqref{Eq:DisparityMetric}, we use a greedy technique to obtain a sub-optimal solution for it. In particular, given a target pair $i,j\in\mathcal{V}$, we start with a candidate cycle $\bar{\Xi} = \{i\}$ which is then \emph{iteratively expanded} (by adding external targets in $\mathcal{V}\backslash\bar{\Xi}$) in a greedy manner (using the greedy expansion approach similar to the ones described in Sec. \ref{Sec:GreedyExplorations}) until it includes the target $j$. The terminal state of the candidate cycle $\bar{\Xi}$ is then considered as the OCC $\bar{\Xi}_{ij}^*$. The same procedure is next used to greedily construct the OCC $\bar{\Xi}_{ji}^*$. Finally, the corresponding CCC value, i.e., the disparity metric $d(i,j)$ in \eqref{Eq:DisparityMetric} is estimated as: $d(i,j) = \frac{1}{2}(\hat{J}(\bar{\Xi}_{ij}^*) + \hat{J}(\bar{\Xi}_{ji}^*))$.

Note, however, that for the computation of the similarity metric, the set of targets that will be part of $\Xi$ is not defined a priori. Thus, the greedy expansion process considered here is slightly different from the one in Sec. \ref{Sec:GreedyExplorations}. Therefore, for the sake of completeness, we now provide the details of the iterative greedy cycle expansion mechanism considered in this section. Take $\bar{\Xi}$ as the current version of the candidate cycle in a certain greedy iteration. Here $\bar{\xi}$ (similarly to $\xi$) represents the corresponding sequence of edges. Note however that, unlike $\xi$, the sequence $\bar{\Xi}$ may not contain all the targets in $\mathcal{V}$ (i.e., $\vert \mathcal{V}\backslash\bar{\Xi}\vert > 0$). There are three types of cycle expansion operations (CEOs) as shown in Fig. \ref{Fig:CEOs} that can be used to expand the current cycle $\bar{\Xi}$ so that it includes an external target $i\in\mathcal{V}\backslash\bar{\Xi}$. 

\paragraph*{\textbf{CEO Type - I}} Replace an edge $(l,j)\in\bar{\xi}$ with two edges: $\{(l,i),(i,j)\}$.

\paragraph*{\textbf{CEO Type - II}} Select a target $j\in\bar{\Xi}$ and inset two target $\{i,j\}$ following $j$ in the sequence $\Xi$.

\paragraph*{\textbf{CEO Type - III}} Select two targets $l,j\in\bar{\Xi}$ such that removing all the intermediate targets between them will not reduce the number of distinct targets in $\bar{\Xi}$.Then replace those intermediate targets with the external target $i$.

\begin{figure}[!h]
    \centering
    \includegraphics[width=0.8\columnwidth]{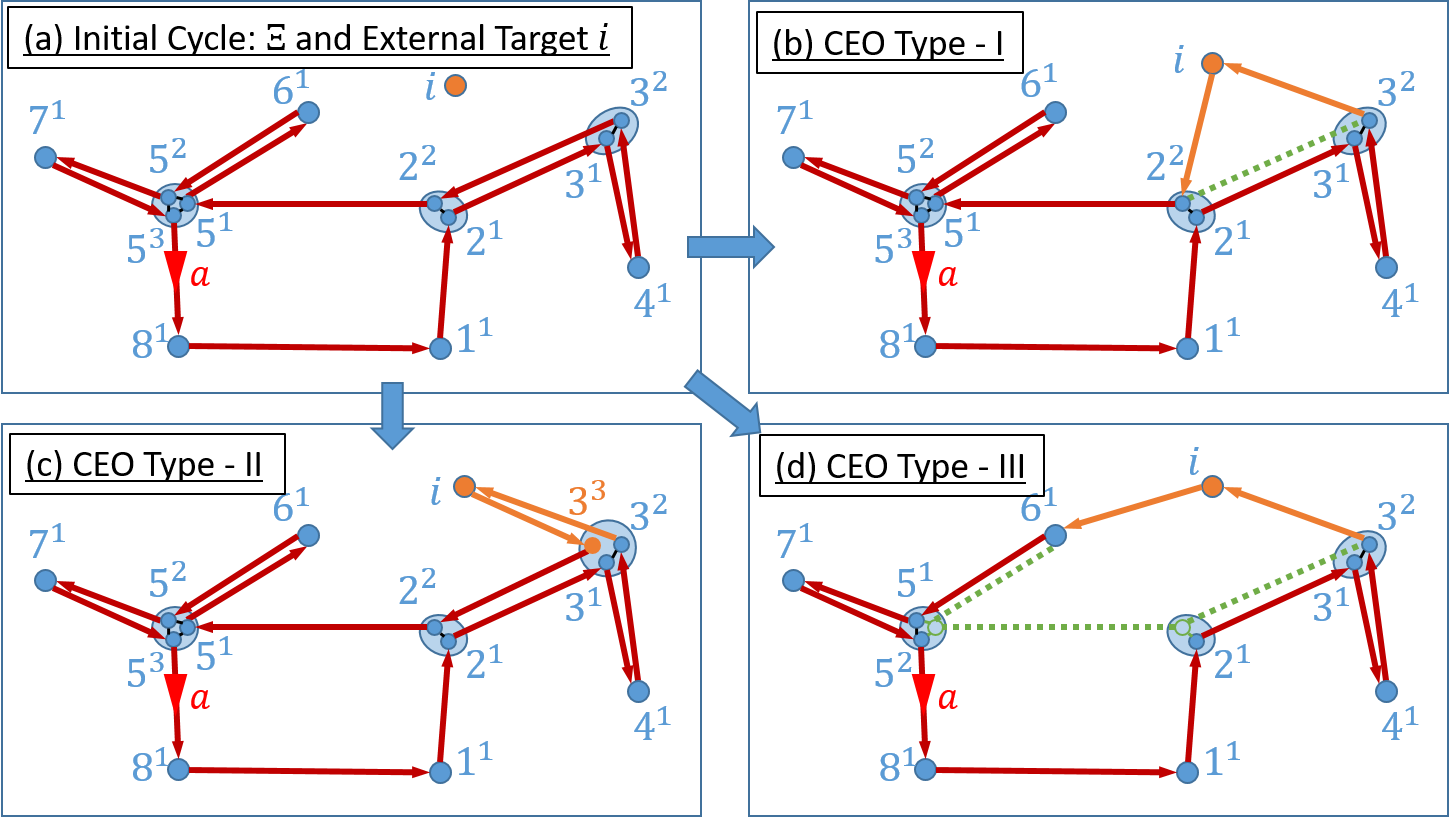}
    \caption{Three types of cycle expansion operations (CEOs).}
    \label{Fig:CEOs}
\end{figure}

In this greedy algorithm, the gain of executing a CEO Type - $Y$ where $Y\in\{\mbox{I, II, III}\}$ defined as
\begin{equation}\label{Eq:CEOGain}
    \Delta \hat{J}(\bar{\Xi}',\,j\vert \bar{\Xi},Y) = \hat{J}(\bar{\Xi}) - \hat{J}(\bar{\Xi}') 
\end{equation}
is exploited to chose the most suitable: 
(i) CEO type $Y$, 
(ii) external target $j$ and 
(iii) expanded cycle $\bar{\Xi}'$ (here, $\bar{\Xi}$ stands for the current cycle). The exact form of the greedy algorithm is provided in Alg. \ref{Alg:DisparityValues}. In all, we use Alg. \ref{Alg:DisparityValues} together with \eqref{Eq:DisparityToSimilarity} to compute the similarity values: $s_{ij},\,\forall i,j\in\mathcal{V}$ between different targets in the network $\mathcal{G}$.

\begin{algorithm}[!h]
\caption{The greedy cycle expansion algorithm for the computation of disparity values $\{d(i,j): i,j\in\mathcal{V}\}$}
\begin{algorithmic}[1]
\State \textbf{Input}: Network of targets $\mathcal{G} = (\mathcal{V},\mathcal{E})$
\State $d(i,j) \leftarrow 0,\,\forall i,j \in \mathcal{V}$;
\For{$i \in \mathcal{V}$} 
\Comment{For each start node.}
    \State $\bar{\Xi} \leftarrow \{i\}$;  \Comment{Initial cycle.} 
    \While{$\vert\mathcal{V}\backslash \bar{\Xi}\vert > 0$} 
    \Comment{Add all external targets.}
        \State $[\bar{\Xi}',\, j] \leftarrow \underset{\bar{\Xi}',\, j\in\mathcal{V}\backslash\bar{\Xi},\, Y\in\{\mbox{I,\,II,\,III}\}}{\arg\max}\,\Delta \hat{J}(\bar{\Xi}',\, j\vert \bar{\Xi},Y)$;
        \Comment{The expanded cycle and the added external target.}
        \State $\bar{\Xi}' \leftarrow \underset{\bar{\Xi}, Y\in\{\mbox{II,\,III}\}}{\argmax} \Delta \hat{J}(\bar{\Xi} \vert \bar{\Xi}',Y)$ 
        \Comment{An optional further refinement for $\bar{\Xi}'$ step based on CMOs.}
        \State $d(i,j) = d(i,j) + \frac{1}{2}\hat{J}(\bar{\Xi}')$;
        \State $\bar{\Xi}\leftarrow \bar{\Xi}'$ \Comment{Update the current cycle.}
    \EndWhile
\EndFor
\State \textbf{Return} $\{d(i,j):\ i,j \in \mathcal{V}\};$
\end{algorithmic}\label{Alg:DisparityValues}
\end{algorithm}

\paragraph*{\textbf{Spectral Clustering Algorithm}}
In spectral clustering, the \emph{Weighted Adjacency Matrix} $W$, the \emph{Degree Matrix} $D$, and the \emph{Laplacian Matrix} $L$ plays a main role \cite{Luxburg2007}. In our case, $W=S=[s_{ij}]_{(i,j)\in\mathcal{V}}$ (i.e., the \emph{Similarity Matrix}), $D = diag(d_1,d_2,\cdots,d_M)$ where $d_i = \sum_{j=1}^M s_{ij}$, and $L = D - W$. Similarly to \cite{Welikala2019Ax2}, we specifically use the normalized spectral clustering technique proposed in \cite{ShiMalik2000} where the normalized Laplacian $L_{rw}=D^{-1}L$ is used instead of $L$.

Algorithm \ref{Alg:SpectralClusteringAlgorithm} outlines the used normalized spectral clustering method. It gives the target clusters $\mathcal{V}^1,\mathcal{V}^2,\ldots,\mathcal{V}^N$ where each $\mathcal{V}^a,\,a\in\mathcal{A}$ is then used to form a \emph{sub-graph} $\mathcal{G}^a = (\mathcal{V}^a,\mathcal{E}^a)$ out of the given graph $\mathcal{G}=(\mathcal{V},\mathcal{E})$ by selecting $\mathcal{E}^a\subseteq \mathcal{E}$ as the set of intra-cluster edges taken from $\mathcal{E}$. Note that the set of inter-cluster edges (i.e., $\mathcal{E}\backslash \cup_{a\in\mathcal{A}} \mathcal{E}^a$) are now not included in any one of these sub-graphs.

\begin{algorithm}[!h]
\caption{Normalized Spectral Clustering Algorithm \cite{ShiMalik2000}}
\begin{algorithmic}[1]
\State \textbf{Input}: Normalized Laplacian $L_{rw}$.
\State Compute first $N$ eigenvectors of $L_{rw}$ as: $u_1,u_2,\cdots,u_N$.
\State Form $U \in \mathbb{R}^{M \times N}$ using $u_1,u_2,\cdots,u_N$ as its columns.
\State For $i = 1,\cdots,M,$ let $y_i \in \mathbb{R}^N$ be the $i$\textsuperscript{th} row of $U$.
\State Cluster the data points $\{y_i\}_{i=1,\cdots,M}$ using the k-means algorithm into $N$ clusters as: $C_1,C_2, \cdots, C_N$.
\State \textbf{Return} $\{\mathcal{V}^1,\mathcal{V}^2,\ldots,\mathcal{V}^N\}$ with each $\mathcal{V}^{i} = \{j: y_j \in C_i\}$
\end{algorithmic}\label{Alg:SpectralClusteringAlgorithm}
\end{algorithm}

Once the sub-graphs are formed, we follow Step 4 of Alg. \ref{Alg:OverallSolution} by executing the greedy high-performing cycle construction procedure (i.e., Alg. \ref{Alg:Greedy_Exploration_Algorithm}) for each sub-graph. The resulting visiting sequence on a sub-graph $\mathcal{G}^a$ is denoted as $\bar{\Xi}^a$ and is assumed to be assigned (arbitrarily) to an agent $a\in \mathcal{A}$.

\begin{figure*}[h!]
\centering
\begin{subfigure}[t]{0.28\columnwidth}
\includegraphics[width=\columnwidth]{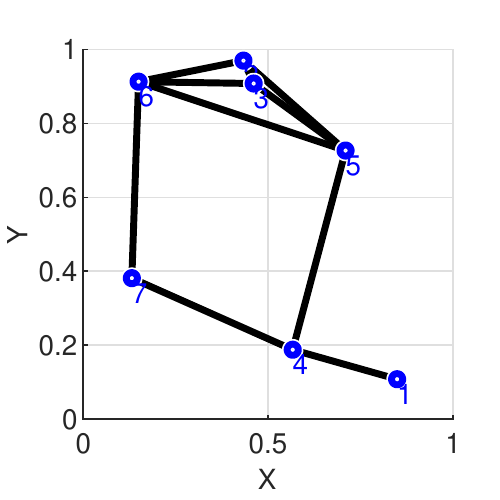}
\caption{Graph $\mathcal{G}$}
\end{subfigure}%
\begin{subfigure}[t]{0.28\columnwidth}
\includegraphics[width=\columnwidth]{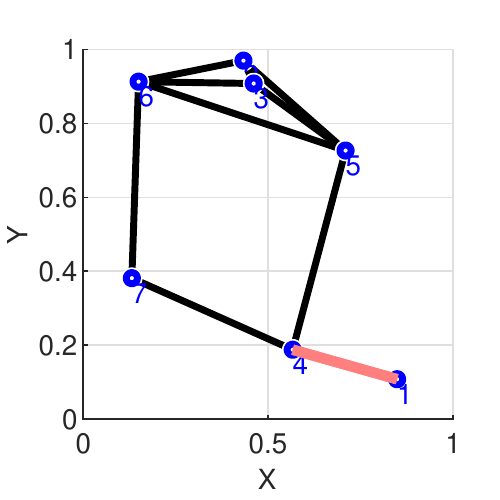}
\caption{$\hat{J}(\bar{\Xi}') = 11.3$}
\end{subfigure}
\begin{subfigure}[t]{0.28\columnwidth}
\includegraphics[width=\columnwidth]{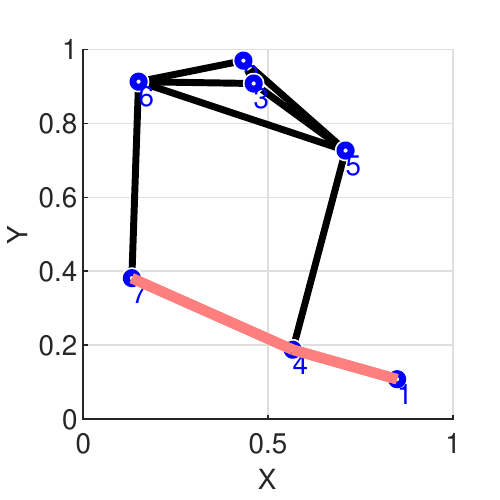}
\caption{$\hat{J}(\bar{\Xi}') = 19.5$}
\end{subfigure}
\begin{subfigure}[t]{0.28\columnwidth}
\includegraphics[width=\columnwidth]{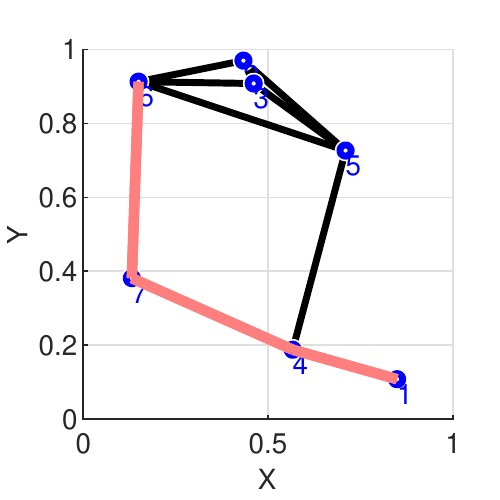}
\caption{$\hat{J}(\bar{\Xi}') = 42.1$}
\end{subfigure}
\begin{subfigure}[t]{0.28\columnwidth}
\includegraphics[width=\columnwidth]{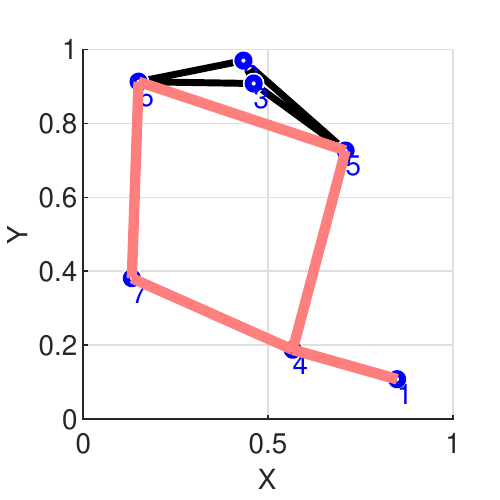}
\caption{$\hat{J}(\bar{\Xi}') = 46.6$}
\end{subfigure}
\begin{subfigure}[t]{0.28\columnwidth}
\includegraphics[width=\columnwidth]{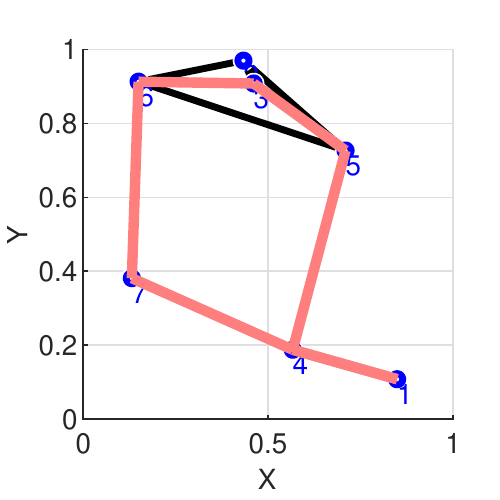}
\caption{$\hat{J}(\bar{\Xi}') = 47.5$}
\end{subfigure}
\begin{subfigure}[t]{0.28\columnwidth}
\includegraphics[width=\columnwidth]{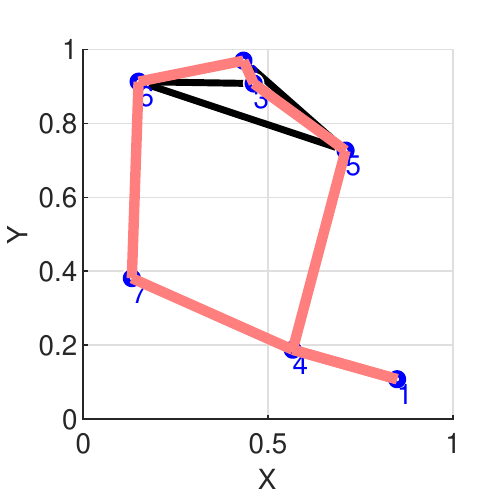}
\caption{$\hat{J}(\bar{\Xi}') = 49.1$}
\end{subfigure}
\caption{Greedy cycle expansion process for the computation of disparity values w.r.t. target 1: $\{d(1,j):j\in\mathcal{V}\}$ using Alg. \ref{Alg:DisparityValues}. The red contours in (b)-(g) show the expanded cycle $\bar{\Xi}'$ after executing Steps 5-9 of Alg. \ref{Alg:DisparityValues} with $i=1$.}
\label{Fig:GreedyCycleExpansions}
\end{figure*}

\begin{figure*}[htp!]
\centering
\begin{subfigure}[t]{0.4\columnwidth}
\includegraphics[width=\columnwidth]{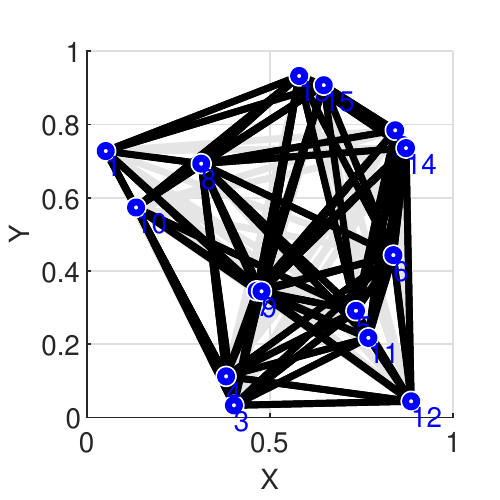}
\caption{Graph $\mathcal{G}$}
\end{subfigure}%
\begin{subfigure}[t]{0.4\columnwidth}
\includegraphics[width=\columnwidth]{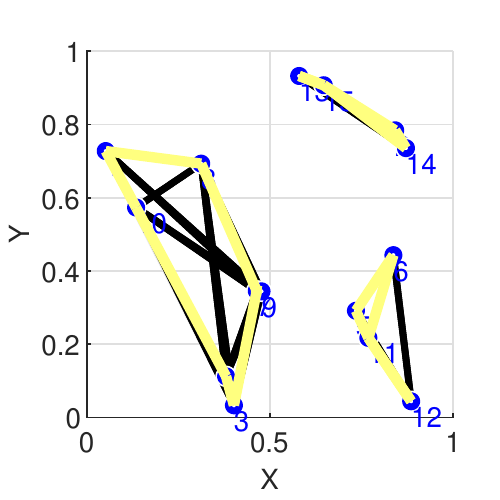}
\caption{$(\mathcal{G}^a,\bar{\Xi}^a): a\in \mathcal{A}$}
\end{subfigure}
\begin{subfigure}[t]{0.4\columnwidth}
\includegraphics[width=\columnwidth]{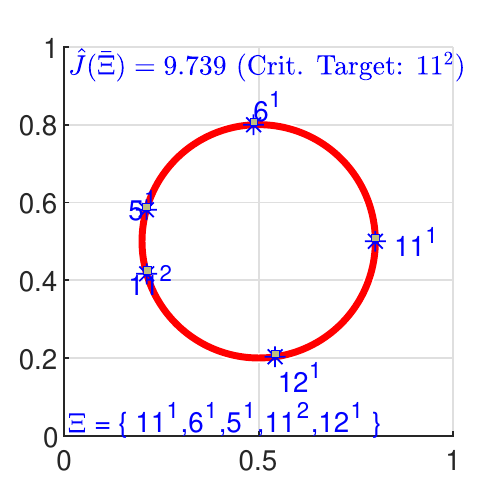}
\caption{Agent 1: $\bar{\Xi}^1$}
\end{subfigure}
\begin{subfigure}[t]{0.4\columnwidth}
\includegraphics[width=\columnwidth]{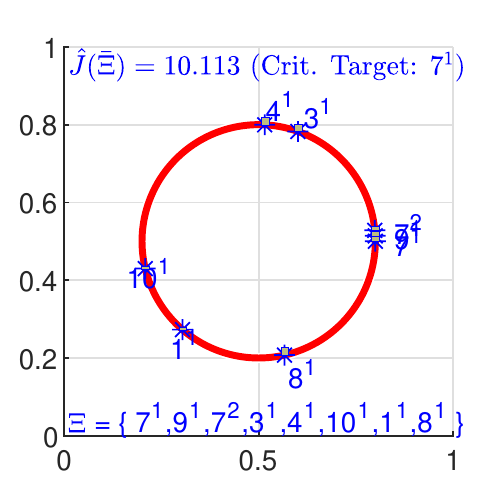}
\caption{Agent 2: $\bar{\Xi}^2$}
\end{subfigure}
\begin{subfigure}[t]{0.4\columnwidth}
\includegraphics[width=\columnwidth]{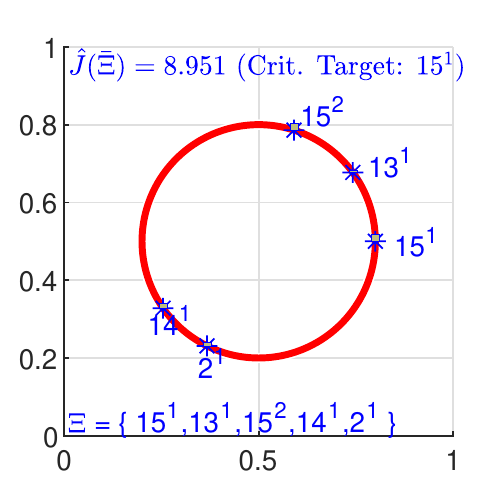}
\caption{Agent 3: $\bar{\Xi}^3$}
\end{subfigure}
\caption{Clustering results and the greedy cycles constructed in each sub-graph for individual agents.}
\label{Fig:ClusteringResults1}
\end{figure*}

\begin{figure}[htp!]
\centering
\begin{subfigure}[t]{0.48\columnwidth}
\includegraphics[width=\columnwidth]{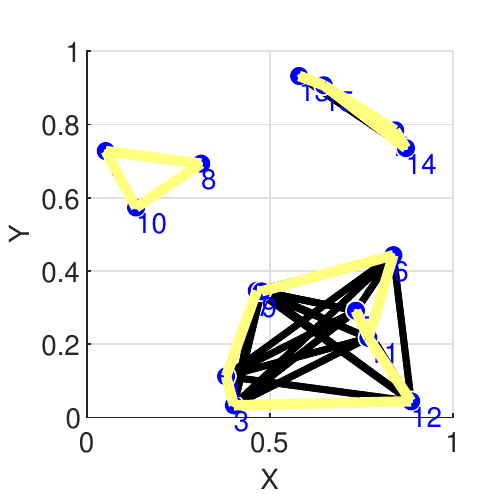}
\caption{$(\mathcal{G}^a,\bar{\Xi}^a): a\in \mathcal{A}$}
\end{subfigure}
\begin{subfigure}[t]{0.48\columnwidth}
\includegraphics[width=\columnwidth]{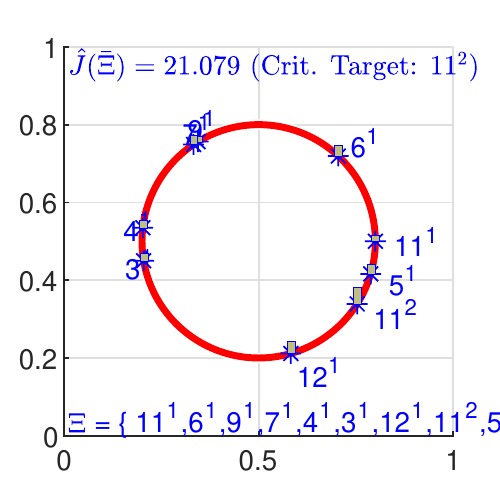}
\caption{Agent 1: $\bar{\Xi}^1$}
\end{subfigure}%
\caption{Clustering results (for the graph in Fig. \ref{Fig:ClusteringResults1}(a)) when the shortest path distance is used as the disparity metric.}
\label{Fig:ClusteringResults2}
\end{figure}

\paragraph*{\textbf{Simulation Results}}
Figure \ref{Fig:GreedyCycleExpansions} shows a few intermediate results obtained when executing Alg. \ref{Alg:DisparityValues} (to compute disparity values according to \eqref{Eq:DisparityMetric}) for an example target network.

Figure \ref{Fig:ClusteringResults1}(a) shows an example target network with 15 targets. Assuming that three agents $\mathcal{A}=\{1,2,3\}$ are to be deployed to monitor these targets, Fig. \ref{Fig:ClusteringResults1}(b) shows the generated sub-graphs when Alg. \ref{Alg:SpectralClusteringAlgorithm} is used.
The yellow contours indicate the cycles constructed within each sub-graph (using Alg. \ref{Alg:Greedy_Exploration_Algorithm}) for each agent to traverse. Figs. \ref{Fig:ClusteringResults1}(c)-(e) provide details of each agent's cycle. In these figures, notice that $\hat{J}(\bar{\Xi}^1) = 9.739,\ \hat{J}(\bar{\Xi}^2) = 10.113$ and $\hat{J}(\bar{\Xi}^1) = 8.951$, which implies that the worst $L_i$ \eqref{eq:lower_bound_definition} value over the network is $10.113$ (by sub-cycle $\bar{\Xi}_7^1$ of target instant $7^1$ on agents 2's cycle).  

We next use the same problem setup in Fig. \ref{Fig:ClusteringResults1}(a) to highlight the importance of the proposed disparity metric in \eqref{Eq:DisparityMetric}. Figure \ref{Fig:ClusteringResults2}(a) shows the clustering result obtained when the similarity values required in Alg. \ref{Alg:SpectralClusteringAlgorithm} are computed using a shortest path distance based disparity metric (instead of \eqref{Eq:DisparityMetric}). Observing the cycle assigned for the Agent 1 shown in Fig. \ref{Fig:ClusteringResults2}(b), it is evident that now the worst $L_i$ \eqref{eq:lower_bound_definition} value over the network is $21.079$ (i.e., a $108.4\%$ degradation from the use of \eqref{Eq:CEOGain}).

\subsection{Target-exchange scheme (TES) used to refine sub-graphs}

Practically, it is preferred to have the persistent monitoring load \emph{balanced} across all the deployed agents. In other words, we prefer to have sub-graphs $\{\mathcal{G}^a\}_{a\in\mathcal{A}}$ that have a numerically closer set of visiting sequence metrics $\{\hat{J}(\bar{\Xi}^a)\}_{a\in\mathcal{A}}$. As a result of the used disparity metric in \eqref{Eq:DisparityMetric}, the spectral clustering algorithm (Alg. \ref{Alg:SpectralClusteringAlgorithm}) often directly provides such a balanced set of sub-graphs (e.g., like in Fig. \ref{Fig:ClusteringResults1}(b) as opposed to Fig. \ref{Fig:ClusteringResults2}(a)).

However, to further enforce this requirement, we propose a target exchange scheme (TES) between sub-graphs that attempts to balance the sub-graphs by minimizing the metric:
\begin{equation}\label{Eq:TESObjective}
    \hat{J}(\mathcal{G}) = \max_{a\in\mathcal{A}}\hat{J}(\bar{\Xi}^a).
\end{equation}
Note that \eqref{Eq:TESObjective} can be simplified using \eqref{Eq:CycleMetricGlobal}, \eqref{eq:maximum_revisit_time} and  \eqref{eq:lower_bound_definition} into:
\begin{equation}\label{Eq:TESObjectiveSimplified}
\hat{J}(\mathcal{G}) = \max_{a\in\mathcal{A}}\max_{i\in\mathcal{V}^a}L_i\Big(
    \max_{k:i^k \in \tau_{i}}\ w_i^k\Big). 
\end{equation}

According to \eqref{Eq:TESObjectiveSimplified}, note that $\hat{J}(\mathcal{G})$ is determined by a specific (critical) target $i^k \in \Xi^a$ where the critical $i,k,a$ are the optimal arguments of the three optimization problems involved in the R.H.S. of \eqref{Eq:TESObjectiveSimplified}. Analogous to Lemma \ref{Lm:CriticalTargetVisit}, it is clear that the only way to decrease $\hat{J}(\mathcal{G})$ is by modifying the sub-cycle corresponding to this particular critical target instant $i^k\in\bar{\Xi}^a$. Let us denote this critical sub-cycle as $\Xi_i^{k,a}$ and note that it is a segment of the cycle $\bar{\Xi}^a$ constructed on the sub-graph $\mathcal{G}^a$. A such feasible sub-cycle modification is to remove a target $j\in\Xi_i^{k,a}$ from agent $a$'s trajectory (i.e., from sub-graph $\mathcal{G}^a$). Clearly, one of the neighboring agents $b \in \mathcal{A}\backslash\{a\}$ will have to annex this removed target $j$ into its cycle $\bar{\Xi}^b$.

Upon such a target exchange between sub-graphs  $\mathcal{G}^a$ and $\mathcal{G}^b$, typically, both cycles $\bar{\Xi}^a$ and $\bar{\Xi}^b$ should be re-computed (using Alg. \ref{Alg:Greedy_Exploration_Algorithm}) respectively on the updated sub-graphs $\mathcal{G}^a$ and $\mathcal{G}^b$. Since Alg. \ref{Alg:Greedy_Exploration_Algorithm} involves solving a TSP problem, we only use it when the optimal target exchange, i.e., the optimal target $j^* \in \Xi_i^{k,a}$ to remove and the optimal neighbor $b^* \in \mathcal{A}\backslash\{a\}$ to receive, is known. To find this optimal $b^*$ and $j^*$, we again use a computationally efficient greedy scheme as follows. 

First, we estimate the gain of annexing an external target $j$ to a cycle $\bar{\Xi}^{b}$ as $\Delta\hat{J}_A(j,\bar{\Xi}^b)$ where 
\begin{equation}\label{Eq:CycleExpansion}
    \Delta\hat{J}_A(j,\bar{\Xi}^b) =  \max_{\bar{\Xi}^{b'},\ Y\in\{\mbox{I,\,II,\,III}\}}
    \Delta \hat{J}(\bar{\Xi}^{b'}, j \vert \bar{\Xi}^b,Y),
\end{equation}
where  $\Delta\hat{J}(\bar{\Xi}^{b'}, j \vert \bar{\Xi}^b,Y)$ is given in \eqref{Eq:CEOGain} and $\mbox{I,\,II,\,III}\}$ refer to the three CEOs considered in this paper.

Next, we estimate the gain of removing a target $j$ from a cycle $\bar{\Xi}^a$ as $\Delta\hat{J}_R(j,\bar{\Xi}^a)$ where 
\begin{equation}\label{Eq:CycleContraction}
    \Delta\hat{J}_R(j,\bar{\Xi}^a) = \hat{J}(\bar{\Xi}^a) - \hat{J}(\bar{\Xi}^{a\#}),
\end{equation} 
where $\bar{\Xi}^{a\#}$ represents the contracted version of the cycle $\bar{\Xi}^a$. This contracted version is obtained by following the steps: (i) remove the entries of $j$ from $\bar{\Xi}^a$, (ii) bridge the gaps created by this removal using corresponding fastest paths and (iii) refine the obtained cycle (say $\bar{\Xi}^{a1}$) using CMOs \eqref{Eq:CMOGain}. In particular,   
\begin{equation}
    \bar{\Xi}^{a\#} = 
    \underset{\bar{\Xi}^{a'},\,Y\in\{\mbox{II,III}\}}{\argmin}\hat{J}(\bar{\Xi}^{a'}\vert \bar{\Xi}^{a1}, Y).
\end{equation}

The final step of determining the optimal target exchange: $b^*$, $j^*$ exploits $\Delta\hat{J}_A(j,\bar{\Xi}^b)$ and $\Delta\hat{J}_R(j,\bar{\Xi}^a)$ functions defined respectively in \eqref{Eq:CycleExpansion} and \eqref{Eq:CycleContraction}. In particular, $b^*$ and $j^*$ are:
\begin{equation}
    [b^*,j^*] = \underset{b\in\mathcal{A}\backslash \{a\},\ j \in\Xi_i^{a,k}}{\argmax} \Delta\hat{J}_A(j,\bar{\Xi}^b) + \Delta\hat{J}_R(j,\bar{\Xi}^a).
\end{equation}

Continuing the previous discussion, upon finding this optimal target exchange, we modify the sub-graphs $\mathcal{G}^a$ and $\mathcal{G}^{b^*}$ appropriately and use Alg. \ref{Alg:Greedy_Exploration_Algorithm} to re-compute the respective cycles $\bar{\Xi}^a$ and $\bar{\Xi}^{b^*}$ on them. However, we only commit to this target exchange if the resulting $\hat{J}(\mathcal{G})$ in \eqref{Eq:TESObjective} is better (i.e., smaller) than before. Clearly, this target exchange process can be iteratively executed until there is no feasible target exchange that results in an improved $\hat{J}(\mathcal{G})$ metric. 

The proposed TES above is distributed as it only involves one agent (sub-graph) and its neighbors at any iteration. Further, the proposed TES is convergent due to the same arguments made in the proof of Lemma \ref{Lm:Convergence}. Finally, apart from the fact that this TES leads to a balanced set of sub-graphs, it also makes the final set of sub-graphs independent of the clustering parameter $\sigma$ in  \eqref{Eq:DisparityToSimilarity} that we have to choose.

\paragraph*{\textbf{Simulation Results}} Figures \ref{fig:Example1TES} and \ref{fig:Example2TES} illustrate the operation and the advantages of the proposed TES for two different persistent monitoring problems.

\begin{figure}[t]
\centering
\begin{subfigure}[t]{0.48\columnwidth}
\includegraphics[width=\columnwidth]{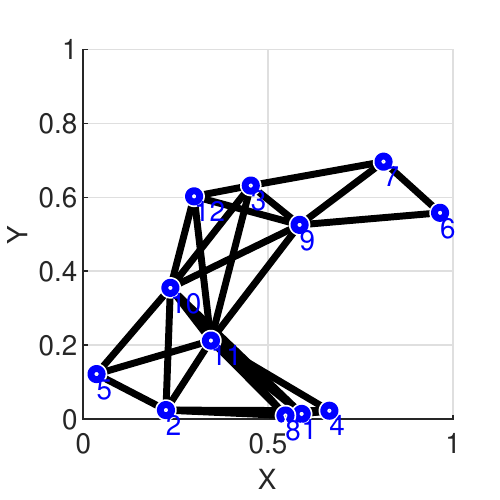}
\caption{Graph $\mathcal{G}=(\mathcal{V},\mathcal{E})$}
\end{subfigure}%
\begin{subfigure}[t]{0.48\columnwidth}
\includegraphics[width=\columnwidth]{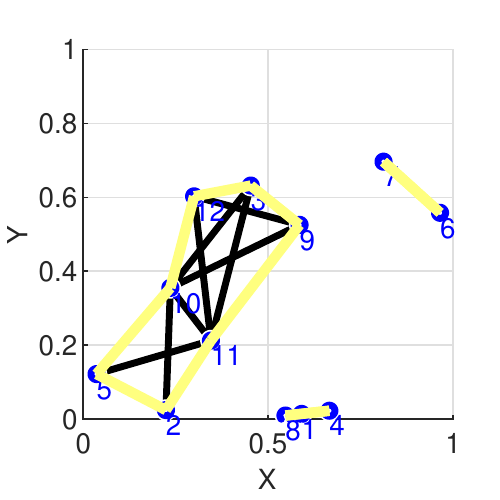}
\caption{Initial sub-graphs and cycles}
\end{subfigure}
\begin{subfigure}[t]{0.48\columnwidth}
\includegraphics[width=\columnwidth]{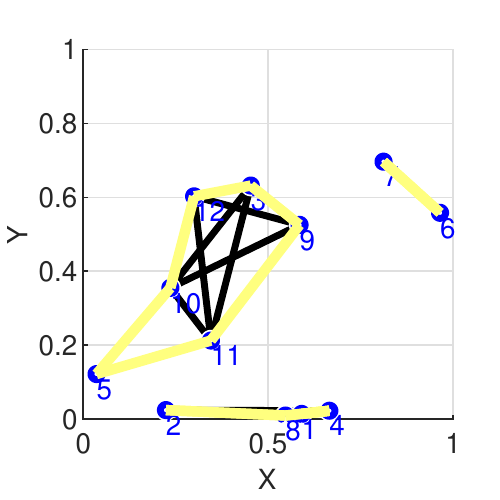}
\caption{After target exchange 1}
\end{subfigure}
\begin{subfigure}[t]{0.48\columnwidth}
\includegraphics[width=\columnwidth]{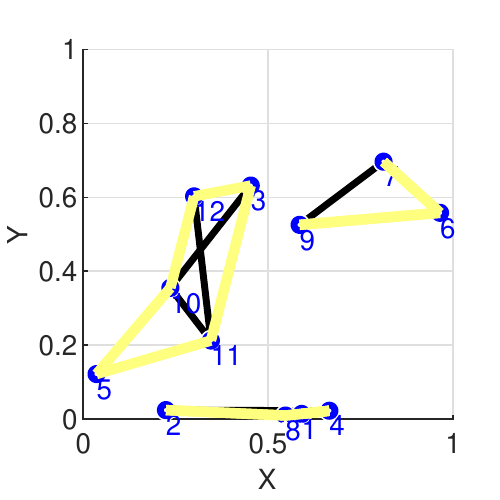}
\caption{After target exchange 2}
\end{subfigure}
\caption{Target Exchange Scheme (TES) Example 1. Initial set of sub-graphs (in (b)): $\hat{J}(\mathcal{G}) = 16.3 = \max\{16.3,\, 6.9,\,9.7\}$. Final set of sub-graphs (in (d)): $\hat{J}(\mathcal{G}) = 11.6 = \max\{11.2,\,11.6,\,10.9\}$ (Balanced and improved by $28.8\%$).}
\label{fig:Example1TES}
\end{figure}

\begin{figure*}[t]
\centering
\begin{subfigure}[t]{0.4\columnwidth}
\includegraphics[width=\columnwidth]{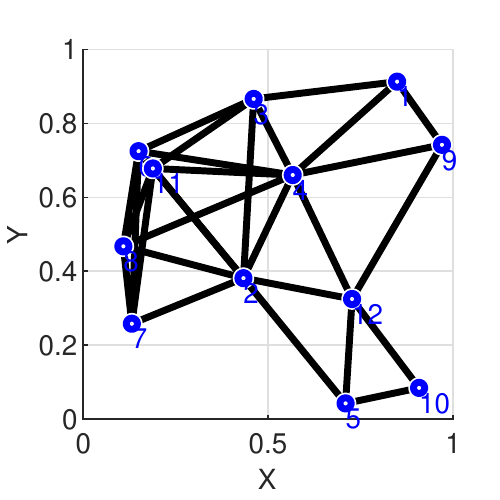}
\caption{Graph $\mathcal{G}=(\mathcal{V},\mathcal{E})$}
\end{subfigure}%
\begin{subfigure}[t]{0.4\columnwidth}
\includegraphics[width=\columnwidth]{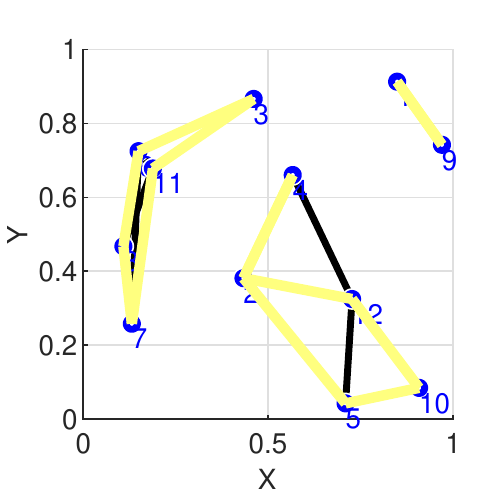}
\caption{Initial sub-graphs}
\end{subfigure}
\begin{subfigure}[t]{0.4\columnwidth}
\includegraphics[width=\columnwidth]{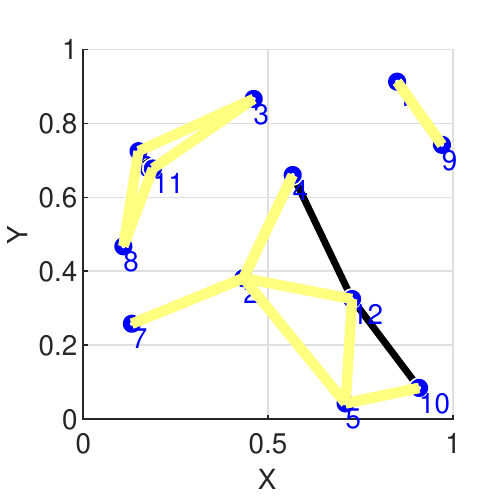}
\caption{After target exchange 1}
\end{subfigure}
\begin{subfigure}[t]{0.4\columnwidth}
\includegraphics[width=\columnwidth]{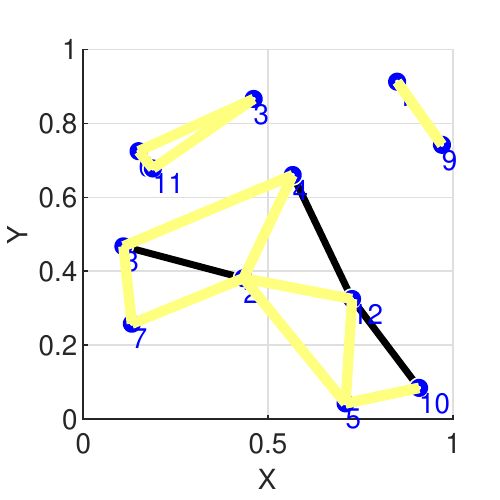}
\caption{After target exchange 2}
\end{subfigure}
\begin{subfigure}[t]{0.4\columnwidth}
\includegraphics[width=\columnwidth]{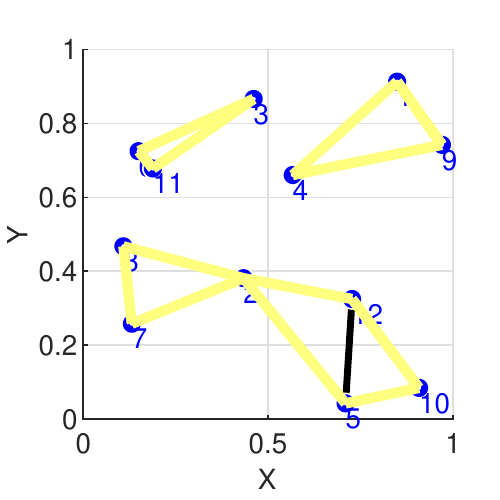}
\caption{After target exchange 3}
\end{subfigure}
\caption{Target Exchange Scheme (TES) Example 2. Initial set of sub-graphs (in (b)): $\hat{J}(\mathcal{G}) = 12.3 = \max\{7.2,\, 12.3,\,4.4\}$. Final set of sub-graphs (in (e)): $\hat{J}(\mathcal{G}) = 8.3 = \max\{7.7,\,8.3,\,5.6\}$ (Balanced and improved by $32.5\%$).}
\label{fig:Example2TES}
\end{figure*}

Finally, we use the randomly generated persistent monitoring problems shown in Fig. \ref{Fig:InitialCondition1} to compare the performance (in terms of the cost $J$ in \eqref{eq:cost}) of persistent monitoring solutions: (i) the minimax approach proposed in this paper (labeled ``Minimax''), (ii) the event-driven receding horizon control approach proposed in \cite{Welikala2020P7} (labeled ``RHC'') and (iii) the threshold-based basic distributed control approach proposed in \cite{Welikala2020P7} (labeled ``BDC''). It should be highlighted that, compared to the minimax solution, both RHC and BDC solutions: (i) are distributed on-line, (ii) can only handle one-dimensional target state dynamics and (iii) have the objective of minimizing the average overall target error covariance observed over a finite period. Despite these structural differences, it is important to note that all three solutions have the same ultimate goal of maintaining the target error covariances as low as possible. 

In particular, Fig. \ref{Fig:InitialCondition1}(a) and \ref{Fig:InitialCondition1}(b) show the initial conditions (i.e., at $t=0$) of the two persistent monitoring problem setups considered. After executing the agent controls given by Minimax, RHC and BDC approaches until $t=500$, the cost value $J$ in \eqref{eq:cost} was evaluated as the maximum target error covariance value observed in the period $t\in[450,500]$. The observed cost values are summarized in Tab. \ref{Tab:ComparisonResults}. The superiority of the Minimax approach proposed in this paper is evident from the reported simulation results in Tab. \ref{Tab:ComparisonResults}. We note however that the other methods in Tab. \ref{Tab:ComparisonResults} were not designed for the specific cost metric we consider in this paper and, to the best of the authors knowledge, there is no algorithm in the scientific literature designed specifically for the cost function we consider there.


\begin{table}[!h]
\centering
\caption{Performance comparison of different agent control methods under different persistent monitoring problem setups}
\label{Tab:ComparisonResults}
\resizebox{0.8\columnwidth}{!}{%
\begin{tabular}{|c|c|c|r|r|}
\hline
\multicolumn{2}{|c|}{\multirow{2}{*}{Cost $J$ in \eqref{eq:cost}}} & \multicolumn{3}{c|}{Agent Control Method} \\ \cline{3-5} 
\multicolumn{2}{|c|}{} & Minimax & \multicolumn{1}{c|}{BDC \cite{Welikala2020P7}} & \multicolumn{1}{c|}{RHC \cite{Welikala2020P7}} \\ \hline
\multirow{2}{*}{\begin{tabular}[c]{@{}c@{}}Problem \\ Configuration\end{tabular}} & 1 & \multicolumn{1}{r|}{\textbf{40.23}} & 54.47 & 56.65 \\ \cline{2-5} 
 & 2 & \multicolumn{1}{r|}{\textbf{46.17}} & 456.55 & 110.87 \\ \hline
\end{tabular}%
}
\end{table}

\begin{figure}[t]
\centering
\begin{subfigure}[t]{0.48\columnwidth}
\includegraphics[width=\columnwidth]{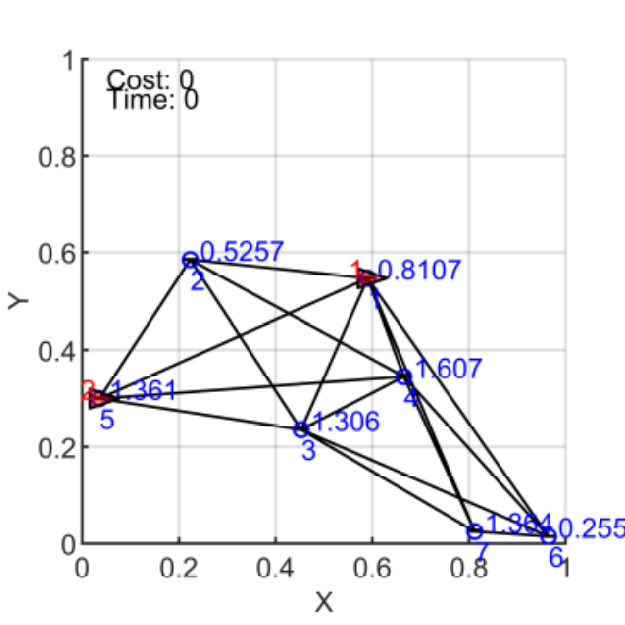}
\caption{Problem Configuration 1}
\end{subfigure}%
\begin{subfigure}[t]{0.48\columnwidth}
\includegraphics[width=\columnwidth]{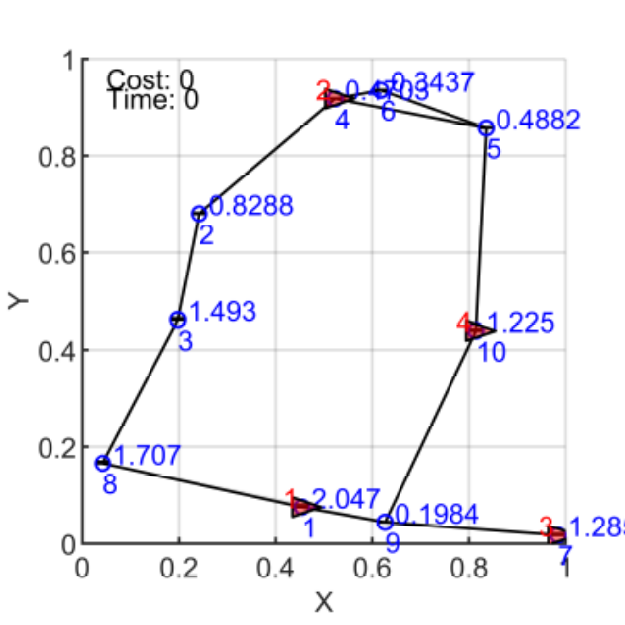}
\caption{Problem Configuration 2}
\end{subfigure}%
\caption{The persistent monitoring problem setups used in the performance comparison (at the initial condition).}
\label{Fig:InitialCondition1}
\end{figure}



\section{Conclusion}
\label{sec:conclusion}
In this paper, we discussed the problem of minimizing the maximum uncertainty of dynamically evolving internal states of targets using multiple mobile sensors. We started from a single agent perspective, presenting provably optimal algorithms for the case where each agent is visited only once. Then we extended this notion to situations where targets are allowed to be visited multiple times in a single period. In order to provide efficient solutions, we developed a greedy scheme to explore the possible visiting sequences. We then extended these results to multi-agent scenarios, assuming that each target would be visited by only one agent. We presented a clustering algorithm to split the targets among different agents. Finally we proposed a distributed scheme for swapping targets between different agents, improving the quality of the initial clustering. Simulation results show the efficiency of our proposed technique.

\appendix
\subsection{Proof of Lemma \ref{lemma:pd_nd_time_derivative}}
\label{app:proof_pd_nd_time_derivative}
Let $\Omega_{i,ss}$ be the solution of the algebraic Riccati equation given by \eqref{eq:dynamics_omega}, when $\eta_i(t)=1$ and $\dot{\Omega}_i=0$, i.e.
\begin{equation}
    \label{eq:dynamics_omega_ss_not_observed}
    A_i\Omega_{i,ss}+\Omega_{i,ss}A_i'+Q_i=0.
\end{equation} Since the system is observable and $Q_i\succ 0$, ${\Omega}_{i,ss}$  is guaranteed to be unique and positive definite \cite{bittanti2012riccati}.

We then define a similar concept to $\Omega_{i,ss}$, but for the case where $\eta_i(t)=0$, $\forall t$. Intuitively, we want to define $\Omega_{i}^\infty$ as being the covariance matrix when the target is never observed. Note that when $A_i$ is unstable, the covariance matrix will diverge thus $\Omega_{i}^\infty$ is not well defined. 
In particular, if $A_i$ is unstable but some of the eigenvalues of $A_i$ are negative, the covariance does not diverge in every direction. Thus, in order to overcome these peculiarities in ``$\Omega_{i}^\infty$", for a given vector $\zeta \in \mathbb{C}^{L_i}$ (i.e. $\zeta$ has the same dimension as the the target state), we define $\zeta^* \Omega_i^\infty \zeta$ as:
\begin{equation}
\label{eq:def_omega_ss}
\zeta^* \Omega_i^\infty \zeta=\lim_{t\rightarrow\infty}\zeta^*\left(\int_0^{t}\exp(A_i\gamma)Q_i\exp(A_i'\gamma)d\gamma\right)\zeta.
\end{equation}
Note that, if $\eta_i(t)=0$, $\lim_{t\rightarrow \infty}\zeta^*\Omega_i(t)\zeta = \zeta^* \Omega_i^\infty \zeta$, independently of the initial condition $\Omega_i(0)$ \cite{bittanti2012riccati}. 

Using Theorem 2 in \cite{dieci1996preserving}, we get that the steady state covariance matrix $\bar{\Omega}_i(t)$ generated by a cyclic schedule where target $i$ observed by some time part of the cycle (but not the entire cycle) is such that:
\begin{equation}
    \label{eq:zeta_ss_without_observation}
    \zeta^*\Omega_{i,ss}\zeta < \zeta^*\bar{\Omega}_i(t)\zeta < \zeta^* \Omega_{i}^\infty \zeta,
\end{equation}
for $\forall \zeta\neq 0 \in \mathbb{R}^{L_i}$ such that $\zeta^* \Omega_{i}^\infty \zeta$ is bounded. Informally, one can think of $\Omega_{i,ss}$ and $\Omega_{i}^\infty$ as being respectively lower and upper bounds on the covariance matrix.

Now that we have defined these lower and upper bounds, we go back to analyze the steady state covariance matrix resulting from an agent trajectory that visits target $i$. First we show that, at steady state, the covariance will increase when the target is not observed. Towards that, we analyze dynamics of the steady state covariance matrix $\dot{\bar{\Omega}}_i$ in time instants where the target is not observed ($\eta_i=0$), thus $\dot{\bar{\Omega}}_i=A_i\bar{\Omega}_i+\bar{\Omega}_iA_i'+Q_i$. If we pick $\zeta$ to be a right eigenvector of $A_i'$ and its corresponding eigenvalue $\lambda$, we get that
\begin{equation}
\label{eq:zeta_proposition_eta_omega_derivative}
\zeta^*(A_i\bar{\Omega}_i+\bar{\Omega}_iA_i'+Q_i) \zeta= \zeta^*\left((\lambda+\lambda^*)\bar{\Omega}_i+Q_i\right)\zeta.
\end{equation}

Note that, if the real part of $\lambda$ is non-negative, then the expression \eqref{eq:zeta_proposition_eta_omega_derivative} is necessarily positive, thus $\zeta^*\dot{\bar{\Omega}}_i(t)\zeta>0$. However, if the real part of $\lambda$ is negative, then \eqref{eq:def_omega_ss} implies that $\zeta^*{\Omega}_i^{\infty}\zeta$ is bounded and from \eqref{eq:dynamics_omega} we get that $(\lambda+\lambda^*)\zeta^*{\Omega}_{i,ss} \zeta + \zeta^*Q_i\zeta = \Omega_{i,ss}G_i\Omega_{i,ss} > 0$. Therefore, using \eqref{eq:zeta_ss_without_observation}, we have 
\begin{equation*}
    \zeta^*\left((\lambda+\lambda^*)\bar{\Omega}_i+Q_i\right)\zeta>(\lambda+\lambda^*)\zeta^*{\Omega}_{i,ss}\zeta + \zeta^*Q_i\zeta > 0.
\end{equation*}

A parallel discussion applies to the eigenvalues of $A_i$, which are the complex conjugates of the eigenvalues of $A_i'$. Since the eigenvectors of $A_i$ and $A_i'$ span $\mathbb{R}^{L_i}$, we have that $\xi$ can be written as a linear combination of these eigenvectors, $\forall \xi \in \mathbb{R}^{L_i}$. Therefore, $\chi^*\bar{\Omega}_i\chi$ is positive for any $\chi\in\mathbb{R}^{L_i}$ whenever $\eta_i=0$ which implies that $\bar{\Omega}_i(t)\succ 0$ if $\eta_i(t)=0$.

We then consider the instants when the target is observed ($\eta_i=1$). Note that $\dot{\Omega}_{i,ss}=0$. Then we define $U=\bar{\Omega}_i-\Omega_{i,ss}$, for which $\dot{\bar{\Omega}}_i=\dot{U}$. Using \eqref{eq:dynamics_omega} and \eqref{eq:dynamics_omega_ss_not_observed}, we get that
\begin{equation}
\begin{aligned}
\dot{U} &= A_i(\bar{\Omega}_i-\Omega_{i,ss})+(\bar{\Omega}_i-\Omega_{i,ss})A_i'-\bar{\Omega}_iG_i\bar{\Omega}_i\\&\qquad+\Omega_{i,ss}G_i\Omega_{i,ss}\\ &=A_iU+UA_i'+\Omega_{i,ss}G_i\Omega_{i,ss}-\bar{\Omega}_iG_i\bar{\Omega}_i.
\end{aligned}
\end{equation}
Additionally, using the fact that $\dot{U}^{-1}=U^{-1}\dot{U}U^{-1}$, we have
\begin{equation*}
\dot{U}^{-1} = (A_i-G_i\Omega_{i,ss})U^{-1}+U^{-1}(A_i-G_i\Omega_{i,ss})'-G_i.
\end{equation*}
We take an eigenvector $\zeta$ of $(A_i-G_i\Omega_{i,ss})'$ and $\lambda$ its corresponding eigenvalue. Note that $A_i'\zeta=\lambda \zeta +(G_i\Omega_{i,ss})'\zeta$. Substituting  $\zeta^*\dot{\Omega}_{i,ss}\zeta=0$ into the Riccati differential equation, we get that
\begin{equation*}
(\lambda+\lambda^*)\zeta^*\Omega_{i,ss}+\zeta^*(Q_i+\Omega_{i,ss}G_i\Omega_{i,ss})\zeta=0,
\end{equation*}
which implies that $(\lambda+\lambda^*)<0$. Computing $\zeta^{*}\dot{U}^{-1}\zeta$, we get
\begin{equation*}
    \zeta^{*}\dot{U}^{-1}\zeta = (\lambda+\lambda^*)\zeta^*U^{-1}\zeta-\zeta^*G_i\zeta,
\end{equation*}
thus $\zeta^{*}\dot{U}^{-1}\zeta<0$, since $U=\bar{\Omega}_i-\Omega_{i,ss}\succ 0$. We can claim the same result for any eigenvalue of $A_i-G_i\Omega_{i,ss}$. Since these eigenvectors together span the space $\mathbb{R}^{L_i}$, we have that $\dot{U}^{-1}$ is negative definite and, consequently, $\dot{\bar{\Omega}}_{i} = \dot{U}$ is also negative definite.

\subsection{Proof of Proposition \ref{prop:derivative_ton_toff}}
\label{app:proof_props_ton_toff}
To prove these propositions, we first define some notations that we omitted in Section \ref{subsec:simple_notation}. Recall that the goal is to optimize the visiting sequence $\Xi$ and the corresponding dwelling sequence $\mathcal{T}$. We now discuss the conversion of indices from the agent's perspective (i.e., $\Xi$ and $\mathcal{T}$) to the target's perspective. First, for each target $i$, we group all the instances where this target was visited and define the vector $\mathcal{P}_i=[p_i^1,...,p_i^{N_i}]$ where $p_i^j$ is the position in the visiting sequence $\Xi$ that $i$ is visited for the $j$\textsuperscript{th} time. Consider, for example, a visiting sequence $\Xi = [1,2,1]$. Then, $N_1=2$, $N_2=1$ and $\mathcal{P}_1=[p_1^1,p_1^2]=[1,3]$, $\mathcal{P}_2=[p_2^1]=[2]$.

Moreover, we define the tuple $(a(q),b(q))$ as the pair such that $p_{a(q)}^{b(q)}=q$. Hence, $a(q)$ is the target being visited at the agent's $q$\textsuperscript{th} visit and $b(q)$ represents the number of times this target has been visited so far (including the current visit). 

Finally, we highlight some important timings and covariance matrix values at the steady state covariance profile. Figure \ref{fig:multiple_observations_same_target} illustrates all the definitions that we will give. We start with $t_{\text{on},i}^k$ and $t_{\text{off},i}^k$ (defined in Section \ref{subsec:simple_notation}) which are given by \begin{subequations}
    \label{eq:targets_perspective_variables}
    \begin{equation}
        t_{\text{on},i}^k = t_{p_i^k},
    \end{equation}
    \begin{equation}
        \label{eq:def_toff}
        t_{\text{off},i}^k= d_{i,a(p_i^k+1)}+\sum_{q=p_i^k+1}^{p_i^{k+1}-1}\left(t_{\text{on},a(q)}^{b(q)}+d_{a(q),a(q+1)}\right).
    \end{equation}
\end{subequations}
In \eqref{eq:def_toff}, $d_{i,a(c(k,i)+1)}$ is the travel time between target $i$ and the next target the agent visits. The index $q$ varies over all the visits the agent makes until it returns to target $i$ (note that $c(k,i)$ and $c(k+1,i)$ give the index of two consecutive visits to target $i$, from the agent's perspective). Moreover,  $t_{\text{on},a(q)}^{b(q)}$ is the time the agent spent at its $q$-th visit and $d_{a(q),a(q+1)}$ is the travel time between the agent's $q$\textsuperscript{th} and $(q+1)$\textsuperscript{th} visit.

Additionally, we define $\overline{\tau}_i^k$ as the instant when the $k$\textsuperscript{th} visit to target $i$ started and $\overline{P}_i^k$ is the covariance at the beginning of this visit. Intuitively, variables with a bar over them refer to a local maximum peak ($\overline{P}_i^k$) and instant ($\overline{\tau}_i^k$). Similarly, $\underline{\tau}_i^k$ and $\underline{P}_i^k$ are respectively the time instant and the covariance at the end of the $k$\textsuperscript{th} visit and represent locally minimum peaks, as represented in Fig. \ref{fig:multiple_observations_same_target}. More formally:
\begin{subequations}
    \label{eq:targets_perspective_variables_2}
    \begin{equation}
        \overline{\tau}_i^k = \sum_{m=0}^{k-1}\left(t_{\text{on},i}^m+t_{\text{off},i}^m\right), \qquad \underline{\tau}_i^k =\overline{\tau}_i^k +t_{\text{on},i}^k,
    \end{equation}
        \begin{equation}
        \overline{P}_i^k = \bar{\Omega}(\overline{\tau}_i^k), \qquad\underline{P}_i^k = \bar{\Omega}(\underline{\tau}_i^k).
        \end{equation}
\end{subequations}
Also, since both the agent trajectories and the steady state covariance are periodic, we have that $t_{\text{on},i}^k=t_{\text{on},i}^{k+N_i}$, $t_{\text{off},i}^k=t_{\text{on},i}^{k+N_i}$, $\overline{P}_i^k=\overline{P}_i^{k+N_i}$ and $\underline{P}_i^k=\underline{P}_i^{k+N_i}$. On the other hand, visiting instants are not periodic, but are spaced by $T$, hence $\overline{\tau}_i^{k+N_i}=T+\overline{\tau}_i^{k}$ and $\underline{\tau}_i^{k+N_i}=T+\underline{\tau}_i^{k}$.

\begin{proof} \lbrack of Proposition \ref{prop:derivative_ton_toff}] 
First, we prove that $\frac{\partial \overline{P}_i^k}{\partial t_{\text{on},i}^m} \prec 0$. Towards that, we define $\Phi_i^m$ as
\begin{equation}
    \Phi_i^m = \int_{0}^{t_{\text{on},i}^m}\exp\left(A_i-
   \bar{\Omega}_i(t+\overline{\tau}_i^m)G_i\right) dt.
\end{equation}
To understand the intuition behind $\Phi_i^m$, we first recall Proposition 3 in \cite{pinto2020periodicfull} which states that the derivative $\frac{\partial {\bar{\Omega}}}{\partial t_{\text{on},i}^m}(t)$, when it exists, for $t\in(\overline{\tau}_i^m,\underline{\tau}_i^m)$ is given by:
\begin{equation}
    \frac{\partial {\bar{\Omega}}}{\partial t_{\text{on},i}^m}(t) = \Sigma^T(t)\frac{\partial {\bar{\Omega}}}{\partial t_{\text{on},i}^m}(\overline{\tau}_i^m)\Sigma(t),
\end{equation}
where $\Sigma$ is the solution of the following ODE:
\begin{equation}
    \dot{\Sigma}(t) - (A-\bar{\Omega}_i(t)G_i)\Sigma(t) = 0,\ \Sigma(\underline{\tau}_i^m)=I.
\end{equation}
Moreover, $\Phi_i^m=\Sigma(\overline{\tau}_i^m)$ can be interpreted as the transition matrix between times $\overline{\tau}_i^m$ and  $\underline{\tau}_i^m$ of the homogeneous version of the ODE for which the derivative $\frac{\partial {\bar{\Omega}}}{\partial t_{\text{on},i}^m}(t)$ is a solution. As a side note, the existence of the derivative is discussed in Appendix A of \cite{pinto2020periodicfull}, and more details on the interpretation of $\Phi_i^m$ can be found on the same paper.

Now, by computing the derivative of $\bar{\Omega}_i(t)$ with respect to $t_{\text{on},i}^m$ as $t\rightarrow(\overline{\tau}_i^k)^-$ (we recall that that at time $t=\overline{\tau}_i^k$ the derivative is discontinuous, since at this instant the target switches from not being observed to being observed), we obtain the following recursive expression:
\begin{equation}
\frac{\partial {\bar{\Omega}}}{\partial t_{\text{on},i}^m}(\underline{\tau}_i^m) = \frac{\partial \underline{P}_i^m}{\partial t_{\text{on},i}^m} = (\Phi_i^m)^T\frac{\partial \overline{P}_i^{m-1}}{\partial t_{\text{on},i}^m}\Phi_i^m + \dot{\bar{\Omega}}_i((\underline{\tau}_i^{m})^-).
\end{equation}

Furthermore, defining $\Psi_i^m=\exp(A_i t_{\text{off},i}^m)$ and using again the result in Proposition 3 of \cite{pinto2020periodicfull}, we get that:
\begin{equation}
    \label{eq:propagation_lower_upper_derivative}
    \frac{\partial \overline{P}_i^{m+1}}{\partial t_{on,i}^m} = (\Psi_i^m)^T\frac{\partial \underline{P}_i^{m+1}}{\partial t_{on,i}^m}\Psi_i^m.
\end{equation}
 
Now, repeating the same steps and propagating the previous expression to the $k$-th visit, $m\leq k \leq m+N_i-1$, we get the recursion:
\begin{multline}
    \frac{\partial \overline{P}_i^{k}}{\partial t_{\text{on},i}^m} =\\ (\Lambda_i^{k,m})^T\left(\frac{\partial \overline{P}_i^{m-1}}{\partial t_{\text{on},i}^m}+(\Phi_i^m)^{-T}\dot{\bar{\Omega}}_i((\underline{\tau}_i^{m})^-)(\Phi_i^m)^{-1}\right)\Lambda_i^{k,m},
\end{multline}
where $\Lambda_i^{k,m}=\prod_{\alpha=m}^{k-1}\Psi_i^\alpha\Phi_i^\alpha$. In particular, for $k=m+N_i-1$, due to periodicity, we have:

\begin{multline}
    \frac{\partial \overline{P}_i^{m-1}}{\partial t_{\text{on},i}^{m}} =  
    (\Lambda_i^{m+N_i-1,m})^T\left(\frac{\partial \overline{P}_i^{m-1}}{\partial t_{\text{on},i}^m}\right.\\\left.+(\Phi_i^m)^{-T}\dot{\bar{\Omega}}_i((\underline{\tau}_i^{m})^-)(\Phi_i^m)^{-1}\right)\Lambda_i^{m+N_i-1,m},
\end{multline}
which is a Lyapunov equation. Note that $\Lambda_i^{m+N_i-1,m}$ is stable, as discussed in Proposition 3 of \cite{pinto2020periodicfull}. Therefore all of its eigenvalues have modulus lower than one. Also, since $\Lambda_i^{m+N_i-1,m}$ is a product of matrix exponentials, its null space is trivial. This implies that, since Lemma \ref{lemma:pd_nd_time_derivative} tells us that $\dot{\bar{\Omega}}_i((\underline{\tau}_i^{m})^-) \prec 0$, the Lyapunov equation has a unique negative definite solution and therefore $\frac{\partial \overline{P}_i^{m-1}}{\partial t_{\text{on},i}^{m}}\prec 0$.

Moreover, note that, for $m \leq k < m+N_i-1$,
\begin{equation}
    \frac{\partial \overline{P}_i^{m-1}}{\partial t_{on,i}^m} = (\Lambda_i^{m+N_i-1,k})^T\frac{\partial \overline{P}_i^{k}}{\partial t_{on,i}^m}\Lambda_i^{m+N_i-1,k},
\end{equation}
which leads us to conclude that, $\forall k$, $\frac{\partial \overline{P}_i^{k}}{\partial t_{on,i}^m}\prec 0.$

The argument to claim that $\frac{\partial \overline{P}_i^k}{\partial t_{\text{off},i}^m} \succ 0$ is very similar to the one we just use to show that $\frac{\partial \overline{P}_i^k}{\partial t_{\text{on},i}^m} \prec 0$. Therefore, only a brief summary will be given.
 Note that 
\begin{equation}
 \frac{\partial \overline{P}_i^m}{\partial t_{\text{off},i}^m} = (\Lambda_i^{m,m-1})^T\frac{\partial \overline{P}_i^{m-1}}{\partial t_{\text{off},i}^m}\Lambda_i^{m,m-1} + \dot{\bar{\Omega}}_i((\overline{\tau}_i^{m})^-).
\end{equation}
Using a similar recursion as in the previous proof, we get that
\begin{multline}
 \frac{\partial \overline{P}_i^{m-1}}{\partial t_{\text{off},i}^m} = (\Lambda_i^{m+N_i,m-1})^T\frac{\partial \overline{P}_i^{m-1}}{\partial t_{\text{off},i}^m}\Lambda_i^{m+N_i,m-1} \\+ (\Lambda_i^{m+N_i,m})^T\dot{\bar{\Omega}}_i((\overline{\tau}_i^{m})^-)\Lambda_i^{m+N_i,m}.
\end{multline}
As $\dot{\bar{\Omega}}_i((\overline{\tau}_i^m)^-) \succ 0$, then $\frac{\partial \underline{P}_i^{m-1}}{\partial t_{\text{off},i}^m} \succ 0$ and thus $\frac{\partial \overline{P}_i^{k}}{\partial t_{\text{off},i}^m} \succ 0$.
\end{proof}

\subsection{Proof the unimodaility of $g_{\text{con}}$ with respect to $T$}
\label{app:unimodality_proof}

In this appendix, we show that when the internal target state $\phi_i$ is a scalar (i.e. $L_i=1$), the optimal peak uncertainty is a unimodal function of the cycle period $T$. First, we define a function $\beta_i(\rho,T)$ that returns the dwell-time $t_{\text{on},i}$. Note that, by definition, $t_{\text{off},i}=T-\beta_i(\rho,T)$, and it defines a peak uncertainty level as   $g_i(||\overline{P}_i(\beta_i(\rho,T),T-\beta_i(\rho,T))||)=\rho$.

First, we recall that the function $\beta(\rho,T)$ is well defined for any $\rho$ such that $g_i(||\Omega_{i,ss}||)\leq \rho \leq g_i(||\Omega_i^\infty||)$ and $T>0$.
This is due to the fact that $\rho=g_i(||\Omega_{i,ss}||)$ will give $t_{\text{on},i}=\beta(\rho,T)=T$ and $\rho= g_i(||\Omega_i^\infty||)$ will give $t_{\text{on},i}=\beta(\rho,T)=0$. Since the peak uncertainty of target $i$ varies continuously with $t_{\text{on},i}$, we have that there will always be a $t_{\text{on},i}$ that yields a given value of $\rho$, if $g_i(||\Omega_{i,ss}||)\leq \rho \leq g_i(||\Omega_i^\infty||)$ and $T>0$.

Now, we make the following assumption on the smoothness of the function $\beta(\rho,T)$, that will be used in the proof of the following lemma.

\begin{assumption}
    If $\Omega_{i,ss}<\rho<\Omega_i^{\infty}$ function $\beta(\rho,T)$ is differentiable with respect to T. 
\end{assumption}

\begin{lemma}
    \label{lemma:strictly_increasing_derivative_beta}
    When the state $\phi_i$ is a scalar, the function $\frac{\partial \beta_i(\rho,T)}{\partial T}$ is strictly positive and increasing.
\end{lemma}
\begin{proof}
    When computing $\frac{\partial \beta(\rho,T)}{\partial T}$, we leave the upper peak $\overline{P}_{i,k}(t_{\text{on},i},t_{\text{off},i})$ constant and vary the period $T$. However, the under peak $\underline{P}_{i,k}(t_{\text{on},i},t_{\text{off},i})$ does not remain constant. Indeed, the variation of the under peak can be computed as:
    \begin{equation*}
        \frac{\partial \underline{P}_{i,k}}{\partial T}=\dot{\bar{\Omega}}_i(\underline{\tau}_i^-)\frac{\partial \beta_i(\rho,T)}{\partial T} = -\dot{\bar{\Omega}}_i(\underline{\tau}_i^+)\left(1-\frac{\partial \beta_i(\rho,T)}{\partial T}\right)
    \end{equation*}
    Therefore, $\frac{\partial \beta_i(\rho,T)}{\partial T} = \frac{2A_i\underline{P}_{i}+Q_i}{\underline{P}_{i}^2G_i}>0$.
    Note that since $\underline{P}_i>\Omega_{i,ss}$, $0<\frac{\partial \beta_i(\rho,T)}{\partial T}<1$. Therefore, $t_{\text{off},i}$ also increases as $T$ increases. This implies that $\underline{P}_i$ decreases with the increase of $T$, since $\underline{P}_i=\exp(-2A_i t_{\text{off},i})(\rho+Q_i)-Q_i$.
    
    Computing the variation of $\frac{\partial \beta_i(\rho,T)}{\partial T}$ with respect to $P_i$, we get that
    $
        \frac{\partial}{\partial P_i}\left(\frac{\partial \beta_i(\rho,T)}{\partial T}\right) =-\frac{2A_iP_i+Q_i}{P_i^3G_i},
    $
    which is negative for any positive value of $P_i$. 
    Since $\frac{\partial \beta_i(\rho,T)}{\partial T}$ is strictly decreasing with $P_i$ and $P_i$ is strictly decreasing with respect to $T$, we get that $\frac{\partial \beta_i(\rho,T)}{\partial T}$ is strictly increasing with the increase of $T$.
\end{proof}

\begin{lemma}
    \label{lemma:beta_max_two_roots}
    Given a visiting sequence $\Xi$, where each target is only visited once, the equation $\sum_{i\in\Xi}\beta_i(\rho,T)=T-t_{\text{travel}}$ for a fixed $\rho$ has at most two solutions.
\end{lemma}
\begin{proof}
Denoting $\Gamma(T) = \sum_{i\in\Xi}\beta_i(\rho,T)-T+t_{\text{travel}}$ and using Lemma \ref{lemma:strictly_increasing_derivative_beta}, we know that $\frac{\partial\Gamma}{\partial T}$ is a strictly increasing function. Now suppose that $\Gamma(T)$ has 3 or more distinct roots and we pick three of them $T_1<T_2<T_3$. The mean value theorem tells us that there is $\theta_1\in(T_1,T_2)$ such that $\frac{\partial\Gamma}{\partial T}(\theta_1)=0$ and 
$\theta_2\in(T_2,T_3)$ such that $\frac{\partial\Gamma}{\partial T}(\theta_2)=0$. This is a contradiction since $\theta_2>\theta_1$ and $\frac{\partial\Gamma}{\partial T}$ is strictly increasing.
\end{proof}

\begin{proposition}
\label{prop:unimodality}
When the state $\phi_i$ is a scalar, for a given visiting sequence $\Xi$, the optimal peak is a unimodal function of $T$.
\end{proposition}
\begin{proof}
We demonstrate this proposition by contradiction, supposing that there are at least two extremum points when considering the optimal peak as a function of $T$ and showing that this contradicts Lemma \ref{lemma:beta_max_two_roots}.

\begin{figure}
    \centering
    \includegraphics[width=0.7\columnwidth]{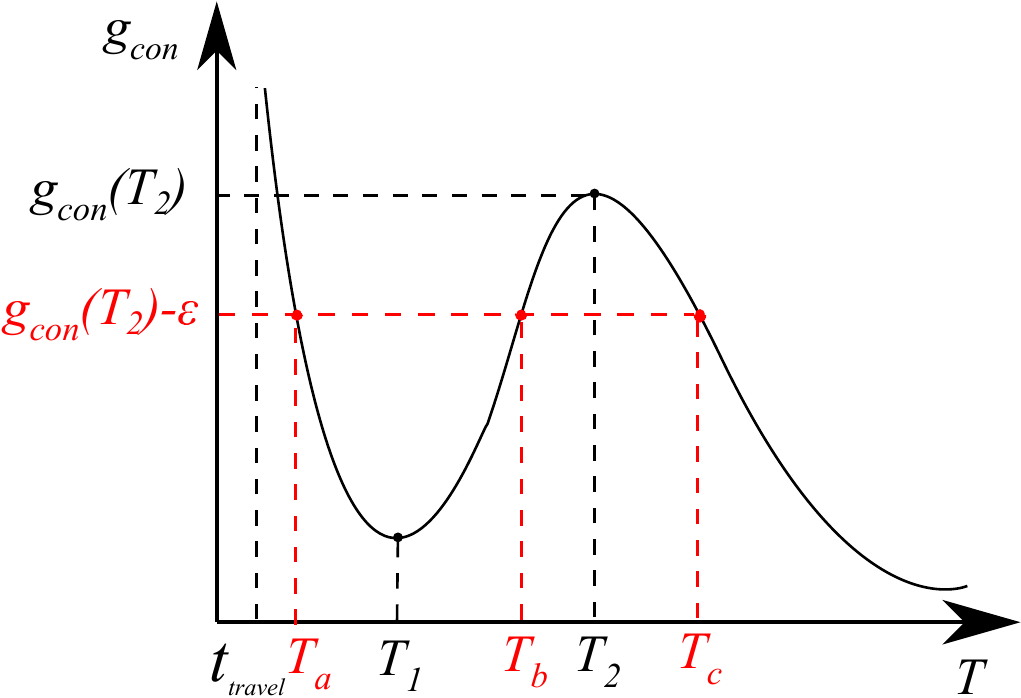}
    \caption{Illustration of the proof of Proposition \ref{prop:unimodality}.}
    \label{fig:proof_unmodality}
    \vspace{-3mm}
\end{figure}

We now show that, if there are two extremum points, then there is at least one peak value that could be generated by three different values of $T$.
First, note that when $T\rightarrow t_{travel}$ the upper peak tends to $\max_ig_i\left(\norm{\Omega_i^{\infty}}\right)$, since in that case no target is observed. In any non-degenerate dwell-times distribution, we have $P_i \prec \Omega_i^{\infty}$. Moreover, since there are at least two extremum points, there must exist a minimum (at $T=T_1$) and a maximum (at $T=T_2$). Note that since the consensus peak for $T=t_{travel}$ is a maximum, we have $t_{travel}<T_1<T_2$.

Let us denote $g_{con}(T)$ as the consensus peak for a given period $T$ and take $\epsilon>0$ (which should be interpreted as a small positive number, but is not necessarily small), such that $g_{con}(T_2)-\epsilon>g(T_1)$ and
$g_{con}(T_2)-\epsilon = g_{con}(T_b)=g_{con}(T_c)$, for some $T_b,\ T_c$ such that $T_1<T_b<T_2$ and $T_c>T_2$. Note that, since $T_2$ is a point of maximum, such $T_b$, $T_c$ and $\epsilon$ exist. Additionally, we note that there exists a $T_a\in[t_{travel},T_1]$ such that $g_{con}(T_a) = g_{con}(T_2)-\epsilon$, since $g_{con}(T)$ is assumed continuous and $g_{con}(t_{travel})>g_{con}(T_2)-\epsilon>g_{con}(T_1)$. These times, along with their respective $g_{con}$ are illustrated in Fig. \ref{fig:proof_unmodality}.

Therefore, if there are two extremum points, then $g_{con}(T_a)=g_{con}(T_b)=g_{con}(T_c)$, with $T_a\neq T_b \neq T_c$, which contradicts Lemma \ref{lemma:beta_max_two_roots}.
\end{proof}

\bibliographystyle{IEEEtran}
\bibliography{references.bib}

\begin{thebibliography}{10}
\providecommand{\url}[1]{#1}
\csname url@samestyle\endcsname
\providecommand{\newblock}{\relax}
\providecommand{\bibinfo}[2]{#2}
\providecommand{\BIBentrySTDinterwordspacing}{\spaceskip=0pt\relax}
\providecommand{\BIBentryALTinterwordstretchfactor}{4}
\providecommand{\BIBentryALTinterwordspacing}{\spaceskip=\fontdimen2\font plus
\BIBentryALTinterwordstretchfactor\fontdimen3\font minus
  \fontdimen4\font\relax}
\providecommand{\BIBforeignlanguage}[2]{{%
\expandafter\ifx\csname l@#1\endcsname\relax
\typeout{** WARNING: IEEEtran.bst: No hyphenation pattern has been}%
\typeout{** loaded for the language `#1'. Using the pattern for}%
\typeout{** the default language instead.}%
\else
\language=\csname l@#1\endcsname
\fi
#2}}
\providecommand{\BIBdecl}{\relax}
\BIBdecl

\bibitem{lin2018kalman}
Z.~Lin, H.~H. Liu, and M.~Wotton, ``{Kalman Filter-based Large-scale Wildfire
  Monitoring with a System of UAVs},'' \emph{IEEE Transactions on Industrial
  Electronics}, vol.~66, no.~1, pp. 606--615, 2018.

\bibitem{ostertag2019robust}
M.~Ostertag, N.~Atanasov, and T.~Rosing, ``{Robust Velocity Control for Minimum
  Steady State Uncertainty in Persistent Monitoring Applications},'' in
  \emph{American Control Conference (ACC)}.\hskip 1em plus 0.5em minus
  0.4em\relax IEEE, 2019, pp. 2501--2508.

\bibitem{lan2016rapidly}
X.~Lan and M.~Schwager, ``Rapidly exploring random cycles: Persistent
  estimation of spatiotemporal fields with multiple sensing robots,''
  \emph{IEEE Transactions on Robotics}, vol.~32, no.~5, pp. 1230--1244, 2016.

\bibitem{Alam:2018ie}
T.~Alam, G.~M. Reis, L.~Bobadilla, and R.~N. Smith, ``{A Data-Driven Deployment
  Approach for Persistent Monitoring in Aquatic Environments},'' in \emph{IEEE
  International Conference on Robotic Computing}, 2018, pp. 147--154.

\bibitem{pinto2020multipleparticletracking}
S.~C. Pinto, N.~A. Vickers, F.~Sharifi, and S.~B. Andersson, ``{Tracking
  Multiple Diffusing Particles Using Information Optimal Control},'' in
  \emph{American Control Conference (under review)}, 2021.

\bibitem{grocholsky2003information}
B.~Grocholsky, A.~Makarenko, and H.~Durrant-Whyte, ``{Information-theoretic
  Coordinated Control of Multiple Sensor Platforms},'' in \emph{2003 IEEE
  International Conference on Robotics and Automation (Cat. No. 03CH37422)},
  vol.~1.\hskip 1em plus 0.5em minus 0.4em\relax IEEE, 2003, pp. 1521--1526.

\bibitem{julian2012distributed}
B.~J. Julian, M.~Angermann, M.~Schwager, and D.~Rus, ``{Distributed Robotic
  Sensor Networks: An Information-theoretic Approach},'' \emph{The
  International Journal of Robotics Research}, vol.~31, no.~10, pp. 1134--1154,
  2012.

\bibitem{lan2014variational}
X.~Lan and M.~Schwager, ``{A Variational Approach to Trajectory Planning for
  Persistent Monitoring of Spatiotemporal Fields},'' in \emph{2014 American
  Control Conference}, 2014, pp. 5627--5632.

\bibitem{hussein2008kalman}
I.~I. Hussein, ``{Kalman Filtering with Optimal Sensor Motion Planning},'' in
  \emph{2008 American Control Conference}.\hskip 1em plus 0.5em minus
  0.4em\relax IEEE, 2008, pp. 3548--3553.

\bibitem{Lavalle1998}
S.~M. Lavalle, ``{Rapidly-Exploring Random Trees: A New Tool for Path
  Planning},'' Tech. Rep., 1998.

\bibitem{pinto2020sdppm}
S.~C. Pinto, S.~B. Andersson, J.~M. Hendrickx, and C.~G. Cassandras, ``{A
  Semidefinite Programming Approach to Discrete-time Infinite Horizon
  Persistent Monitoring},'' in \emph{European Control Conference (under
  review)}, 2021.

\bibitem{chen2020multi}
J.~Chen, A.~Baskaran, Z.~Zhang, and P.~Tokekar, ``Multi-agent reinforcement
  learning for persistent monitoring,'' \emph{arXiv preprint arXiv:2011.01129},
  2020.

\bibitem{cassandras2013optimal}
C.~G. Cassandras, X.~Lin, and X.~Ding, ``{An Optimal Control Approach to the
  Multi-agent Persistent Monitoring Problem},'' \emph{IEEE Transactions on
  Automatic Control}, vol.~58, no.~4, pp. 947--961, 2013.

\bibitem{cassandras2010perturbation}
C.~G. Cassandras, Y.~Wardi, C.~G. Panayiotou, and C.~Yao, ``{Perturbation
  Analysis and Optimization of Stochastic Hybrid Systems},'' \emph{European
  Journal of Control}, vol.~16, no.~6, pp. 642--661, 2010.

\bibitem{pinto2020periodicfull}
S.~C. Pinto, S.~B. Andersson, J.~M. Hendrickx, and C.~G. Cassandras,
  ``{Multi-Agent Persistent Monitoring of Targets with Uncertain States},'' in
  \emph{Arxiv}.\hskip 1em plus 0.5em minus 0.4em\relax Available online, 2020.

\bibitem{Welikala2020P7}
S.~Welikala and C.~G. Cassandras, ``{Event-Driven Receding Horizon Control for
  Distributed Estimation in Network Systems},'' in \emph{Proc. of American
  Control Conf. (to appear)}, 2021.

\bibitem{Welikala2019P3}
------, ``{Asymptotic Analysis for Greedy Initialization of Threshold-Based
  Distributed Optimization of Persistent Monitoring on Graphs},'' in
  \emph{Proc. of 21st IFAC World Congress}, 2020.

\bibitem{Yu2016}
J.~Yu, M.~Schwager, and D.~Rus, ``{Correlated Orienteering Problem and its
  Application to Persistent Monitoring Tasks},'' \emph{IEEE Trans. on
  Robotics}, vol.~32, no.~5, pp. 1106--1118, 2016.

\bibitem{shakhatreh2019unmanned}
H.~Shakhatreh, A.~H. Sawalmeh, A.~Al-Fuqaha, Z.~Dou, E.~Almaita, I.~Khalil,
  N.~S. Othman, A.~Khreishah, and M.~Guizani, ``{Unmanned Aerial Vehicles
  (UAVs): A Survey on Civil Applications and Key Research Challenges},''
  \emph{Ieee Access}, vol.~7, pp. 48\,572--48\,634, 2019.

\bibitem{Luxburg2007}
\BIBentryALTinterwordspacing
U.~von Luxburg, ``{A Tutorial on Spectral Clustering},'' 2007. [Online].
  Available: \url{http://arxiv.org/abs/0711.0189}
\BIBentrySTDinterwordspacing

\bibitem{zhao2014optimal}
L.~Zhao, W.~Zhang, J.~Hu, A.~Abate, and C.~J. Tomlin, ``{On the Optimal
  Solutions of the Infinite-horizon Linear Sensor Scheduling Problem},''
  \emph{IEEE Transactions on Automatic Control}, vol.~59, no.~10, pp.
  2825--2830, 2014.

\bibitem{pinto2019periodicfull}
S.~C. Pinto, S.~B. Andersson, J.~M. Hendrickx, and C.~G. Cassandras, ``{Optimal
  Periodic Multi-Agent Persistent Monitoring of the Uncertain State of a Finite
  Set of Targets},'' in \emph{Arxiv}.\hskip 1em plus 0.5em minus 0.4em\relax
  Available online, 2019.

\bibitem{FB-LNS}
\BIBentryALTinterwordspacing
F.~Bullo, \emph{Lectures on Network Systems}, 1st~ed.\hskip 1em plus 0.5em
  minus 0.4em\relax Kindle Direct Publishing, 2020, with contributions by J.
  Cortes, F. Dorfler, and S. Martinez. [Online]. Available:
  \url{http://motion.me.ucsb.edu/book-lns}
\BIBentrySTDinterwordspacing

\bibitem{chu2004structure}
E.-W. Chu, H.-Y. Fan, W.-W. Lin, and C.-S. Wang, ``{Structure-preserving
  Algorithms for Periodic Discrete-time Algebraic Riccati Equations},''
  \emph{International Journal of Control}, vol.~77, no.~8, pp. 767--788, 2004.

\bibitem{kiefer1953sequential}
J.~Kiefer, ``{Sequential Minimax Search for a Maximum},'' \emph{Proceedings of
  the American Mathematical Society}, vol.~4, no.~3, pp. 502--506, 1953.

\bibitem{TANG2000267}
L.~Tang, J.~Liu, A.~Rong, and Z.~Yang, ``{A Multiple Traveling Salesman Problem
  Model for Hot Rolling Scheduling in Shanghai Baoshan Iron \& Steel
  Complex},'' \emph{European Journal of Operational Research}, vol. 124, no.~2,
  pp. 267 -- 282, 2000.

\bibitem{Hari2019}
S.~K. Hari, S.~Rathinam, S.~Darbha, K.~Kalyanam, S.~G. Manyam, and D.~Casbeer,
  ``{The Generalized Persistent Monitoring Problem},'' in \emph{Proc. of
  American Control Conf.}, vol. 2019-July, 2019, pp. 2783--2788.

\bibitem{Kirk2020}
\BIBentryALTinterwordspacing
J.~Kirk, ``{Traveling Salesman Problem - Genetic Algorithm},'' 2020. [Online].
  Available:
  \url{https://www.mathworks.com/matlabcentral/fileexchange/13680-traveling-salesman-problem-genetic-algorithm}
\BIBentrySTDinterwordspacing

\bibitem{Welikala2019Ax2}
\BIBentryALTinterwordspacing
S.~Welikala and C.~G. Cassandras, ``{Asymptotic Analysis for Greedy
  Initialization of Threshold-Based Distributed Optimization of Persistent
  Monitoring on Graphs},'' 2019. [Online]. Available:
  \url{http://arxiv.org/abs/1911.02658}
\BIBentrySTDinterwordspacing

\bibitem{ShiMalik2000}
{Jianbo Shi} and J.~Malik, ``{Normalized cuts and image segmentation},''
  \emph{IEEE Trans. on Pattern Analysis and Machine Intelligence}, vol.~22,
  no.~8, pp. 888--905, 2000.

\bibitem{bittanti2012riccati}
S.~Bittanti, A.~J. Laub, and J.~C. Willems, \emph{The Riccati Equation}.\hskip
  1em plus 0.5em minus 0.4em\relax Springer Science \& Business Media, 2012.

\bibitem{dieci1996preserving}
L.~Dieci and T.~Eirola, ``{Preserving Monotonicity in the Numerical Solution of
  Riccati Differential Equations},'' \emph{Numerische Mathematik}, vol.~74,
  no.~1, pp. 35--47, 1996.

\end{thebibliography}
\end{document}